%% file: main.tex
\documentclass[a4paper,12pt]{article}
\usepackage[T1]{fontenc}
\usepackage{amsmath,amssymb,amsthm}
\usepackage[ruled,vlined]{algorithm2e}
\usepackage[mathscr]{euscript}
\usepackage{mathrsfs}
\usepackage{mathtools}
\usepackage{enumitem}
\usepackage{color}
\usepackage{graphicx}
\usepackage{epsfig}
\usepackage{placeins}
\usepackage{longtable}
\usepackage{pdfpages}
\usepackage{float}
\usepackage{subcaption}
\usepackage[font=small]{caption}
\usepackage[hang,flushmargin]{footmisc} 
\usepackage{booktabs}
\usepackage{rotating,tabularx}
\usepackage{tikz}
\usetikzlibrary{positioning}
\usepackage{comment}
\usepackage{natbib}
\usepackage{setspace}
\usepackage{url}
\usepackage[left=3cm,right=3cm,bottom=3cm,top=3cm]{geometry}
\numberwithin{equation}{section}
\renewcommand{\baselinestretch}{1.2}\normalsize

\allowdisplaybreaks

\usepackage{hyperref}
\hypersetup{
hidelinks,
}

\input{macros}

\begin{document}

\heading{High-Dimensional Panel Data Models}{with Interactive Fixed Effects:}{Beyond the Linear Case}

\authors
{Maximilian R\"ucker\renewcommand{\thefootnote}{1}\footnotemark[1]}{Ulm University}
{Michael Vogt\renewcommand{\thefootnote}{2}\footnotemark[2]}{Ulm University} 
{Oliver Linton\renewcommand{\thefootnote}{3}\footnotemark[3]}{University of Cambridge}
\vspace{-0.9cm}

\renewcommand{\thefootnote}{1}
\footnotetext[1]{Corresponding author. Address: Institute of Statistics, Department of Mathematics and Economics, Ulm University, Helmholtzstrasse 20, 89081 Ulm, Germany. \\
Email: \texttt{maximilian.ruecker@uni-ulm.de}.}
\renewcommand{\thefootnote}{2}
\footnotetext[2]{Address: Institute of Statistics, Department of Mathematics and Economics, Ulm University, Helmholtzstrasse 20, 89081 Ulm, Germany. Email: \texttt{m.vogt@uni-ulm.de}.}
\renewcommand{\thefootnote}{3}
\footnotetext[3]{Address: Faculty of Economics, University of Cambridge, Austin Robinson Building, Sidgwick Avenue, Cambridge, CB3 9DD, UK. Email: \texttt{obl20@cam.ac.uk}.}
\renewcommand{\thefootnote}{\arabic{footnote}}
\setcounter{footnote}{3}

\renewcommand{\baselinestretch}{1.0}\normalsize
\renewcommand{\abstractname}{}
\begin{abstract}
\noindent Modern economic panel data sets are often high-dimensional: they contain information on a wide variety of control variables whose number may even exceed the sample size. Nevertheless, the literature on econometric methods for high-dimensional panels is quite limited. In this paper, we study high-dimensional panel models with interactive fixed effects where the regression function has an additive structure, i.e., each covariate enters the model via an unknown nonlinear component function. We develop estimation methodology and theory in this additive framework which substantially extends previous work on the high-dimensional linear case by \cite{Mruecker2025CCE}. Our main theoretical contributions are to derive convergence rates for the proposed estimator in both the small-$T$ and large-$T$ panel settings and to establish restricted eigenvalue (RE) conditions for the projected nonlinear design matrix of our model under primitive assumptions. In particular, we verify these RE conditions for polynomial transformations of Gaussian regressors with a general covariance structure. The theoretical results are complemented by comprehensive Monte Carlo experiments and an empirical application.
\end{abstract}

\renewcommand{\baselinestretch}{1.2}\normalsize



\input{introduction}

\input{methods}

\input{theory}

\input{extensions}

\input{simulations}

\input{empirical_study}

\FloatBarrier
\section*{Acknowledgements}
Financial support by the DFG (German Research Foundation) -- project number 501082519 -- is gratefully acknowledged.

\bibliographystyle{ims}
{\small
\setlength{\bibsep}{0.2em}
\bibliography{bibliography}}

\newpage
\setcounter{page}{1}
\input{appendixA}

\input{appendixB}

\input{appendixC}

\end{document}

%% file: macros.tex
\newcommand{\reals}{\mathbb{R}}

\newcommand{\naturals}{\mathbb{N}}

\newcommand{\pr}{\mathbb{P}} 
\newcommand{\ex}{\mathbb{E}} 
\newcommand{\var}{\textnormal{Var}} 
\newcommand{\cov}{\textnormal{Cov}} 
\newcommand{\tr}{\textnormal{trace}} 
\newcommand{\normal}{N} 
\newcommand{\ind}{1} 
\newcommand{\bs}[1]{\boldsymbol{#1}}
\newcommand{\supp}{\textnormal{supp}}

\newcommand{\eig}{\psi} 
\newcommand{\Eig}{\Psi} 
\newcommand{\pen}{\lambda} 

\DeclareMathOperator*{\argmin}{arg\,min}

\newcommand{\convp}{\stackrel{p}{\longrightarrow}} 

\theoremstyle{plain}

\newtheorem{theorem}{Theorem}[section]

\newtheorem{lemma}[theorem]{Lemma}

\newtheorem{definition}[theorem]{Definition}
\newtheorem{remark}[theorem]{Remark}

\newtheorem{lemmaA}{Lemma}
\newtheorem{propA}[lemmaA]{Proposition}

\DeclareFontFamily{U}{mathx}{}
\DeclareFontShape{U}{mathx}{m}{n}{<-> mathx10}{}
\DeclareSymbolFont{mathx}{U}{mathx}{m}{n}
\DeclareMathAccent{\widecheck}{0}{mathx}{"71}

\makeatletter
\newcommand{\lefteqno}{\let\veqno\@@leqno}
\makeatother

\makeatletter
\newsavebox\myboxA
\newsavebox\myboxB
\newlength\mylenA

\newcommand*\xoverline[2][0.75]{%
    \sbox{\myboxA}{$\m@th#2$}%
    \setbox\myboxB\null
    \ht\myboxB=\ht\myboxA%
    \dp\myboxB=\dp\myboxA%
    \wd\myboxB=#1\wd\myboxA
    \sbox\myboxB{$\m@th\overline{\copy\myboxB}$}
    \setlength\mylenA{\the\wd\myboxA}
    \addtolength\mylenA{-\the\wd\myboxB}%
    \ifdim\wd\myboxB<\wd\myboxA%
       \rlap{\hskip 0.5\mylenA\usebox\myboxB}{\usebox\myboxA}%
    \else
        \hskip -0.5\mylenA\rlap{\usebox\myboxA}{\hskip 0.5\mylenA\usebox\myboxB}%
    \fi}
\makeatother

\newcommand{\heading}[3]
{  \setcounter{page}{1}
   \begin{center}

   {\LARGE \textbf{#1}}\\[0.3cm]
   {\LARGE \textbf{#2}}
   \\[0.35cm]
   {\LARGE \textbf{#3}}
    \end{center}
}

\newcommand{\authors}[6]
{  \begin{center}
      \begin{minipage}[c][2cm][c]{4.4cm}
      \begin{center} 
      {\large #1} 
      \vspace{0.05cm}
      
      #2 
      \end{center}
      \end{minipage}
      \begin{minipage}[c][2cm][c]{4cm}
      \begin{center} 
      {\large #3}
      \vspace{0.05cm}
      
      #4 
      \end{center}
      \end{minipage}
      \begin{minipage}[c][2cm][c]{4.4cm}
      \begin{center} 
      {\large #5}
      \vspace{0.05cm}

      #6 
      \end{center}
      \end{minipage}
   \end{center}
}


%% file: introduction.tex
\section{Introduction}

Panel data models are widely used in applied econometrics. Their popularity stems from their ability to account for unobserved heterogeneity in a flexible and tractable manner. A very popular class of panel models are linear interactive fixed effects models of the form 
\begin{equation}\label{eq:model-linear-intro} 
Y_{it} = \sum_{j=1}^p X_{it,j}\beta_j + \gamma_i^\top F_t + \varepsilon_{it}, 
\end{equation}
  where $\{ (Y_{it}, X_{it,1},\ldots,X_{it,p})^\top: 1 \le t \le T, \, 1 \le i \le n \}$ is the observed panel data set and $\beta = (\beta_1,\ldots,\beta_p)^\top$ is the parameter vector to be estimated. The error structure of the model has two components: (i) a standard idiosyncratic error $\varepsilon_{it}$ with $\ex[\varepsilon_{it}] = 0$ and (ii) the interactive fixed effects component $\gamma_i^\top F_t$ with $F_t$ denoting a vector of unobserved common factors and $\gamma_i$ being unit-specific factor loadings. Importantly, the interactive fixed effects $\gamma_i^\top F_t$ are allowed to be correlated with the regressors $X_{it,j}$, thus inducing endogeneity in the model. The seminal work of \cite{Pesaran2006} introduced the so-called Common Correlated Effects (CCE) approach to estimate the unknown parameter vector $\beta$ in a low-dimensional version of model \eqref{eq:model-linear-intro}, where the number of parameters $p$ is small compared to the dimensions $n$ and $T$. Recently, \cite{Mruecker2025CCE} have developed an extension of the CCE approach to high dimensions, where $p$ is allowed to be large relative to $n$ and $T$, potentially much larger than the sample size $nT$ itself.


A natural generalization of the linear model \eqref{eq:model-linear-intro} is an additive model of the form 
\begin{equation}\label{eq:model-additive-intro}
Y_{it} = \sum_{j=1}^p m_j(X_{it,j}) + \gamma_i^\top F_t + \varepsilon_{it},
\end{equation}
where each covariate $X_{it,j}$ enters the model via an unknown nonparametric link function $m_j$. Typically, the functions $m_j$ are assumed to belong to some smoothness class and are approximated via a basis expansion $m_j(x) = \phi_j(x)^\top \beta_j + \delta_j(x)$, where $\phi_j = (\phi_{j1},\ldots,\phi_{jL_j})^\top$ is a vector of (known) basis functions such as polynomials or splines, $\beta_j = (\beta_{j1},\ldots,\beta_{jL_j})^\top$ is a vector of (unknown) parameters and $\delta_j$ captures the approximation error of the basis expansion. 
An important special case, which is highly relevant for econometric practice, arises from neglecting the approximation error (i.e., $\delta_j \equiv 0$) and assuming $L_j$ to be a fixed natural number that does not grow with the sample size $nT$. This leads to a parametric additive model where each component function $m_j$ can be represented exactly as a linear combination of finitely many functions $\phi_{j1},\ldots,\phi_{jL_j}$. A leading example, which is ubiquitous in applied econometrics, is a polynomial setting where $\phi_{j\ell}$ are the first $L_j$ monomials $\phi_{j\ell}(x) = x^{\ell}$ for $\ell =1,\ldots,L_j$; see e.g.\ \cite{HAUSMAN1991273} who treat a very classical setup regarding measurement errors. By setting, for instance, $L_j = 3$ in this polynomial setting, we incorporate the $j$-th regressor $X_{it,j}$ itself in the model as well as the quadratic and cubic effects $X_{it,j}^2$ and $X_{it,j}^3$. Quadratic terms arise naturally in models with convex adjustment costs, diminishing returns in production or risk-averse utility, where marginal effects vary smoothly with the state variable. Cubic terms allow for asymmetric responses and inflection points, such as expansionary versus contractionary dynamics in investment or labor supply. The polynomial degree is usually fixed ex ante in practice, often to capture certain features -- such as curvature or local nonlinearities -- motivated by economic theory.

The main aim of this paper is to develop estimation methods for the additive panel model \eqref{eq:model-additive-intro}, which work in both low and high dimensions, and to back them up by theory. We focus on the parametric additive model, leaving the nonparametric additive setting for future research. We demonstrate that the methodology proposed by \cite{Mruecker2025CCE} for the linear case can, with mild adjustments, be extended to the parametric (as well as the nonparametric) additive case, thereby considerably broadening its scope. 
Extending the underlying theory, in contrast, is far from trivial. The main issue is this: in the CCE approach, the covariates $X_{it} = (X_{it,1},\ldots,X_{it,p})^\top$ are modelled linearly in the factors $F_t$. More specifically, 
\begin{equation}\label{eq:model-regressors-intro}
X_{it} = \bs{\Gamma}_i F_t + Z_{it}, 
\end{equation}
where $Z_{it}$ is an idiosyncratic component and $\bs{\Gamma}_i F_t$ is a factor component with $\bs{\Gamma}_i$ denoting a matrix of factor loadings. The core idea of the CCE approach -- no matter whether the original low-dimensional or the high-dimensional version is considered -- is to (approximately) eliminate the factors $F_t$ from the model by a suitable data-driven projection prior to estimation. In the linear model \eqref{eq:model-linear-intro}, both the factor component $\bs{\Gamma}_i F_t$ in the regressors and the factor component $\gamma_i^\top F_t$ in the error structure of the model can be eliminated in this way. In the additive model \eqref{eq:model-additive-intro}, in contrast, the regressors enter the model nonlinearly (via the link functions $m_j$), implying that a projection will no longer eliminate the factors from the regressors. As discussed in detail in Section \ref{sec:heuristics}, this produces highly non-trivial technical challenges that are entirely absent in the linear case treated in \cite{Mruecker2025CCE}. The main theoretical contribution of the paper is to overcome these challenges.

More specifically, the central theoretical challenge is to show that the projected design matrix in our additive model is sufficiently regular to support consistent estimation of $\beta$. In particular, we need to establish a suitable restricted eigenvalue (RE) condition. In the linear case, the projection approximately eliminates the common factors from the design matrix. In the additive case, in contrast, the projected design matrix continues to depend on the factors and their loadings through the nonlinear transformations. Its stochastic structure is therefore substantially more complex, and the arguments developed in \cite{Mruecker2025CCE} cannot be applied directly. We develop new techniques that account for this remaining factor dependence to verify the required RE condition. 
Our results are related to those of \cite{scheidegger2023spectral}, who establish RE conditions for additive models with hidden confounders in a cross-sectional setting. However, their approach relies on a transformation of the design matrix whose stochastic structure differs fundamentally from the projection used in our CCE framework. In particular, the decomposition techniques employed in their analysis are not available for the estimated factor projection considered here. We develop arguments that account directly for the interaction between the nonlinear transformations and the factor projection, drawing on Gaussian maximum correlation inequalities established by \cite{GUO20221037}. These arguments provide a theoretical justification for high-dimensional CCE estimation beyond the linear model, both when the time series length is fixed and when it increases with the cross-sectional sample size.

\subsubsection*{Literature review}

Comprehensive surveys of interactive fixed effects models are provided by \cite{Bonhomme2024} and \cite{ditzen2025interactivegroupednonseparablefixed}. 
While there is a large amount of literature on interactive fixed effects models in the classical low-dimen\-sio\-nal case, the literature on the high-dimensional case is very limited: \cite{LuSu2016} extend the least squares method of \cite{Bai2009} to a high-dimensional dynamic panel model by adding a group-lasso penalty. However, they only consider a situation where the dimensionality of the model grows fairly slowly with the sample size. Hence, their setting is not really high-dimensional. \cite{BelloniChenPadillaWang2019} and \cite{gao2025robustestimationinferencehighdimensional} develop nuclear norm penalized estimation methods for high-dimensional panel settings with interactive fixed effects. Finally, as already discussed above, \cite{Mruecker2025CCE} devise an extension of the CCE approach to high dimensions. 

Besides interactive fixed effects models, models with other random and fixed effects structures have been studied in the high-dimensional regime. We briefly discuss some examples: \cite{Kock2013} derives theory for bridge estimators in both random and fixed effects models, while \cite{Kock2016} analyzes a model with a hybrid error structure that is in-between random and fixed effects. \cite{Belloni01102016GUNCONTROL} investigate high-dimensional two-way fixed effects models 
and propose a cluster-lasso estimator that accounts for within-cluster dependence. More recently, \cite{poselliClarke2025} have developed a double machine learning framework for partially linear panel models with fixed effects, allowing for high-dimensional and potentially nonlinear nuisance functions. \cite{chernozhukov2026arellanobondlassoestimatordynamic} study high-dimensional dynamic linear panel models with fixed effects and develop an Arellano-Bond lasso estimator to handle the many-instruments problem. In a similar vein, \cite{semenova2023} focus on heterogeneous treatment effects in dynamic panels under weak dependence, so that high dimensionality enters through treatment-effect heterogeneity and nuisance estimation. 

\subsubsection*{Structure of the paper}

Section \ref{sec:model} gives a detailed description of the parametric additive model that we analyze in the paper. The estimation stra\-tegy is described in Section \ref{sec:estimation-method}, while the theory is presented in Section \ref{sec:theoreticalresults}. 
Our main theoretical result is stated in Section \ref{sec:convergenceresults} and provides the convergence rate of our estimator in both the small-$T$ and the large-$T$ panel case. The central assumption required to derive the convergence rate is a so-called restricted eigenvalue condition. In Section \ref{sec:analysisofcomp}, we show that this condition is satisfied under a set of reasonable low-level assumptions. Extensions of our methods are discussed in Section \ref{sec:extensions}. There, we in particular propose a method to perform inference in our model setting. Moreover, we outline how to extend our methods and theory to the nonparametric additive case and highlight the technical complications involved. The methodological and theoretical analysis of the paper is complemented by a comprehensive simulation study in Section \ref{sec:simulations} and an empirical application to financial data in Section \ref{sec:empiricalstudy}.

\subsubsection*{\texttt{R} code} 
The \texttt{R} package \texttt{hdcce} implements the methods proposed in this paper and is publicly available on the Comprehensive R Archive Network (CRAN). Replication code for the Monte Carlo simulations and the empirical application is available at \url{https://github.com/RueckerM/nonLinearIFE_replicationFiles}.

\subsubsection*{Notation}

Matrices are denoted by bold letters,
whereas scalars and vectors are printed in normal font. For a vector $v =
(v_1,\ldots,v_q)^\top \in \reals^q$ and a set $S \subseteq \{1,\ldots,q\}$,
we let $v_S = (v_i: i \in S)$ be the vector which consists of the entries $%
v_i$ with $i \in S$ only. In addition, we sometimes write $v_{-i}$ to denote
the vector $v$ without the $i$-th component. We let $\|v\| = (\sum_i
v_i^2)^{1/2}$ denote the Euclidean norm of $v$, $\|v\|_1 = \sum_i|v_i|$ its $%
\ell_1$-norm, and $\|v\|_{\infty} = \max_i|v_i|$ its $\ell_\infty$-norm. For
a generic matrix $\bs{A}^{T \times p}$, we denote the row vectors by $A_t$
and the column vectors by $A_{(j)}$, that is, $\bs{A} = (A_1 \ldots
A_T)^\top = (A_{(1)} \ldots A_{(p)})$. Moreover, the matrix $\bs{A}$ without
the $t$-th row is denoted by $\bs{A}_{-t}$ and that without the $j$-th
column by $\bs{A}_{(-j)}$. 
The symbols $\eig_{\min}(\bs{A})$ and $\eig%
_{\max}(\bs{A})$ are used to denote the minimal and maximal eigenvalue of a
square matrix $\bs{A} \in \reals^{q \times q}$. In addition, we sometimes
write $\eig_1(\bs{A}) \ge \eig_2(\bs{A}) \ge \ldots \ge \eig_q(\bs{A})$ to
denote the eigenvalues of $\bs{A}$ (in decreasing order). For a general (not
necessarily square) matrix $\bs{A} = (a_{ij})$, $\| \bs{A} \|$, $\|\bs{A}%
\|_1 $, $\|\bs{A}\|_\infty$ and $\|\bs{A}\|_{\text{max}}$ are its spectral
norm, $\ell_1$-norm, $\ell_\infty$-norm and elementwise norm, respectively.
In particular, $\| \bs{A} \| = \eig_{\max}^{1/2}(\bs{A}^\top \bs{A})$, $\|%
\bs{A}\|_1 = \max_j \sum_i |a_{ij}|$, $\|\bs{A}\|_\infty = \max_i \sum_j
|a_{ij}|$ and $\|\bs{A}\|_{\max} = \max_{ij} |a_{ij}|$. The symbol $\bs{A}%
^{-}$ stands for the generalized inverse of a matrix $\bs{A}$ and the symbol 
$\bs{I}_q$ for the $q \times q$ identity matrix. Sometimes, we also write $%
\bs{I}$ instead of $\bs{I}_q$ for short. Finally, the indicator function is
denoted by $\ind(\cdot)$ and the cardinality of a set $S$ by $|S|$.

%% file: methods.tex
\section{Model framework}\label{sec:model}

We observe a sample of panel data $\{ (Y_{it}, X_{it}): 1 \le t \le T, \, 1 \le i \le n \}$ with real-valued random variables $Y_{it}$ and $\reals^p$-valued random vectors $X_{it} = (X_{it,1},\ldots,X_{it,p})^\top$, where $n$ is the cross-section dimension and $T$ the time series length. We analyze both the large-$T$ case, where $n \to \infty$ and $T \to \infty$, and the small-$T$ case, where $n \to \infty$ but $T$ is a fixed natural number. Throughout the paper, we regard both $T$ and $p$ as a function of $n$, that is, $T=T(n)$ and $p=p(n)$. Hence, asymptotic statements are to be understood in the sense that $n \to \infty$ (and $T=T(n) \to \infty$ in the large-$T$ case).

The observed data are supposed to follow the additive interactive fixed effects model
\begin{equation}\label{eq:model-add}
Y_{it} = \sum_{j=1}^p m_j(X_{it,j}) + \gamma_i^\top F_t+ \varepsilon_{it}
\end{equation}
where $m_j$ are unknown functions, $\varepsilon_{it}$ is an idiosyncratic error term with $\ex[\varepsilon_{it}] = 0$,  $F_t = (F_{t,1},\ldots,F_{t,K})^\top$ is a vector of unobserved factors and $\gamma_i = (\gamma_{i,1},\ldots,\gamma_{i,K})^\top$ is a vector of unobserved factor loadings for subject $i$. Throughout the paper, we treat the factors $F_t$ as non-random parameters, which corresponds to an implicit conditioning on the factors. The factor loadings $\gamma_i$, in contrast, are considered to be random. This modelling choice is quite common in the literature; see e.g.\ \cite{MoonWeidner2015} and \cite{Mruecker2025CCE}. As already mentioned in the introduction, we restrict attention to the parametric additive case, i.e., we model the link functions $m_j$ as 
\begin{equation}\label{eq:model-mjs}
m_j(x)= \phi_{j}(x)^\top \beta_j,  
\end{equation}
where $\phi_{j} = (\phi_{j1},\ldots, \phi_{jL_j})^\top$ is a vector of known functions, $\beta_j = (\beta_{j1},\ldots,\beta_{jL_j})^\top \in \reals^{L_j}$ is a vector of unknown parameters, and $L_j$ are fixed natural numbers with the property that $\max_{1 \le j \le p} L_j \le L_{\max} < \infty$ for some sufficiently large $L_{\max} \in \naturals$. As usual in the CCE literature, the regressors $X_{it}$ in \eqref{eq:model-add} are supposed to have the structure 
\begin{equation}\label{eq:model-reg}
X_{it} = \bs{\Gamma}_i F_t + Z_{it},
\end{equation}
where $\bs{\Gamma}_i \in \reals^{p \times K}$ is a matrix of individual-specific factor loadings and $Z_{it}$ represents the idiosyncratic part of the regressors with $\ex[Z_{it}] = 0$ for all $i$ and $t$. Notably, the regressors $X_{it}$ may be correlated with the error structure $e_{it} = \gamma_i^\top F_t + \varepsilon_{it}$ in the main equation \eqref{eq:model-add} because the factor loadings $\bs{\Gamma}_i$ and $\gamma_i$ are allowed to be correlated with each other in an arbitrary way. Hence, the regressors are endogenous.

We allow model \eqref{eq:model-add} to be high-dimensional in the sense that the number of additive components $d=\sum_{j=1}^p L_j$ may be very large, potentially much larger than the sample size $nT$ itself. To handle this high dimensionality, we impose certain sparsity constraints on the model. In particular, we assume that many of the parameters $\beta_{j\ell}$ $(1 \le j \le p, \, 1 \le \ell \le L_j)$ are equal to zero. A bit more formally, letting $S = \{(j,\ell): \beta_{j\ell} \ne 0\}$ be the active set and $s = |S|$ the sparsity index of the model, we assume that $s$ is small relative to the sample size $nT$. The precise sparsity constraints are presented in Section \ref{sec:assumptions}.

To develop our estimation methods, it is convenient to formulate the model in matrix notation. Specifically, for each cross-sectional unit $i$, we write 
\begin{equation}\label{eq:model-matrix}
Y_i = \bs{\Phi}_i \beta + \bs{F} \gamma_i + \varepsilon_i, 
\end{equation}
where 
$Y_i = (Y_{i1},\ldots,Y_{iT})^\top \in \reals^T$ is the response vector, the error vector $\varepsilon_i$ is defined analogously, $\beta = (\beta_1^\top,\ldots,\beta_p^\top)^\top \in \reals^{d}$ is the unknown parameter vector to be estimated, $\bs{F} = (F_1 \ldots F_T)^\top \in \reals^{T \times K}$ is the matrix of factors and 
\[ \bs{\Phi}_i = (\Phi(X_{i1}) \ldots \Phi(X_{iT}))^\top \in \reals^{T \times d} \] 
is the design matrix with $\Phi(X_{it}) =(\phi_{1}(X_{it,1})^\top,\ldots,\phi_{p}(X_{it,p})^\top)^\top$ and $\phi_{j}(X_{it,j})= (\phi_{j1}(X_{it,j}),\ldots,\phi_{jL_j}(X_{it,j}))^\top$. The observed regressors $\bs{X}_i = (X_{i1}, \ldots, X_{iT})^\top$ from which the design matrix $\bs{\Phi}_i$ is computed are assumed to have the form \eqref{eq:model-reg}, which can be reformulated as
\begin{equation}\label{eq:model-reg-matrix}
\bs{X}_i = \bs{F} \bs{\Gamma}_i^\top + \bs{Z}_i    
\end{equation}
with $\bs{Z}_i$ defined analogously as $\bs{X}_i$. For technical reasons, we work with the following empirically centred version of \eqref{eq:model-matrix} in what follows:  
\begin{equation}\label{eq:model-centred}
Y_{i} - \overline{Y}
 = ({\bs{\Phi}}_i - \overline{\bs{\Phi}}) \beta + \bs{F} (\gamma_i - \overline{\gamma}) + (\varepsilon_i - \overline{\varepsilon})
\end{equation}
for $1 \le i \le n$. Here, $\overline{Y} = (\overline{Y}_1,\ldots,\overline{Y}_T)^\top$ is the vector of cross-sectional averages $\overline{Y}_t = n^{-1} \sum_{i=1}^n Y_{it}$ of the responses and the other quantities $\overline{\bs{\Phi}}$, $\overline{\gamma}$ and $\overline{\varepsilon}$ are defined analogously.

\section{Estimation methods}\label{sec:estimation-method}

In this section, we construct an estimator $\widehat{\beta}_\lambda$ of the unknown parameter vector $\beta$, which immediately yields an estimator of the component function $m_j$ for each $j$, in particular, the estimator $\widehat{m}_j(x) =  \phi_{j}(x)^\top \widehat{\beta}_{\lambda,j}$.

\subsection{Heuristics}
\label{sec:heuristics}

Before we give a precise definition of the estimator $\widehat{\beta}_\lambda$, we describe the main ideas of our approach and explain in detail why the additive case with $m_j(x)= \beta_j^\top \phi_{j}(x)$ is much more difficult to handle than the linear case with $m_j(x) = \beta_j x$. To do so, it is instructive to consider the (infeasible) oracle setting where the factors $F_t$ are known.

In the linear case, we may pursue the following strategy: we first eliminate or ``project away'' the factors from the model by a suitable transformation and then apply high-dimensional regression techniques to the transformed data. More formally, consider the linear model 
\begin{equation}\label{eq:model-linear} 
Y_i = \bs{X}_i \beta +  \bs{F} \gamma_i + \varepsilon_i \quad \text{with} \quad \bs{X}_i = \bs{F} \bs{\Gamma}_i^\top + \bs{Z}_i 
\end{equation} 
and let $\bs{\Pi}$ be the projection matrix onto the orthogonal complement of the column space of $\bs{F}$, i.e., 
\begin{equation}\label{eq:oracle-Pi}
\bs{\Pi} = \bs{I} - \bs{F} (\bs{F}^\top \bs{F})^{-} \bs{F}^\top, 
\end{equation} 
which is known in the oracle setting.
Applying $\bs{\Pi}$ to the linear model \eqref{eq:model-linear} yields the projected model
\begin{equation}\label{eq:model-linear-projected} 
\bs{\Pi} Y_i = \bs{\Pi} \bs{X}_i \beta + \bs{\Pi} \varepsilon_i \quad \text{with} \quad \bs{\Pi} \bs{X}_i = \bs{\Pi} \bs{Z}_i, 
\end{equation} 
thus completely eliminating the factor components $\bs{F} \gamma_i$ and $\bs{F} \bs{\Gamma}_i^\top$ from the error and covariate structure in \eqref{eq:model-linear}. Equation \eqref{eq:model-linear-projected} is a more or less standard high-dimensional regression model whose regressors $\bs{\Pi} \bs{X}_i$ are no longer endogenous (because endogeneity is solely induced by the factor components). We may thus expect the following: as long as the overall design matrix of the projected model 
\[ \bs{X}^\perp := \begin{pmatrix} \bs{\Pi} \bs{X}_1 \\ \vdots \\ \bs{\Pi} \bs{X}_n \end{pmatrix} \]
is regular enough in the sense of satisfying a suitable restricted eigenvalue (RE) condition, we can apply standard high-dimensional regression techniques such as the lasso to the transformed data sample $\{ (\bs{\Pi} Y_i, \bs{\Pi} \bs{X}_i): 1 \le i \le n \}$ in order to obtain an estimator of $\beta$ with good theoretical properties. This is exactly the approach pursued in \cite{Mruecker2025CCE}.

Let us now try to adapt this strategy to the additive case. Notably, the additive model \eqref{eq:model-matrix} is identical to the linear model \eqref{eq:model-linear} except that we now have the more complicated design matrix $\bs{\Phi}_i = (\Phi(X_{i1}),\ldots,\Phi(X_{iT}))^\top$. Applying the projection matrix $\bs{\Pi}$ to the additive model \eqref{eq:model-matrix} yields 
\begin{equation}\label{eq:model-matrix-projected} 
\bs{\Pi} Y_i = \bs{\Pi} \bs{\Phi}_i \beta +  \bs{\Pi} \varepsilon_i \quad \text{with} \quad \bs{\Pi} \bs{\Phi}_i = \bs{\Pi} (\Phi(X_{i1}) \ldots \Phi(X_{iT}))^\top.
\end{equation} 
As in the linear case, the factor component gets fully eliminated from the error structure. However, as the factors enter the design matrix $\bs{\Phi}_i$ in a nonlinear fashion (due to the nonlinearity of the functions $\phi_{j\ell}$), they do not get eliminated (not even approximately) from the design matrix $\bs{\Pi} \bs{\Phi}_i$. Importantly, endogeneity of the regressors is accounted for nevertheless because correlation between errors and regressors is solely induced by correlation between the factor components $\bs{F} \gamma_i$ and $\bs{F} \bs{\Gamma}_i^\top$ in the error and covariate structure. As the error component $\bs{F} \gamma_i$ is eliminated in the projected model \eqref{eq:model-matrix-projected}, there is no correlation left between errors and design part. Analogous to the linear case, we may thus expect the following: as long as the overall design matrix of the projected model 
\begin{equation}\label{eq:oracle-design-matrix}
\bs{\Phi}^\perp := \begin{pmatrix} \bs{\Pi} \bs{\Phi}_1 \\ \vdots \\ \bs{\Pi} \bs{\Phi}_n \end{pmatrix} 
\end{equation}
is regular enough in the sense of satisfying an RE condition, we can obtain a reasonable estimator of $\beta$ by running the lasso or related high-dimensional regression techniques on the transformed sample $\{ (\bs{\Pi} Y_i, \bs{\Pi} \bs{\Phi}_i): 1 \le i \le n \}$. There is, however, a crucial difference between the linear and the additive case: the design matrix $\bs{\Phi}^\perp$ has a vastly more complex structure than $\bs{X}^\perp$. Specifically, whereas the factors $\bs{F}$ are contained in $\bs{\Phi}^\perp$ in a complicated nonlinear fashion, they are not contained at all in $\bs{X}^\perp$. This makes a substantial difference when it comes to the RE condition: the techniques from \cite{Mruecker2025CCE} to show that $\bs{X}^\perp$ fulfills an RE condition are essentially useless in the additive case because $\bs{\Phi}^\perp$ is structurally so different from $\bs{X}^\perp$. Verifying an RE condition for $\bs{\Phi}^\perp$ turns out to be a theoretically much harder problem which requires totally different tools.  

In summary: The estimation approach proposed by \cite{Mruecker2025CCE} for the linear case can be extended to the additive case with minor modifications, as detailed in the remainder of this section. Consequently, their approach applies much more broadly than in the linear case alone. On the downside, the theory developed for the linear case does not readily extend to the additive case; instead, substantially different arguments are required, as shown in Section \ref{sec:theoreticalresults} and the corresponding proofs.

\subsection{Estimation algorithm}
\label{sec:estimationalgo}

We now define our estimator $\widehat{\beta}_\lambda = (\widehat{\beta}_{\lambda,1}^\top,\ldots,\widehat{\beta}_{\lambda,p}^\top)^\top$ of the unknown parameter vector $\beta = (\beta_1^\top,\ldots,\beta_p^\top)^\top$, which generalizes the estimator for the linear case proposed by \cite{Mruecker2025CCE}. We refer to it as the HD-CCE (short for high-dimensional common correlated effects) estimator.

\subsubsection*{Step 1: Estimation of the unknown number of factors $\bs{K}$}

Let $\overline{\bs{X}} = (\overline{X}_1 \ldots \overline{X}_T)^\top$ be the matrix of cross-sectional averages $\overline{X}_t = n^{-1} \sum_{i=1}^n X_{it}$. Compute the high-dimensional $p \times p$ matrix 
$\widehat{\bs{\Sigma}} = T^{-1} \overline{\bs{X}}^\top \overline{\bs{X}}$ and perform an eigendecomposition of $\widehat{\bs{\Sigma}}$, which yields the eigenvalues $\widehat{\eig}_1 \ge \widehat{\eig}_2 \ge \ldots \ge \widehat{\eig}_p \ge 0$ and the corresponding orthonormal eigenvectors $\widehat{U}_1,\ldots,\widehat{U}_p$. Estimate the unknown number of factors $K$ by 
\[ \widehat{K} = \sum_{j=1}^p \ind\big(\widehat{\eig}_j \ge \tau\big), \]
where $\tau = \tau_{n,T}$ is a threshold parameter that is of slightly smaller order than $p$. Precise technical conditions on $\tau$ are provided in Theorem \ref{theo:main}, while rules for selecting $\tau$ in practice are discussed in Section \ref{sec:est:tuning:tau}. 

\subsubsection*{Step 2: Approximation of the unknown projection matrix $\bs{\Pi}$}

Let $\widehat{\bs{U}} = (\widehat{U}_1 \ldots \widehat{U}_{\widehat{K}})$ be the matrix of eigenvectors of $\widehat{\bs{\Sigma}}$ that correspond to the $\widehat{K}$ largest eigenvalues $\widehat{\eig}_1 \ge \ldots \ge \widehat{\eig}_{\widehat{K}}$ and define $\widehat{\bs{W}} = \overline{\bs{X}} \widehat{\bs{U}}$. Approximate $\bs{\Pi}$ by 
\[ \widehat{\bs{\Pi}} = \bs{I} - \widehat{\bs{W}} (\widehat{\bs{W}}^\top \widehat{\bs{W}})^{-} \widehat{\bs{W}}^\top. \]

\subsubsection*{Step 3: Estimation of $\bs{\beta}$} 

Run a lasso regression on the transformed data sample $\{ (\widehat{Y}_i, \widehat{\bs{\Phi}}_i): 1 \le i \le n \}$, where $\widehat{Y}_i = \widehat{\bs{\Pi}} (Y_i - \overline{Y})$ and $\widehat{\bs{\Phi}}_i = \widehat{\bs{\Pi}}(\bs{\Phi}_i - \overline{\bs{\Phi}})$. Specifically, define the HD-CCE estimator of $\beta$ by  
\[ \widehat{\beta}_\pen \in \underset{b \in \reals^{d}}{\text{argmin}} \bigg\{ \frac{1}{nT} \sum_{i=1}^n \big\| \widehat{Y}_i - \widehat{\bs{\Phi}}_i b \big\|^2 + \pen \sum_{j=1}^p\sum_{\ell=1}^{L_j} |b_{j\ell}| \bigg\}, \]
where $\pen > 0$ is the penalty parameter of the lasso and we use the notation $b = (b_1^\top,\ldots,b_p^\top)^\top$ with $b_j = (b_{j1},\ldots,b_{jL_j})^\top$ for each $j$. 
It is also possible to define a (sparse) group lasso version of the HD-CCE estimator. However, to keep the exposition as simple as possible, we work with the standard lasso version defined above throughout the paper. Notably, our theoretical results can be extended to the (sparse) group lasso case by rather straightforward adaptions of the technical arguments.

\subsection{Tuning parameter choice}

The HD-CCE estimator $\widehat{\beta}_\pen$ depends on two tuning parameters: the threshold parameter $\tau$ for the estimation of $K$ and the penalty parameter $\pen$ of the lasso.

\subsubsection*{Choice of $\bs{\tau}$}\label{sec:est:tuning:tau}

It can be shown formally that the eigenvalues $\widehat{\eig}_k$ are of order $p$ for $k \le K$ but of much smaller order for $k > K$. Hence, to ensure that $\widehat{K}$ is a consistent estimator of $K$, we need to choose $\tau$ such that it separates the ``large'' eigenvalues of order $p$ (that is, those with $k \le K$) from the ``small'' ones (that is, those with $k > K$). As a practical rule-of-thumb, we regard an eigenvalue $\widehat{\eig}_k$ as ``small'' if $\widehat{\eig}_k/\widehat{\eig}_1 < \alpha$ with some small $\alpha$ (such as $\alpha = 0.05$ or $\alpha = 0.01$). Put differently, we regard $\widehat{\eig}_k$ as ``small'' if it is less than $100\cdot\alpha\%$ of the largest eigenvalue $\widehat{\eig}_1$ in size. This rule-of-thumb results in the choice $\tau = \alpha \widehat{\eig}_1$.\footnote{From a theoretical point of view, we need to let $\alpha = \alpha_{n,T}$ slowly go to $0$ with increasing sample size to make sure that $\tau = \alpha \widehat{\eig}_1$ is of somewhat smaller order than $p$ and thus produces a consistent estimator $\widehat{K}$ of $K$. In practice, however, the sample size is fixed, implying that $\alpha$ is a fixed number as well. We thus do not reflect the dependence of $\alpha$ on $n$ and $T$ in the notation.}


\subsubsection*{Choice of $\bs{\pen}$ for the HD-CCE estimator}\label{sec:est:tuning:lambda}

We choose the penalty parameter $\lambda$ by a version of $M$-fold cross-validation:
Divide the sample $\{(\widehat{Y}_i,\widehat{\bs{\Phi}}_i): i=1,\ldots,n\}$ of the projected data into $M$ folds $\mathcal{F}_1,\ldots,\mathcal{F}_{M}$, where $\mathcal{F}_{m} = \{(\widehat{Y}_i,\widehat{\bs{\Phi}}_i): i \in \mathcal{I}_m \}$ with $\mathcal{I}_m = \{ (m-1) \lfloor n/M \rfloor + 1,\ldots, m \lfloor n/M \rfloor\}$ for $m=1,\ldots,M-1$ and $\mathcal{I}_M = \{ (M-1) \lfloor n/M \rfloor + 1,\ldots, n\}$. Then run standard $M$-fold cross-validation over a grid of $\lambda$-values.\footnote{In our \texttt{R} package, we use the grid chosen by the cross-validation function of the \texttt{glmnet} package.}

%% file: theory.tex
\section{Theoretical results}\label{sec:theoreticalresults}

\subsection{Assumptions}\label{sec:assumptions}

To state our technical assumptions for the components of model \eqref{eq:model-matrix}--\eqref{eq:model-reg-matrix}, let $C < \infty$ denote a generic positive constant that does not depend on $n$, $T$ and $p$. Moreover, for each cross-sectional unit $i$, define $\mathcal{C}_{\varepsilon}^{(i)} = \{ \varepsilon_{it}: 1 \le t \le T \}$ along with $\mathcal{C}_{\varepsilon} = \bigcup_{i=1}^n \mathcal{C}_{\varepsilon}^{(i)}$, and analogously let $\mathcal{C}_{Z}^{(i)} = \{ Z_{it}: 1 \le t \le T \}$ together with $\mathcal{C}_{Z} = \bigcup_{i=1}^n \mathcal{C}_{Z}^{(i)}$. 
\begin{enumerate}[label=(C\arabic*),leftmargin=1.0cm]

\item \label{C:loadings} 
The factor loadings $\{(\gamma_i, \bs{\Gamma}_i): 1 \le i \le n \}$ are independent of $\mathcal{C}_{\varepsilon}$ and $\mathcal{C}_{Z}$. Moreover, they are independent and identically distributed across $i$, with means $\gamma = \ex[\gamma_i]$ and $\bs{\Gamma} = \ex[\bs{\Gamma}_i]$, and satisfy the moment restrictions $\max_{k} \ex[|\gamma_{ik}|^\theta]\leq C$ and $\max_{j,k} \ex[|\Gamma_{i,jk}|^\theta]\leq C$ for some $\theta>8$. 

\item \label{C:eps} 
The family of random variables $\mathcal{C}_{\varepsilon}$ is independent of $\mathcal{C}_{Z}$, and the time series processes $\mathcal{C}_{\varepsilon}^{(i)}$ are mutually independent across $i$. Moreover, for each $t$, the variables $\{\varepsilon_{it}: 1 \le i \le n \}$ are identically distributed. Finally, for all $i$ and $t$,  $\ex[\varepsilon_{it}] = 0$ and $\ex[|\varepsilon_{it}|^\theta] \le C$ for some $\theta > 8$. 

\item \label{C:Z} The processes $\mathcal{C}_{Z}^{(i)}$ are mutually independent across $i$. Furthermore, for each $t$, the variables $\{Z_{it}: 1 \le i \le n \}$ are identically distributed. Finally, for all $i$, $j$ and $t$, $\ex[Z_{it,j}] = 0$ as well as $\ex[|Z_{it,j}|^{\theta}] \le C$ for some $\theta > 8$.  

\item \label{C:phi} For all $i$, $j$, $\ell$ and $t$, $\ex [|\phi_{j\ell}(X_{it,j})|^\theta] \leq C$ for some $\theta > 8$.


\item \label{C:factorssummable} The factors $F_t$ are normalized to be orthonormal, i.e., $(\bs{F}^\top\bs{F})/T = \bs{I}_{K\times K}$. It holds that $\max_{1 \le k \le K} \{ T^{-1} \sum_{t=1}^T|F_{t,k}|^\theta \} \le C$ for some $\theta>8$. 

\item \label{C:id2GAMMATGAMMAEV} 
The minimal and maximal eigenvalue $\eig_{\min}(\bs{\Gamma}^\top \bs{\Gamma}/p)$ and $\eig_{\max}(\bs{\Gamma}^\top \bs{\Gamma}/p)$ of the matrix $\bs{\Gamma}^\top \bs{\Gamma}/p$ are such that $0 < c_{\min} \le \eig_{\min} (\bs{\Gamma}^\top \bs{\Gamma}/p) \le \eig_{\max} (\bs{\Gamma}^\top \bs{\Gamma}/p) \le c_{\max} < \infty$ for some fixed constants $c_{\min}$ and $c_{\max}$. 

\end{enumerate}
\ref{C:loadings}--\ref{C:phi} impose a series of rather standard independence and moment conditions on the model variables. 
As is well-known from the literature on factor models, the orthogonality assumption on the factors in \ref{C:factorssummable} is without loss of generality. Moreover, the second constraint in \ref{C:factorssummable}, according to which $\max_{1 \le k \le K} \{ T^{-1} \sum_{t=1}^T|F_{t,k}|^\theta \} \le C$ is rather mild as well. If $\{F_{t,k}: 1 \le t \le T \}$ were a time series of weakly dependent random variables with sufficiently many moments, then standard concentration bounds would imply that $ T^{-1} \sum_{t=1}^T|F_{t,k}|^\theta \le C < \infty$ with probability approaching 1. Hence, if we think of our factors as realizations of such time series, the second requirement in  \ref{C:factorssummable} will be fulfilled with high probability.
By assuming \ref{C:id2GAMMATGAMMAEV}, we focus on the case of strong factors as in most other works on high-dimensional approximate factor models \citep[see e.g.][]{Fan2013}. 
A profound discussion of weak factor detection can be found in \cite{Onatski2010}.

In addition to the above assumptions, we require the design matrix 
\[\widehat{\bs{\Phi}} =(\widehat{\bs{\Phi}}_1^\top \ldots \widehat{\bs{\Phi}}_n^\top)^\top \in \reals^{nT \times d} \quad \textnormal{with} \quad \widehat{\bs{\Phi}}_i =  \widehat{\bs{\Pi}}(\bs{\Phi}_i - \overline{\bs{\Phi}})\]  to be sufficiently well-behaved in the sense that it fulfills a restricted eigenvalue condition. 
\begin{definition}\label{def:RE-condition}
A matrix $\bs{A} \in \reals^{nT \times d}$ fulfills the restricted eigenvalue condition $\textnormal{RE}(I,{\varphi^2})$ for some index set $I \subseteq \{ (j,\ell): 1 \le j \le p, \, 1 \le \ell \le L_j \}$ and a constant $\varphi^2 > 0$ if  
\[ \| b \|_1 ^2 \le \frac{|I|}{\varphi^2} \frac{\| \bs{A} b \|^2}{nT} \qquad \text{for all } b \text{ with }  \| b_{I^c} \|_1 \le 3 \| b_I \|_1. \]
\end{definition}
\noindent With this definition, the formal assumption on the design matrix $\widehat{\bs{\Phi}}$ reads as follows:
\begin{enumerate}[label=(C\arabic*), leftmargin=1.0cm]
\setcounter{enumi}{6}
\item \label{C:id:PhiCompatibilityCondition} The design matrix $\widehat{\bs{\Phi}}$ satisfies the condition RE$(S,{\varphi^2})$ for some constant $\varphi > 0$ with probability tending to $1$, i.e., there exist non-negative numbers $c_{n,T}$ with $c_{n,T}\to 0$ such that 
\[ \pr\Big( \widehat{\bs{\Phi}} \ \textnormal{fulfills}\  \textnormal{RE}(S, {\varphi^2}) \Big) = 1 -c_{n,T}.\]
\end{enumerate} 
Such a condition is standard in high-dimensional statistics; see \cite{BickelRitovTsybakov2009} or \cite{BuehlmannvandeGeer2011}. Many papers on high-dimensional additive models impose an RE condition without further justification, even though it is far from clear that such a condition holds under reasonable assumptions in such models. There are, however, a few notable exceptions: Both \cite{Meier2009} and \cite{HuandHorrowitzAdditive2006} justify an RE condition in an additive model for cross-sectional data where the component functions are modelled by splines and the covariates have (known) bounded support. Moreover, \cite{scheidegger2023spectral} verify an RE condition in an additive model with hidden confounders and Gaussian cross-sectional design. We verify the RE condition from \ref{C:id:PhiCompatibilityCondition} under simple low-level assumptions in Section \ref{sec:analysisofcomp}. 

Besides \ref{C:loadings}--\ref{C:id:PhiCompatibilityCondition}, we need some restrictions on the dimension $p$ and the sparsity index $s$. In the large-$T$ case, we impose the following conditions on the dimension para\-meters $n$, $T$, $p$, $s$ and $K$: 
\begin{enumerate}[label=(D$_{\ell}$\arabic*), leftmargin=1.15cm]

 \item \label{C:nTp-large} The dimensions $n$, $T$ and $p$ are such that $(npT)^{\frac{1}{\theta-\xi}} = o(\sqrt{n /\log(pT)})$ for some small $\xi>0$ with $\theta$ specified in \ref{C:loadings}--\ref{C:factorssummable}.

\item \label{C:s-large} The set $S=\{(j,\ell): \beta_{j\ell} \ne 0\}$ of non-zero components of $\beta$ has cardinality $s := |S|$ with $s = o(\sqrt{n / \log(pT)} (npT)^{-2/\theta} )$. 
\item \label{C:K-large} The number of factors $K$ is a fixed natural number with $K < T$ and $K \le p$.
\end{enumerate}
\ref{C:nTp-large} is equivalent to $pT=o(n^{\{(\theta-\xi)/2\} - 1} \{\log(pT)\}^{-(\theta-\xi)/2})$. To better understand its implications, suppose that both $p$ and $T$ grow polynomially in $n$, that is, $p = n^a$ and $T = n^b$ for some $a, b > 0$. In this case,  \ref{C:nTp-large} simplifies to the condition that $a + b <\frac{\theta-\xi}{2}-1$. Hence, how fast $p$ (and $T$) can grow relative to $n$ depends on how many moments $\theta$ the model variables have. If all moments exist, meaning that $\theta$ can be chosen as large as desired, then $a$ and $b$ can be arbitrarily large. In particular, for any fixed $b>0$, we can set $a=q(b+1)$ with $q>0$ arbitrarily large, implying that $p=n^{q(b+1)}=(nT)^q$. Hence, if all moments $\theta$ exist, then $p$ can grow as an arbitrary polynomial of the sample size $nT$. If the number of moments $\theta$ is only moderate, in contrast, then \ref{C:nTp-large} implies substantially stronger restrictions on the growth of $p$ (and $T$). 
\ref{C:s-large} imposes constraints on the growth of the sparsity index $s$, i.e., on the number of non-zero components of $\beta$. The constraints again depend on the number of moments $\theta$: the more moments exist, the weaker the restrictions on $s$. In particular, if all moments exist, then $s$ can grow as fast as $\sqrt{n}$ (up to log factors). This restriction on the sparsity index $s$ is more stringent than the restrictions typically imposed in high-dimensional models for cross-sectional i.i.d.\ data, where the sparsity index is usually required to grow slightly more slowly than the square root of the sample size. The more stringent restriction on $s$ is a consequence of the fact that we do not impose any conditions on the time series dependence of the model variables, as discussed in more detail in Remark \ref{remark:largeT-Rate} below. 
In the small-$T$ case, we work with conditions on the dimension parameters $n$, $T$, $p$, $s$ and $K$ which are completely analogous to those for the large-$T$ case: 
\begin{enumerate}[label=(D$_s$\arabic*), leftmargin=1.15cm]
\item  \label{C:nTp-small}   The dimensions $n$ and $p$ are such that $(np)^{\frac{1}{\theta-\xi}} = o(\sqrt{n / \log(p)})$ for some small $\xi>0$ with $\theta$ specified in \ref{C:loadings}--\ref{C:factorssummable}.

\item \label{C:s-small} The set $S=\{(j,\ell): \beta_{j\ell} \ne 0\}$ of non-zero components of $\beta$ has cardinality $s := |S|$ with $s = o(\sqrt{n/\log (p)} (np)^{-2/\theta})$. 
\item \label{C:K-small} The number of factors $K$ is a fixed natural number with $K < T$ and $K \le p$.
\end{enumerate}

\noindent Importantly, the above conditions on the dimension parameters (both in the large-$T$ and the small-$T$ case) allow us to deal with a wide range of scenarios (as long as the tails of the model variables are not too thick, i.e., as long as sufficiently many moments $\theta$ exist). In particular, we can deal with ``standard low-dimensional'' scenarios where $p$ is small and fixed, with ``moderately high-dimensional'' scenarios where $p$ is fairly large but still smaller than the sample size $nT$ and with ``truly high-dimensional'' scenarios where $p$ exceeds the sample size $nT$.

\subsection{Identification}\label{sec:identification}

The same arguments as in \cite{Mruecker2025CCE} yield that the number of factors $K$ and the sparse parameter vector $\beta$ are identified in model \eqref{eq:model-matrix}--\eqref{eq:model-reg-matrix}. More formally, the following results hold.  
\begin{lemma}\label{lemma:id-K}
Let \ref{C:loadings}--\ref{C:id2GAMMATGAMMAEV} be satisfied. 
Moreover, assume \ref{C:nTp-large}--\ref{C:K-large} in the large-$T$ case and \ref{C:nTp-small}--\ref{C:K-small} in the small-$T$ case. Then the number of factors $K$ is unique for sufficiently large $n$.
\end{lemma}
\begin{lemma}\label{lemma:id-beta}
Let \ref{C:loadings}--\ref{C:id2GAMMATGAMMAEV} be fulfilled. Moreover, assume \ref{C:nTp-large}--\ref{C:K-large} in the large-$T$ case and \ref{C:nTp-small}--\ref{C:K-small} in the small-$T$ case. If the oracle design matrix $\bs{\Phi}^\perp$ from \eqref{eq:oracle-design-matrix} satisfies the RE($I,\varphi^2$) condition for some $\varphi > 0$ and all $|I|\leq 2s$ with probability tending to $1$, then the sparse parameter vector $\beta$ is unique for sufficiently large $n$.
\end{lemma}
\noindent The assumption that $\bs{\Phi}^\perp$ satisfies the RE($I,\varphi^2$) condition for all $|I|\leq 2s$ with probability tending to $1$ can be verified by the same arguments as \ref{C:id:PhiCompatibilityCondition}. We refer to Section \ref{sec:analysisofcomp} for a thorough treatment of \ref{C:id:PhiCompatibilityCondition}. The proofs of Lemmas \ref{lemma:id-K} and \ref{lemma:id-beta} are omitted as they are virtually identical to the proofs of Theorems 3.1 and 3.2 from \cite{Mruecker2025CCE}.

\subsection{Convergence rates}\label{sec:convergenceresults}

We now derive the convergence rate of the HD-CCE estimator $\widehat{\beta}_\lambda$, assuming that the RE condition \ref{C:id:PhiCompatibilityCondition} on the design $\widehat{\bs{\Phi}}$ is satisfied. A thorough analysis of \ref{C:id:PhiCompatibilityCondition} is provided in the subsequent section. There, we in particular show that \ref{C:id:PhiCompatibilityCondition} is fulfilled under reasonable low-level assumptions. 

\begin{theorem}\label{theo:main} 
Let \ref{C:loadings}--\ref{C:id:PhiCompatibilityCondition} be fulfilled. In addition, assume \ref{C:nTp-large}--\ref{C:K-large} in the large-$T$ case and \ref{C:nTp-small}--\ref{C:K-small} in the small-$T$ case. Pick the threshold parameter $\tau$ such that $\tau = o(p)$ and $\{p \sqrt{\log(pT)}/\sqrt{n}\}  / \tau = o(1)$. Moreover, choose the penalty constant as $\lambda = h_n \nu_{n,T}$, where
$\nu_{n,T} = \sqrt{\log(pT)/n} \, (npT)^{2/\theta}$
and $\{h_n\}$ is a slowly diverging sequence, e.g., $h_n =\log \log(n)$. Then 
\[\|\widehat{\beta}_{\lambda} - \beta \|_1 =O_p\left(  s \lambda   \right). \] 
\end{theorem}

\noindent The proof of Theorem \ref{theo:main} is given in the Appendix. 

\begin{remark}
Replacing the RE$(I,\varphi^2)$ condition in \ref{C:id:PhiCompatibilityCondition} by the slightly stronger RE condition
\[\|b\|_2^2 \leq \frac{1}{\varphi^2}\frac{\|\bs{A}b\|^2}{nT} \quad \textnormal{for all}\ b \ \textnormal{with}\ \|b_{I^c}\|_1 \leq 3 \|b_I\|_1, \] 
it is straightforward to modify the proof of Theorem \ref{theo:main} and to derive the $\ell_2$-rate
\[\|\widehat{\beta}_{\lambda} - \beta \|_2 =O_p\left(  \sqrt{s}\lambda \right).\] 
\end{remark}

\begin{remark}\label{remark:largeT-Rate}
The rate derived in Theorem \ref{theo:main} can be written as
\[ \|\widehat{\beta}_{\lambda} - \beta \|_1 =O_p\left(  q_{n,T} \, \frac{s}{\sqrt{n}} \right) \quad \text{with} \quad q_{n,T} = h_n \sqrt{\log(pT)} (npT)^{2/\theta}. \]
Hence, the rate equals $s/\sqrt{n}$ times a factor $q_{n,T}$ which grows with the sample size. The larger $\theta$, i.e., the more moments the model variables have, the more slowly the growth of $q_{n,T}$. In the extreme case where all moments exist, i.e, where $\theta$ can be chosen arbitrarily large, $q_{n,T}$ becomes a product of log-factors and thus diverges rather slowly. Consequently, as long as the model variables have sufficiently many moments, the rate comes close to $s/\sqrt{n}$. \\
Interestingly, in the large-$T$ case, this rate does not improve as $T$ increases. 
The reason is as follows: our technical assumptions in Section \ref{sec:assumptions} do not impose any constraints on the time series dependencies of the model variables, specifically, of the variables $Z_{it}$ and $\varepsilon_{it}$. Hence, the time series dependencies are potentially very strong, implying that a larger time series length $T$ does not yield additional information. As a result, the rate does not improve with larger $T$. 
The absence of such an improvement should not be interpreted as an intrinsic limitation of the proposed HD-CCE estimator. Rather, it reflects the generality of our assumptions, which allow for arbitrary temporal dependence in the idiosyncratic components. Under additional restrictions on the time series dependence, such as independence or suitable mixing conditions, faster convergence rates may be attainable. The main challenge in establishing faster rates is that the projected design matrix $\widehat{\bs{\Pi}} \bs{\Phi}_i$ involves the factors in an intricate nonlinear manner. As a consequence, the transformed regressors (i.e., the rows of the matrix $\widehat{\bs{\Pi}} \bs{\Phi}_i$) may exhibit complex time series dependencies even if the original model variables $Z_{it}$ are weakly dependent or even independent. We leave the question whether faster rates are attainable for future research.
\end{remark}



\subsection{Analysis of the RE condition}\label{sec:analysisofcomp}

We now have a closer look at condition \ref{C:id:PhiCompatibilityCondition}, according to which the projected design matrix $\widehat{\bs{\Phi}}$ satisfies the RE$(S,\varphi^2)$ condition from Definition \ref{def:RE-condition} with probability tending to $1$. 
The overall strategy of our analysis is as follows: we first relate the RE condition to a more tractable population version and then verify this population version under suitable assumptions. 
We begin by relating the RE$(S,\varphi^2)$ condition on $\widehat{\bs{\Phi}}$ to the following population level condition: there exists a constant $c^{(1)} > 0$ such that for all $T$ and $p$, 
\begin{equation}\label{eq:CC-theoretical}
\frac{1}{T} \ex\left[\|\bs{\Pi} (\bs{\Phi}_1 -  \ex[\bs{\Phi}_1])v\|^2\right] \geq c^{(1)} \sum_{j=1}^p\sum_{\ell=1}^{L_j} |v_{j\ell}|^2
\quad \text{for all } v\in \reals^{d},
\end{equation}
where $\bs{\Pi}$ is the oracle projection matrix defined in \eqref{eq:oracle-Pi}. 
\begin{theorem}\label{prop:compatibility} 
Let \ref{C:loadings}--\ref{C:id2GAMMATGAMMAEV} be fulfilled. Moreover, assume \ref{C:nTp-large}--\ref{C:K-large} in the large-$T$ case and \ref{C:nTp-small}--\ref{C:K-small} in the small-$T$ case.  If condition \eqref{eq:CC-theoretical} holds, then the $RE(S,\varphi^2)$ condition on $\widehat{\bs{\Phi}}$ is satisfied with probability tending to $1$, that is,  \ref{C:id:PhiCompatibilityCondition} is fulfilled. 
\end{theorem}
\noindent We next verify the population version \eqref{eq:CC-theoretical} of the RE condition for the case of Gaussian regressors and polynomial functions $\phi_{j\ell}$. 
\begin{theorem}\label{prop:compatibility-parametric}
Assume \ref{C:loadings}--\ref{C:id2GAMMATGAMMAEV} and the following:
\begin{enumerate}[label=(\alph*),leftmargin=0.75cm, itemsep=0pt]
\item \label{cond:GammaGaussian} 
For each $i$, the vectorized version of $\bs{\Gamma}_i$ (which results from stacking the columns of $\bs{\Gamma}_i$) is Gaussian. 
\item \label{cond:ZGaussian} The vectors $Z_{it}$ are independent and identically distributed across $i$ and $t$. For each $i$ and $t$, $Z_{it}$ is a centred Gaussian random vector with covariance matrix $\bs{\Sigma}_Z$, where $\lambda_{\textnormal{min}}(\bs{\Sigma}_Z)\geq c_Z$ for some positive constant $c_Z$.
\item \label{cond:factorsbounded} The factors $F_t$ are uniformly bounded, i.e., $|F_{t,k}|\leq C  < \infty$ for some $C>0$.
\item \label{cond:polynomoialdictionary} The functions $\phi_{j\ell}$ are polynomial, in particular, $\phi_{j\ell}(x) = x^\ell$ for $1 \le \ell \le L_j$ and all $j$. 
\end{enumerate} 
Under these conditions, \eqref{eq:CC-theoretical} is fulfilled.
\end{theorem}
\noindent It is worth noting that the above result is not restricted to polynomial functions $\phi_{j\ell}$. The proof goes through for any set of functions $\phi_{j\ell}$ with the following property: there exists a fixed constant $c$ (which is in particular independent of $j$ and $t$) such that  
\[ \psi_{\min}(\bs{\Omega}_{jt}) \ge c > 0, \]
where $\bs{\Omega}_{jt}$ is the covariance matrix of the random vector $\phi_j(X_{1t,j}) = (\phi_{j1}(X_{1t,j}),\ldots$ $\ldots,\phi_{jL_j}(X_{1t,j}))^\top$ and $\psi_{\min}(\bs{\Omega}_{jt})$ is the minimal eigenvalue of $\bs{\Omega}_{jt}$. In other words, the eigenvalues of the $L_j\times L_j$ covariance matrices $\bs{\Omega}_{jt}$ have to be uniformly lower bounded.

The general strategy of our analysis to first replace the RE condition by a population version and then to verify this population version under certain constraints is not uncommon in the literature; see \cite{Meier2009}, \cite{2013ITIT...59.3434R} and \cite{scheidegger2023spectral} among others. 
In order to prove Theorem \ref{prop:compatibility-parametric}, we exploit recent achievements on maximum correlation theory by \cite{GUO20221037}. The same theory is used in \cite{scheidegger2023spectral} who consider an additive model with hidden confounders that can be regarded as a cross-sectional analogue of our panel framework with interactive fixed effects. Our technical arguments, however, differ substantially from theirs. The main reason is this: Similar to our projection approach, \cite{scheidegger2023spectral} use certain data transformations to eliminate the hidden confounders, which correspond to the unobserved factors in our setting. Their transformation device, however, has quite different stochastic properties than our projection device $\widehat{\bs{\Pi}}$ and can thus be decomposed in a way which our device $\widehat{\bs{\Pi}}$ does not allow for. Therefore, their arguments are not applicable to our modeling framework. Notably, our arguments rely on Gaussianity of the covariates but allow for a general covariance structure. If we restrict the covariates to have a quite specific covariance structure, in particular, if the covariates can be grouped into independent blocks of fixed size, it is possible to drop the Gaussianity assumption and to build on more classical proof techniques from \cite{Stone1985} which are employed, e.g., in \cite{HuandHorrowitzAdditive2006}.


%% file: extensions.tex
\section{Extensions}\label{sec:extensions}

\subsection{The nonparametric additive case}

So far, we have restricted attention to the parametric additive case, where the component functions $m_j$ in model \eqref{eq:model-add} can be represented as a linear combination of finitely many (known) functions $\phi_{j1},\ldots,\phi_{jL_j}$. We now discuss how to extend our methods and theory to the more general case where $m_j$ are unknown nonparametric functions that belong to some smoothness class. 

A natural strategy to extend our methodology is to make use of nonparametric series methods: Assuming that the functions $m_j$ belong to a suitable smoothness class, they can be expanded in terms of some basis $\mathcal{B}_j = \{ \phi_{j\ell}: \ell =1,2,\ldots \}$, i.e., they can be represented as $m_j(x) = \sum_{\ell=1}^\infty \beta_{j\ell} \phi_{j\ell}(x)$ with certain coefficients $\beta_{j\ell}$. Truncating this series expansion at order $L_j$ yields
\[ m_j(x) = \phi_j^\top(x)\beta_j + \delta_j(x), \]
where $\phi_j(x) := ( \phi_{j1}(x),\ldots,\phi_{jL_j}(x))^\top$, $\beta_j := (\beta_{j1},\ldots,\beta_{jL_j})^\top$ and the term $\delta_j(x) := \sum_{\ell=L_j+1}^\infty \beta_{j\ell} \phi_{j\ell}(x)$  captures the approximation error due to truncation. Substituting this truncated series expansion back into model equation \eqref{eq:model-add} gives
\[ Y_{it} = \sum_{j=1}^p \beta_j^\top \phi_j(X_{it,j}) + \Delta_{it} + \gamma_i^\top F_t + \varepsilon_{it}, \]
where $\Delta_{it} = \sum_{j=1}^p \delta_j(X_{it,j})$.
This model differs from the parametric version analyzed in the previous sections in two respects: (i) it comprises an additional error term $\Delta_{it}$ and (ii) the numbers $L_j$ are not treated as fixed any more but grow with the sample size, which allows the approximation error $\Delta_{it}$ to vanish asymptotically. Despite these differences, we can run exactly the same estimation strategy as in the parametric case. Hence, our methodology carries over to the nonparametric case without any issues.

The theory, in contrast, does not carry over as straightforwardly. The main difficulty is that our covariates have the form $X_{it} = \bs{\Gamma}_i F_t + Z_{it}$, where the factors $F_t$ are unobserved deterministic quantities that vary over time. Consequently, the distribution of $X_{it}$ is itself time-dependent. In particular, if $X_{it}$ has compact support, this support will generally vary with $t$.
This poses substantial technical challenges for the analysis of the nonparametric additive case. Existing theory for nonparametric series estimators is typically developed under the assumption that the regressors are supported on a fixed, known compact set; see e.g.\ \cite{Stone1982}, \cite{NEWEY1997147}, and \cite{HuandHorrowitzAdditive2006} for spline-based series methods. Such an assumption is incompatible with our framework, where the support of the regressors changes over time through the latent factors.
One possible remedy is to transform the covariates $X_{it}$ to live on a known compact domain which is the same for each $t$. This, though, would require a different (data-driven) transformation for each $t$, which is extremely difficult to handle from a theoretical point of view. 
Another possible approach is to work with a Hermite series expansion \citep[see e.g.][]{dongHermiteexpansion2019,Eftekhari2021INFERENCEHERMITE}. Hermite polynomials are defined on the whole real line and form an orthogonal basis in a Gaussian-weighted \(L^2(\mathbb R)\) space. Hence, they are a natural candidate for unbounded covariates $X_{it}$, especially for Gaussian designs. However, even if we restricted attention to a Gaussian design \(X_{it,j} \sim N(\mu_{tj},\sigma_{tj}^2)\), deriving theory for Hermite polynomials would still be quite challenging. In order to exploit the orthogonality of the Hermite basis, we would have to standardize the regressors $X_{it,j}$ differently for each period $t$ because of their time-varying mean $\mu_{tj}$ and variance $\sigma_{tj}^2$. This in turn would produce Hermite series expansions with time-varying coefficients, which would significantly increase the dimensionality of the estimation problem. Besides, we would have to deal with measurement/estimation error in the standardized covariates as the time-varying mean $\mu_{tj}$ and variance $\sigma_{tj}^2$ are not observed.

The upshot is this: while we conjecture that it is possible to extend our theory to the nonparametric additive case, this is far from trivial, with many complex technical issues arising along the way. 

\subsection{Interactions}\label{sec:extensionInteractions}

In applied econometrics, it is common practice to include interaction terms between the covariates in the model. As an example, we may extend our parametric additive model to include all possible pairwise interactions. The resulting model has the form 
\[ Y_{it} = \sum_{j=1}^p \beta_j^\top \phi_j(X_{it,j}) + \sum_{1 \le j < j' \le p} \vartheta_{jj'} X_{it,j} X_{it,j'} + \gamma_i^\top F_t+ \varepsilon_{it}, \]
where $\vartheta_{jj'}$ are unknown parameters that capture the interaction effects. Stacking the parameters $\vartheta_{jj'}$ in a vector $\vartheta = (\vartheta_{jj'}: 1 \le j < j' \le p)$ and the interactions $X_{it,j} X_{it,j'}$ in a vector $D_{it} = (X_{it,j} X_{it,j'}: 1 \le j < j' \le p)$, we can write the model more compactly as 
\begin{equation}\label{eq:model-interaction} 
Y_{it} = \sum_{j=1}^p \beta_j^\top \phi_j(X_{it,j}) + \vartheta^\top D_{it} + \gamma_i^\top F_t+ \varepsilon_{it}. 
\end{equation}
The estimation algorithm from Section \ref{sec:estimationalgo} extends straightforwardly to model \eqref{eq:model-interaction}: Steps 1 and 2 of the algorithm remain unchanged and the lasso of the parameter vector $(\beta_1^\top,\ldots,\beta_p^\top,\vartheta^\top)^\top$ in Step 3 is defined in the same way as before with the design matrix  
\begin{equation}\label{eq:design-matrix-interactions}
\begin{pmatrix} \widehat{\bs{\Phi}}_1 & \widehat{\bs{D}}_1 \\ \vdots & \vdots \\ \widehat{\bs{\Phi}}_n & \widehat{\bs{D}}_n \end{pmatrix}, 
\end{equation}
where $\widehat{\bs{\Phi}}_i = \widehat{\bs{\Pi}}(\bs{\Phi}_i - \overline{\bs{\Phi}})$ and $\widehat{\bs{D}}_i = \widehat{\bs{\Pi}}(\bs{D}_i - \overline{\bs{D}})$ with $\bs{D}_i= (D_{i1} \ldots D_{iT})^\top$. 
The convergence theory from Section \ref{sec:convergenceresults} that assumes the design matrix to satisfy the RE condition can be adapted without any issues as well. The theory for verifying the RE condition in Section \ref{sec:analysisofcomp}, in contrast, does not carry over to the model with interactions. The reason is that this theory heavily relies on the Gaussianity of the covariates $X_{it,j}$ and thus does not apply to the new covariates $X_{it,j} X_{it,j'}$. Nevertheless, our methods appear to work well in the presence of interactions, as demonstrated by a series of simulation exercises in Section \ref{sec:simulations}.

\subsection{Inference}
\label{sec:inference}

The theoretical results developed in the preceding sections concern coefficient estimation in the parametric additive model. We now turn to statistical inference, which raises additional difficulties.

In the special case where all coefficients associated with nonlinear basis functions are zero, the parametric additive model reduces to the linear model considered in \cite{Mruecker2025CCE}. In this linear setting, a desparsified version of the HD-CCE estimator is asymptotically normal under suitable regularity conditions and can therefore be used to construct confidence intervals and conduct hypothesis tests for individual regression coefficients. Extending this desparsification strategy to our nonlinear setting is, however, far from straightforward. In what follows, we discuss the ensuing difficulties and propose an alternative procedure for inference on individual additive components. The proposed procedure is motivated heuristically, and its finite-sample performance is investigated through Monte Carlo simulations. Its asymptotic validity remains an open theoretical question. The test should therefore be regarded as an exploratory extension rather than a formally justified inferential procedure.\footnote{The proposed procedure is novel even in a simple setting with i.i.d.\ data. We are currently working on establishing its theoretical validity in such an i.i.d.\ setting and plan to extend the analysis to the panel setting considered here in a subsequent step.} 

To formulate the testing problem, we return to the parametric additive model \eqref{eq:model-add}--\eqref{eq:model-reg} analyzed in the main body of the paper. According to this model, 
\[
Y_{it} 
= \sum_{j=1}^p m_j(X_{it,j}) + \gamma_i^\top F_t + \varepsilon_{it},
\]
where each component function has the form $m_j(x) = \phi_j(x)^\top \beta_j$ with $\phi_{j} = (\phi_{j1},\ldots$ $\ldots, \phi_{jL_j})^\top$ a vector of fixed length $L_j$. We are interested in whether the $j$-th covariate has any effect on the response. Since this effect is captured by the function $m_j$, we consider the testing problem 
\[ H_0: m_j = 0 
\quad \text{versus} \quad 
H_1: m_j \ne 0 \]
for a fixed $j$. 

The main difficulty in applying the desparsified lasso approach from \cite{Mruecker2025CCE} to this testing problem concerns the underlying nodewise regressions. Standard theoretical arguments, including those in \cite{Mruecker2025CCE}, typically require these regressions to be (approximately) sparse. While this assumption may be plausible for the original covariates, it becomes much harder to justify once nonlinear transformations are included in the model.
To see why, consider the special case where each component function $m_j$ is quadratic, that is, $m_j(x)=\beta_{j1}x+\beta_{j2}x^2$. The resulting design contains both the original covariates $X_{it,j}$ and their squares $X_{it,j}^2$. Suppose that $X_{it,j}$ admits a sparse nodewise regression of the form 
\begin{equation}\label{eq:model-nodewise-nonlinear-simple}
X_{it,j} = \sum_{j'\neq j} \phi_{j'}(X_{it,j'})^\top \theta_{j'} + F_t^\top \nu_i + u_{it},
\end{equation}
where only a small number of the coefficients in the vectors $\theta_{j'}$ are non-zero. Squaring this equation generates a much more complicated regression structure for $X_{it,j}^2$. In particular, it produces squares and interactions of the original regressors, as well as interactions involving the factors $F_t$ and the error term $u_{it}$. These additional terms need not admit a sparse representation in terms of the regressors included in the model. Consequently, sparsity of the nodewise regression for $X_{it,j}$ does not imply sparsity of the corresponding regression for $X_{it,j}^2$. The same difficulty arises more generally for higher-order polynomial terms and other nonlinear transformations.
A further complication concerns the factor structure. Squaring the nodewise equation \eqref{eq:model-nodewise-nonlinear-simple} also generates nonlinear functions of the factors, such as quadratic terms and interactions between different factor components. Hence, the error term in the nodewise regression for $X_{it,j}^2$ may contain nonlinear factor components. Unlike the original linear factor terms, these nonlinear components are generally not eliminated by the projection matrix $\widehat{\bs{\Pi}}$. As a result, the projected nodewise regression may still contain factor-dependent errors that are correlated with the regressors, giving rise to endogeneity problems.
These difficulties prevent a straightforward application of conventional desparsified lasso theory to our additive model. To address them, we propose an alternative inference procedure that avoids nodewise regressions for the nonlinear transformations and instead relies on a single nodewise regression for the original covariate of interest. We first explain the heuristic idea behind our approach and then give a formal description of the test procedure.

\subsubsection*{Heuristic idea}

Consider the oracle case where the projection matrix $\bs{\Pi}$ from \eqref{eq:oracle-Pi} is known and suppose that the $j$-th covariate satisfies the nodewise regression equation 
\begin{equation}\label{eq:model-nodewise}
X_{it,j} = X_{it,-j}\theta + F_t^\top \nu_i + u_{it}, 
\end{equation}
where $X_{it,-j} = (X_{it,1},\ldots,X_{it,j-1}, X_{it,j+1},\ldots, X_{it,p})^\top$, $\theta \in \reals^{p-1}$ is a sparse para\-meter vector, $\nu_i \in \mathbb{R}^K$ are factor loadings, and $u_{it}$ is a zero-mean error term that is independent of $X_{it,-j}$. The projected errors $\bs{\Pi} u_i$ with $u_i = (u_{i1},\ldots,u_{iT})^\top$ can be estimated by the methods from Section \ref{sec:estimation-method}. For simplicity, we neglect the resulting estimation error and suppose that the random vectors $\bs{\Pi} u_i$ are observed. Let $\widehat{\beta}_\lambda$ be the HD-CCE estimator of $\beta$ and 
\[ R_i = \bs{\Pi} (Y_i - \overline{Y}) - \bs{\Pi} (\bs{\Phi}_{i(-j)} - \overline{\bs{\Phi}}_{(-j)}) \widehat{\beta}_{\lambda,-j} \]
the residual vector under the null in the empirically centred version of our parametric additive model, which can be formulated as $Y_i = \bs{\Phi}_i \beta + \bs{F} \gamma_i + \varepsilon_i$ according to \eqref{eq:model-matrix}. Moreover, let $\tau: \reals \to \reals$ be any smooth, non-negative function which has compact support $[-1,1]$ and integrates up to $1$. Examples are standard kernel functions such as the Epanechnikov and biweight kernel. We consider the normalized version $\tau_{w,h}(u) = \tau( \frac{u-w}{h}) / \sqrt{h}$ for fixed $w \in \reals$ and $h > 0$, which is a small bump around $w$ supported on the interval $[w-h,w+h]$. Letting $\tau_{w,h}(\bs{\Pi} u_i)$ be the component-wise application of $\tau_{w,h}$ to the vector $\bs{\Pi} u_i$, we now have a closer look at the statistic
\begin{align*}
\Psi_{w,h}  
 & := \frac{\sum_{i=1}^n R_i^\top \tau_{w,h}(\bs{\Pi} u_i)}{ \{\sum_{i=1}^n \| \bs{\Pi} \tau_{w,h}(\bs{\Pi} u_i) \|^2\}^{1/2} }. 
\end{align*}
Since $R_i$ has empirically mean zero, the statistic $\Psi_{w,h}$ measures (up to a multiplicative constant) the empirical correlation between $R_i$ and $\tau_{w,h}(\bs{\Pi} u_i)$. Plugging the definition of $R_i$ into the formula for $\Psi_{w,h}$ yields 
\[ \Psi_{w,h} = \Psi_{w,h}^A + \Psi_{w,h}^B + \Psi_{w,h}^C \]
with 
\begin{align*}
\Psi_{w,h}^A & = \frac{\sum_{i=1}^n \{\bs{\Pi} (\bs{\Phi}_{i(j)} - \overline{\bs{\Phi}}_{(j)})\beta_j\}^\top \tau_{w,h}(\bs{\Pi} u_i)}{ \{\sum_{i=1}^n \| \bs{\Pi} \tau_{w,h}(\bs{\Pi} u_i) \|^2\}^{1/2} } \\
\Psi_{w,h}^B & = \frac{(\beta_{-j} - \widehat{\beta}_{\lambda,-j})^\top \sum_{i=1}^n \{ \bs{\Pi} (\bs{\Phi}_{i(-j)} - \overline{\bs{\Phi}}_{(-j)}) \}^\top  \tau_{w,h}(\bs{\Pi} u_i)}{ \{\sum_{i=1}^n \| \bs{\Pi} \tau_{w,h}(\bs{\Pi} u_i) \|^2\}^{1/2} } \\
\Psi_{w,h}^C & = \frac{\sum_{i=1}^n (\bs{\Pi}(\varepsilon_i - \overline{\varepsilon}))^\top \tau_{w,h}(\bs{\Pi} u_i)}{ \{\sum_{i=1}^n \| \bs{\Pi} \tau_{w,h}(\bs{\Pi} u_i) \|^2\}^{1/2} }. 
\end{align*}
Assuming that $u_{i}$ is independent from $\varepsilon_{i}$, it should be possible to show that the third component $\Psi_{w,h}^C$ is asymptotically normal. If $u_{i}$ is assumed to be independent from $\bs{X}_{i(-j)}$ as well, it follows that $\tau_{w,h}(\bs{\Pi} u_i)$ is independent from $\bs{\Pi} (\bs{\Phi}_{i(-j)} - \overline{\bs{\Phi}}_{(-j)})$. Therefore, the sum $\sum_{i=1}^n \{ \bs{\Pi} (\bs{\Phi}_{i(-j)} - \overline{\bs{\Phi}}_{(-j)}) \}^\top  \tau_{w,h}(\bs{\Pi} u_i)$ in the second component  $\Psi_{w,h}^B$ has mean zero. Using this together with standard arguments, one should be able to show that $\Psi_{w,h}^B$ converges to zero in probability, thus being asymptotically negligible. 
Finally, the first component $\Psi_{w,h}^A$ is exactly equal to zero under $H_0$, because $\bs{\Phi}_{i(j)} \beta_j = (m_j(X_{i1,j}),\ldots,m_j(X_{iT,j}))^\top$ and $m_j = 0$ under the null. Taken together, these considerations suggest that the statistic $\Psi_{w,h}$ is asymptotically normal under $H_0$. Under the alternative $H_1$, in contrast, the first component $\Psi_{w,h}^A$ can be expected to be nonzero for certain values of $w$ at least: for different choices of $w$, the vectors $\tau_{w,h}(\bs{\Pi} u_i) \in \reals^T$ can be regarded as different directions in $\reals^T$. By computing the inner products $\sum_{i=1}^n \{\bs{\Pi} (\bs{\Phi}_{i(j)} - \overline{\bs{\Phi}}_{(j)})\beta_j\}^\top \tau_{w,h}(\bs{\Pi} u_i)$ for different points $w$, we try to find a direction in $\reals^T$ that aligns with the vector $\bs{\Pi} (\bs{\Phi}_{i(j)} - \overline{\bs{\Phi}}_{(j)})\beta_j$. Put differently, since $\bs{\Phi}_{i(j)} \beta_j = (m_j(X_{i1,j}),\ldots,m_j(X_{iT,j}))^\top$, we try to find a deviation of the function $m_j$ from $0$. Taken together, these heuristic considerations suggest to test $H_0$ by means of the  statistic 
\[ \Psi := \max_{w \in \mathcal{W}} |\Psi_{w,h}|, \]
where $\mathcal{W}$ is a given set of locations $w$.

\subsubsection*{Test procedure}

We now turn the above heuristics into a feasible test procedure. 
\medskip

\noindent \textit{Construction of projection matrix.} We construct the proxy of the projection matrix $\bs{\Pi}$ slightly differently than for estimation purposes. In particular, we replace the matrix $\overline{\bs{X}} = n^{-1} \sum_{i=1}^n \bs{X}_i$ by $\overline{\bs{X}}_{(-j)}$ which results from eliminating the $j$-th column of $\overline{\bs{X}}$. This helps control certain asymptotic bias terms. Once this replacement is done, the construction proceeds as before: (i) Compute the $(p-1) \times (p-1)$ matrix $\widecheck{\bs{\Sigma}} = \overline{\bs{X}}_{(-j)}^\top \overline{\bs{X}}_{(-j)}/T $ with eigenvalues $\widecheck{\eig}_1 \geq \ldots \geq  \widecheck{\eig}_{p-1} \geq 0$ and corresponding eigenvectors $\widecheck{U}_{(1)}, \ldots, \widecheck{U}_{(p-1)}$. 
(ii) Estimate $\bs{\Pi}$ by
$\widecheck{\bs{\Pi}} = \bs{I} - \widecheck{\bs{W}}(\widecheck{\bs{W}}^\top \widecheck{\bs{W}})^{-}\widecheck{\bs{W}}^\top$,
where $\widecheck{\bs{W}} = \overline{\bs{X}}_{(-j)} \widecheck{\bs{U}}$ and $\widecheck{\bs{U}} = (\widecheck{U}_{1} \ldots \widecheck{U}_{\widehat{K}})$ with $\widehat{K}$ as defined before.
\medskip

\noindent \textit{Construction of nodewise residuals.}
We impose the following nodewise regression equation on the $j$-th covariate: 
\begin{equation}\label{eq:model-nodewise-nonlinear}
X_{it,j} = \sum_{j'\neq j} \phi_{j'}^{\textnormal{nw}} (X_{it,j'})^\top \theta_{j'} + F_t^\top \nu_i + u_{it},
\end{equation}
where $\phi_{j'}^{\textnormal{nw}} = (\phi_{j'1}^{\textnormal{nw}}, \ldots, \phi_{j'L_{j'}^{\textnormal{nw}}}^{\textnormal{nw}})^\top$ is a vector of known functions, $L_{j'}^{\textnormal{nw}}$ is a fixed natural number that does not grow with the sample size, $\theta_{j'} \in \reals^{L_{j'}^{\textnormal{nw}}}$ are unknown parameter vectors, $\nu_i \in \mathbb{R}^K$ are factor loadings, and $u_{it}$ is a zero-mean error term that is independent of $X_{it,-j}$. We assume that the parameter vector $\theta = (\theta_1^\top, \ldots, \theta_{j-1}^\top, \theta_{j+1}^\top, \ldots, \theta_p^\top)^\top$ is sparse in the sense that many of its components $\theta_{j'}$ are exactly equal to zero. Note that in the heuristic discussion above, we have considered the special case with $L_{j'}^{\textnormal{nw}} =1$ and $\phi_{j'1}^{\textnormal{nw}}(x) = x$ for simplicity. The nodewise model \eqref{eq:model-nodewise-nonlinear} can be reformulated in matrix notation as
\begin{equation}\label{eq:model-nodewise-mat}
X_{i(j)} = \bs{\Phi}_i^{\textnormal{nw}} \theta + \bs{F} \nu_i + u_i,
\end{equation}
where $\bs{\Phi}_i^{\textnormal{nw}} = (\Phi^{\textnormal{nw}}(X_{i1}) \ldots \Phi^\textnormal{nw}(X_{iT}))^\top \in \reals^{T\times \sum_{j'\neq j} L_{j'}^{\textnormal{nw}}}$ and $\Phi^\textnormal{nw}(X_{it})$ is the vector with entries $\phi_{j'l}^\textnormal{nw}(X_{it,j'})$ for $j'\neq j$ and $l=1,\ldots,L_{j'}^{\textnormal{nw}}$. We estimate $\theta$ by 
\begin{align*}
\widecheck{\theta}_\kappa \in \argmin_{\vartheta \in\reals^{\sum_{j'\neq j}L_{j'}^{\textnormal{nw}}}}\left\{\frac{1}{nT}\sum_{i=1}^n \big\|\widecheck{X}_{i(j)} -  \widecheck{\bs{\Phi}}_i^{\textnormal{nw}}  \vartheta \big\|^2 + \kappa\|\vartheta \|_1\right\},
\end{align*}
where $\widecheck{X}_{i(j)} = \widecheck{\bs{\Pi}} ({X}_{i(j)} - \overline{X}_{(j)})$ and $\widecheck{\bs{\Phi}}_i^{\textnormal{nw}} = \widecheck{\bs{\Pi}} (\bs{\Phi}_{i}^{\textnormal{nw}}- \overline{\bs{\Phi}}^{\textnormal{nw}})$ are the projected and empirically centered versions of $X_{i(j)}$ and $\bs{\Phi}_{i}^{\textnormal{nw}}$, respectively, and $\kappa$ is the penalty constant of the lasso. With this notation in place, we define
\[\widecheck{u}_i = \widecheck{X}_{i(j)}   -\widecheck{\bs{\Phi}}_i^{\textnormal{nw}}   \widecheck{\theta}_{\kappa}\] 
to be the vector of nodewise residuals.
\medskip

\noindent \textit{Construction of residuals.} We estimate $\beta$ as described in Section \ref{sec:estimation-method}   with $\widehat{\bs{\Pi}}$ replaced by $\widecheck{\bs{\Pi}}$:
\[ \widecheck{\beta}_\pen \in \underset{b \in \reals^{d}}{\text{argmin}} \bigg\{ \frac{1}{nT} \sum_{i=1}^n \big\|\widecheck{Y}_{i} - \widecheck{\bs{\Phi}}_i b \big\|^2 + \pen \|b\|_1 \bigg\}, \]
where $\widecheck{Y}_{i} = \widecheck{\bs{\Pi}} (Y_i - \overline{Y})$ and $\widecheck{\bs{\Phi}}_i= \widecheck{\bs{\Pi}} (\bs{\Phi}_{i}- \overline{\bs{\Phi}})$ are the projected and empirically centered versions of $Y_i$ and $\bs{\Phi}_{i}$, respectively. The vector of residuals can then be expressed as
\[ R_i = \widecheck{Y}_i  - \widecheck{\bs{\Phi}}_{i (-j)} \widecheck{\beta}_{\pen,-j}. \]

\noindent \textit{Definition of test statistic.} 
The test statistic is defined as 
\[ \Psi := \max_{w\in \mathcal{W}} |\Psi_{w,h}| \quad \text{with} \quad \Psi_{w,h} = \frac{\sum_{i=1}^n R_i^\top \tau_{w,h}(\widecheck{u}_i)}{\{ \sum_{i=1}^n \|\widecheck{\bs{\Pi}} \tau_{w,h}(\widecheck{u}_i)\|^2 \}^{1/2}} \]
for each $w \in \mathcal{W}$, where $\mathcal{W} = \{ w \in [-C,C]: w = -C + (2\ell-1) h$ for some positive integer $ \ell \}$ with some pre-specified constant $C$. 
\medskip

\noindent \textit{Test procedure.} 
We propose to reject \(H_0\) at level $\alpha \in (0,1)$ when \(\Psi\geq\widehat c_{1-\alpha}\), where \(\widehat c_{1-\alpha}\) is a Gaussian-coupling approximation to the \((1-\alpha)\)-quantile of \(\Psi\) under \(H_0\) (see Chernozhukov et al., 2013) which we define in detail below. The theoretical justification of this rejection rule, including its asymptotic size and power properties, is left for future research.
By the heuristic arguments from above, under $H_0$, the test statistic $\Psi$ should approximately behave like
\[ \Psi^0 = \max_{w\in \mathcal{W}} |\Psi_{w,h}^0|, \]
where 
\[ \Psi_{w,h}^0 = \frac{\sum_{i=1}^n (\widecheck{\bs{\Pi}}(\varepsilon_i - \overline{\varepsilon}))^\top \tau_{w,h}(\widecheck{u}_i)}{ \{\sum_{i=1}^n \| \widecheck{\bs{\Pi}} \tau_{w,h}(\widecheck{u}_i) \|^2\}^{1/2}}. \]
As the sample averages $\overline{\varepsilon}$ are approximately equal to $0$, we neglect them in what follows and rewrite $\Psi_{w,h}^0$ as
\[ \Psi_{w,h}^0 = \frac{1}{\sqrt{nT}} \sum_{i=1}^n \sum_{t=1}^T d_{it,w} \varepsilon_{it} \]
with
\[ d_{i,w} = (d_{i1,w},\ldots,d_{iT,w})^\top = \frac{\widecheck{\bs{\Pi}} \tau_{w,h}(\widecheck{u}_i)}{ \{(nT)^{-1} \sum_{i=1}^n \| \widecheck{\bs{\Pi}} \tau_{w,h}(\widecheck{u}_i) \|^2\}^{1/2}}. \]
Now let $\{G_{it}: 1 \le i \le n, \, 1 \le t \le T \}$ be a collection of Gaussian random variables with the same mean and covariance structure as the error terms $\{\varepsilon_{it}: 1 \le i \le n, \, 1 \le t \le T \}$. For simplicity, we assume the errors $\varepsilon_{it}$ to be i.i.d.\ both across $i$ and $t$ with mean $0$ and variance $\sigma^2$. Under this assumption, the Gaussian variables $G_{it}$ are i.i.d.\ with mean $0$ and variance $\sigma^2$ as well. In practice, we replace the unknown error variance $\sigma^2$ by an estimator, in particular, by $\widehat{\sigma}^2 = \{n(T-\widehat{K})\}^{-1} \sum_{i=1}^n \| \widecheck{Y}_i - \widecheck{\bs{\Phi}}_i \widecheck{\beta}_{\lambda} \|^2$, which was already proposed in \cite{Mruecker2025CCE}. Using the variables $G_{it}$, we can construct a Gaussian analogue of the statistic $\Psi^0$ as
\[ \Psi^{\textnormal{Gauss}} = \max_{w \in \mathcal{W}} \Big| \frac{1}{\sqrt{nT}} \sum_{i=1}^n \sum_{t=1}^T d_{it,w} G_{it} \Big| \]
and define 
\[ \widehat{c}_{1-\alpha}
= \inf\bigg\{ c>0: \pr\big( \Psi^{\textnormal{Gauss}} \le c \big) \ge 1-\alpha \bigg\}, \]
which can be computed (approximately) by Monte Carlo simulations.

%% file: simulations.tex
\section{Simulation study}\label{sec:simulations}

\subsection{Simulation design}\label{sec:sim-design}

We simulate data from the parametric additive model \eqref{eq:model-add} with $K=3$ latent factors. Building on the simulation design  in \cite{Mruecker2025CCE}, we generate the components of the model as follows:
\begin{itemize}[leftmargin=0.45cm]

\item The error terms $\varepsilon _{it}$ are standard normal draws independent across $i$ and $t$.

\item The factors $F_t = (F_{t,1},F_{t,2},F_{t,3})^\top$ are generated as stationary AR(1) processes with zero means and unit variances. Specifically, for each $k\in \{1,2,3\}$, we let $F_{t,k}=0.5F_{t-1,k}+w_{t,k}$, where the innovations $w_{t,k}$ are $N(0,0.75)$-distributed and independent across $t$ and $k$. By construction, the factors are orthonormal in the sense that $\ex[F_tF_t^\top]=\boldsymbol{I}_{K}$. 

\item The $p=1+3g$ covariates $X_{it}$ are constructed as follows: The first covariate is generated according to the nodewise regression structure 
\begin{equation}\label{eq:nodewise-sim}
X_{it,1} = X_{it,-1} \theta + F_t^\top \nu_i + u_{it}, 
\end{equation}
where $\theta = (\theta^{\textnormal{nw}},0,\ldots,0)^\top$ is a sparse parameter vector whose only non-zero entry is the first element $\theta^{\textnormal{nw}}$ (the value of which is chosen below). Moreover, the variables $u_{it}$ are standard normal draws independent across $i$ and $t$, and we set $\nu_i = 0$ for all $i$ for simplicity. The remaining $3g$ regressors are generated using the factor structure 
\[ X_{it,-1} = \bs{\Gamma}_{i,-1}F_{t}+Z_{it,-1}, \]
where the random vectors $Z_{it,-1}\in \reals^{3g}$ are drawn independently across $i$ and $t$ from a multivariate standard normal distribution $\normal(0,\bs{I})$ and 
\[ \bs{\Gamma}_{i,-1}=\begin{pmatrix} \Gamma^{(1)}_i & 0 & 0 \\ 0 & \Gamma^{(2)}_i & 0 \\ 0 & 0 & \Gamma^{(3)}_i \end{pmatrix} \in \reals ^{3g\times 3}\]
with random vectors $\Gamma _{i}^{(1)}=(\Gamma _{i,1},\dots ,\Gamma _{i,g})^{\top }$, $\Gamma_{i}^{(2)}=(\Gamma _{i,g+1},\dots ,\Gamma _{i,2g})^{\top }$ and $\Gamma_{i}^{(3)}=(\Gamma _{i,2g+1},\dots ,\Gamma _{i,3g})^{\top }$ of length $g$ (which are specified below). 

\item We collect the factor loadings from the outcome and the regressor equations in a large vector $G_{i}=(\gamma _{i}^{\top}$, $\{\Gamma _{i}^{(1)}\}^{\top }$, $\{\Gamma
_{i}^{(2)}\}^{\top }$, $\{\Gamma _{i}^{(3)}\}^{\top })^{\top }$ and draw the random vectors $G_{i}$ independently from a multivariate normal distribution $\normal(\mu ,\bs{\Omega})$. Here, $\mu=(\mu_{\gamma}^\top, 1,\ldots, 1)^\top\in \reals^{K+3g}$ with $\mu_{\gamma}=(3,\ldots,3)^\top\in\reals^K$ and 
\[
\boldsymbol{\Omega }=
\begin{pmatrix}
1 & \rho & \cdots & \rho\\
\rho & \ddots & \ddots & \vdots\\
\vdots & \ddots & \ddots & \rho\\
\rho & \cdots & \rho & 1
\end{pmatrix},
\]
so that $\rho$ governs the pairwise correlation between the factor loadings. We set $\rho = 0.75$ throughout. 

\item We pick $\theta^{\textnormal{nw}}$ such that the covariates $X_{it,j}$ have the same mean and variance for all $j$. The resulting choice is $\theta^{\textnormal{nw}} = \sqrt{2/3}$. (Obviously, $\ex[X_{it,j}]=0$ for all $j$. Moreover, straightforward calculations yield that  $\ex[X_{it,j}^2]=3$ for all $j > 1$ and $\ex[X_{it,1}^2] = 3 (\theta^{\textnormal{nw}})^2 + 1$. Setting $\theta^{\textnormal{nw}} = \sqrt{2/3}$, we thus get that $\ex[X_{it,1}^2] = 3$.)

\item The component functions $m_j$ are chosen as follows: 
\[ \begin{aligned}
m_1(x)&=0.2x+0.1x^2+0.01x^3 \\ 
m_3(x) &= 0.1x^2+0.01x^3 \\
m_5(x) & =0.2x+0.01x^3 
\end{aligned} \qquad 
\begin{aligned}
m_2(x)&= 0.2x \\
m_4(x) &= 0.1x^2 \\
m_6(x) & =0.2x+0.1x^2
\end{aligned} \]
and $m_j = 0$ for $j > 6$. For estimation purposes, we let $\phi_j = (\phi_{j1},\ldots,\phi_{j5})^\top$ with $\phi_{j\ell}(x) = x^\ell$ for $1 \le \ell \le L_j = 5$ and all $j$. We thus consider monomials up to order $5$ for each component $j$. 
Given our choice of the component functions $m_j$, the total number of non-zero coefficients $\beta_{j\ell}$ is $|S| = |\{(j,\ell): \beta_{j\ell} \neq 0 \}| = 11$.

\item We set the cross-section dimension to $n=50$ throughout and consider two time series lengths $T \in \{15,50\}$. The choice of $p$ varies across simulation experiments as detailed below. 

\item All Monte Carlo experiments are based on $1000$ simulation runs.

\end{itemize}

\begin{table}[t]
\caption{Overview of estimators}\label{table:estimators}
\centering
{\small\begin{tabular}{@{\extracolsep{5pt}} llll} 
\toprule\\[-0.5cm]
estimator & label & description \\[0.1cm]
\hline\\[-0.4cm]
$\widehat{\beta}_\lambda$ & L &  centering  + $\bs{\widehat{\Pi}}$ + lasso (see Section \ref{sec:estimationalgo})  \\[0.05cm]
$\widehat{\beta}_{\lambda}^{\text{oracle}}$ & L-O &  centering + $\bs{\Pi}$ + lasso  \\[0.05cm]
$\widehat{\beta}_{\textnormal{LS}}^{\text{double-oracle}}$ & LS-O$^2$  & centering +  $\bs{\Pi}$ + least squares on support $S$ \\[0.05cm]
$\widehat{\beta}_{\lambda}^{\text{naive}}$ & N &  no centering + no projection + lasso
\\[0.1cm]
\bottomrule
\end{tabular}}
\end{table}

\noindent To implement our HD-CCE estimator $\widehat{\beta}_\lambda$, we choose the tuning parameters $\lambda$ and $\tau$ as described in Section \ref{sec:est:tuning:tau}: $\lambda$ is chosen by 10-fold cross-validation, where the folds partition the cross-sectional direction, and $\tau$ is set to $\tau = \alpha \widehat{\eig}_1$ with $\alpha = 0.01$, where $\widehat{\eig}_1$ is the largest eigenvalue of $\widehat{\bs{\Sigma}}$. In what follows, we evaluate the finite sample performance of the HD-CCE estimator $\widehat{\beta}_\lambda$. To do so, we compare it to three benchmarks: 
\begin{itemize}[leftmargin=0.45cm]
\item the ``oracle'' estimator $\widehat{\beta}_\lambda^{\text{oracle}}$ which is computed in exactly the same way as $\widehat{\beta}_\lambda$ except that the proxy $\widehat{\bs{\Pi}}$ is replaced by the oracle matrix $\bs{\Pi}$;  
\item the ``double oracle'' estimator $\widehat{\beta}_{\text{LS}}^{\text{double-oracle}}$ which is obtained by applying least squares to the data projected by the oracle matrix $\bs{\Pi}$, using only the true support $S = \{(j,\ell): \beta_{j\ell} \ne 0\}$;
\item the naive estimator $\widehat{\beta}_\lambda^{\textnormal{naive}}$ which is computed by performing (cross-validated) lasso directly on the unprojected data $\{(Y_i,\bs{\Phi}_i): 1 \le i \le n\}$. 
\end{itemize} 
A brief description of these estimators can be found in Table \ref{table:estimators}. In what follows, we often write $\widehat{\beta}^{(a)}$ to denote the estimator based on the algorithm $a\in \{\textnormal{L, L-O}, \textnormal{LS-O}^2, \textnormal{N}\}$.

\subsection{Parameter estimation}\label{sec:SIM-parameter-estimation}

We consider the model setting described above with $p=601$ covariates. The main model equation has the form 
\[ Y_{it} = m_1(X_{it,1})+\ldots+m_6(X_{it,6}) + F_t^\top\gamma_i + \varepsilon_{it}. \]
Notably, the signal-to-noise ratio in this model is approximately equal to $1$, i.e., 
\[ \textnormal{SNR} :=
\frac{\var\big(\sum_{j=1}^6 m_j(X_{it,j})\big)}
{\var(F_t^\top\gamma_i+\varepsilon_{it})}
\approx 1.
\]
The estimators $\widehat{\beta}^{(a)}$ with $a \in  \{$L, L-O, N$\}$ are computed from the sample
$\{(Y_{it},\Phi_{it}): 1 \le i \le n,\ 1 \le t \le T\}$,
where the design vector $\Phi_{it}\in\mathbb{R}^{pL}$ contains the transformations $\phi_{j\ell}(X_{it,j}) = X_{it,j}^\ell$ ($1 \le \ell \le 5$) for each of the $p=601$ covariates $j$, resulting in an overall dimension of $pL=3005$. The double-oracle estimator $\widehat{\beta}^{(\textnormal{LS-O}^2)}$, in contrast, is computed from the sample points $(Y_{it},\Phi_{it}^*)$, where $\Phi_{it}^* \in \mathbb{R}^{|S|}$ only contains the transformations $\phi_{j\ell}(X_{it,j})$ with non-zero coefficients $\beta_{j\ell}$.

\begin{figure}[p]
\centering
\begin{subfigure}[b]{0.65\textwidth}
\centering
\includegraphics[width = \textwidth]{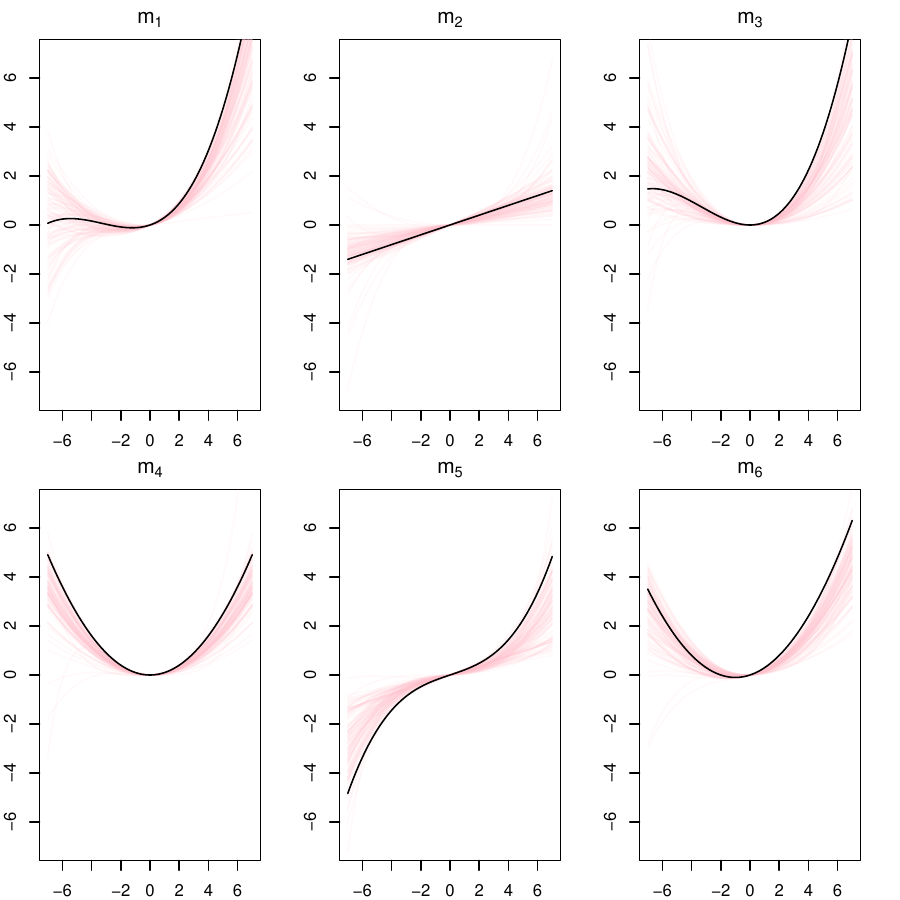}
\vspace{-0.3cm}

\caption{$T = 15$}
\end{subfigure}
\vspace{0.2cm}

\begin{subfigure}[b]{0.65\textwidth}
\centering
\includegraphics[width = \textwidth]{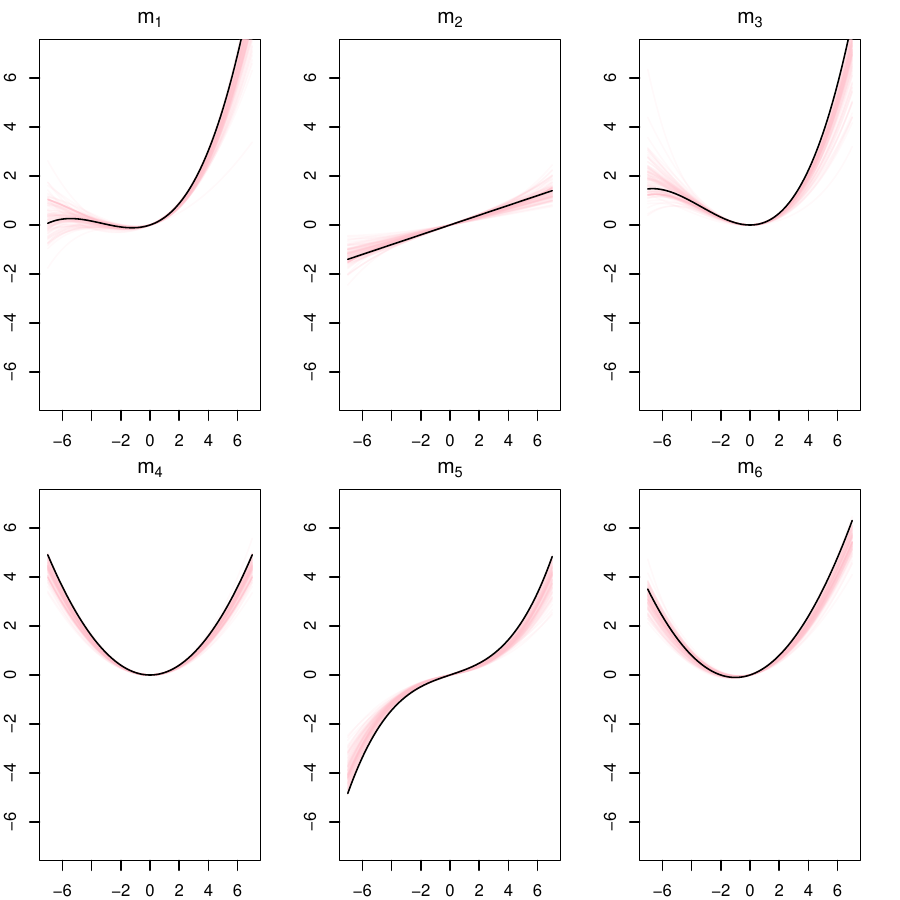}
\vspace{-0.3cm}

\caption{$T=50$}
\end{subfigure}

\caption{HD-CCE estimates of the component functions $m_1, \ldots, m_6$. Each panel shows the estimates $\widehat{m}_j^{(L)} = \widehat{\beta}_{j1}^{(L)} x+ \ldots + \widehat{\beta}_{j5}^{(L)} x^5$ from $100$ simulation runs (pink lines), together with the true regression function $m_j$ (black line), for some $j \in \{1,\ldots,6\}$.}
\label{fig:EstAccuracy3}
\end{figure}

\begin{figure}[p]
\centering
\begin{subfigure}[p]{0.475\textwidth}   
\centering 
\includegraphics[width=\textwidth]{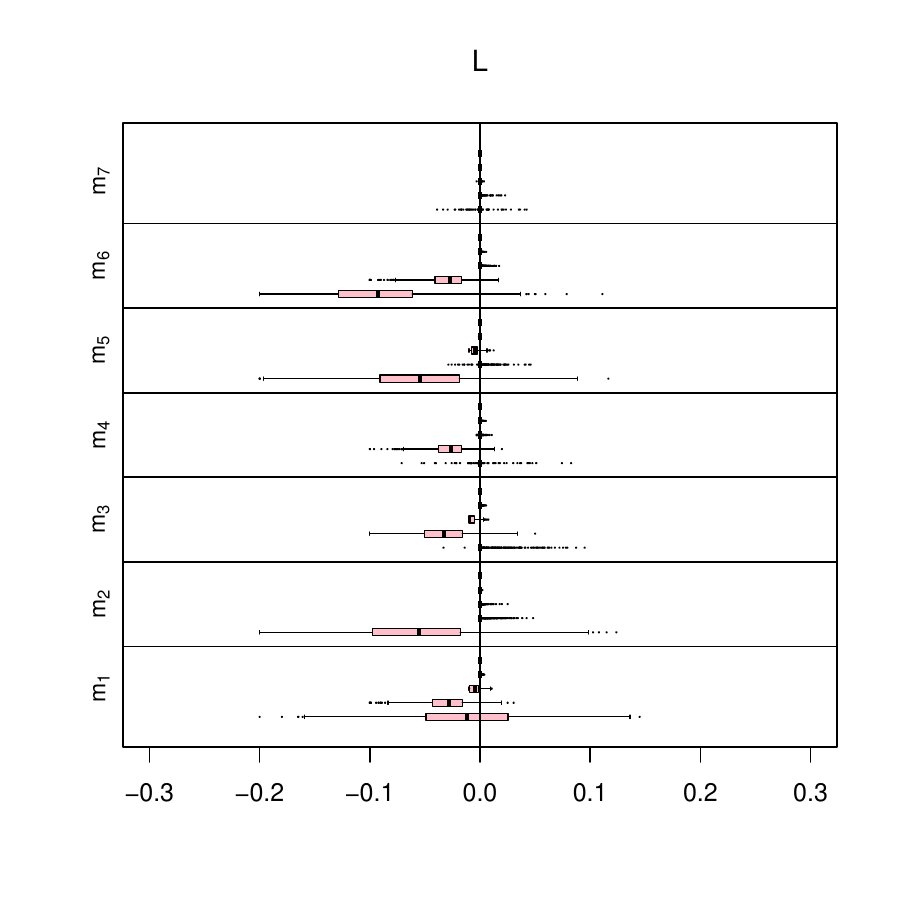} \\[-0.3cm]
\includegraphics[width=\textwidth]{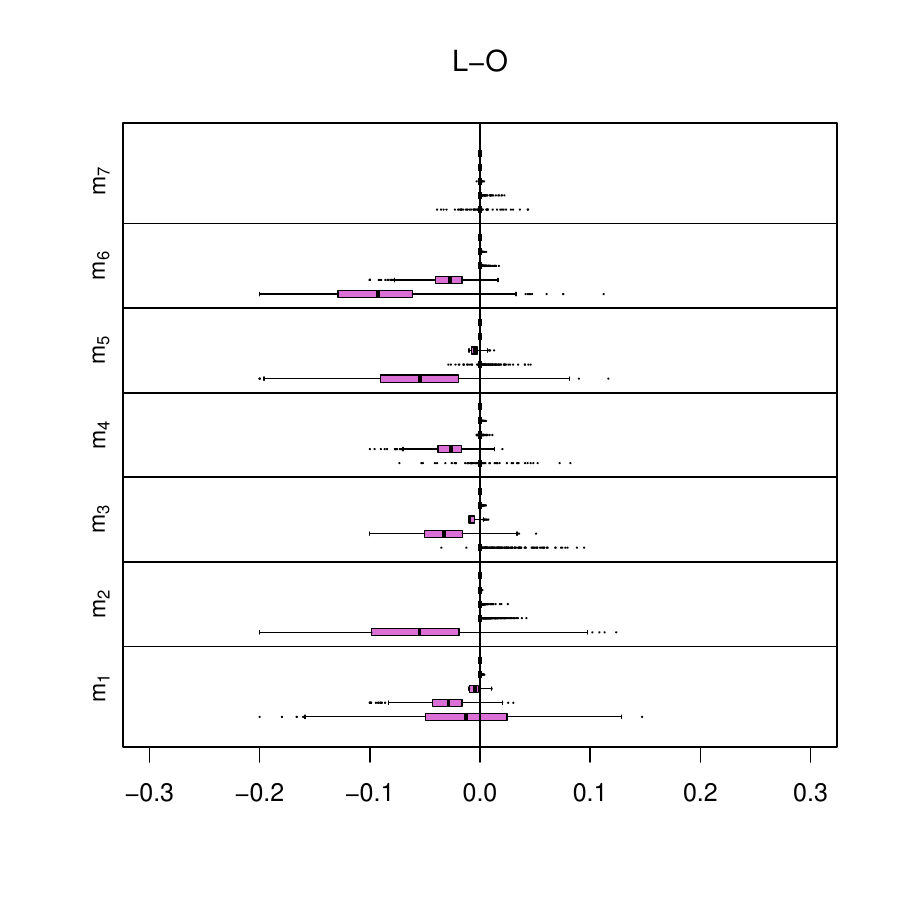} \\[-0.3cm]
\includegraphics[width=\textwidth]{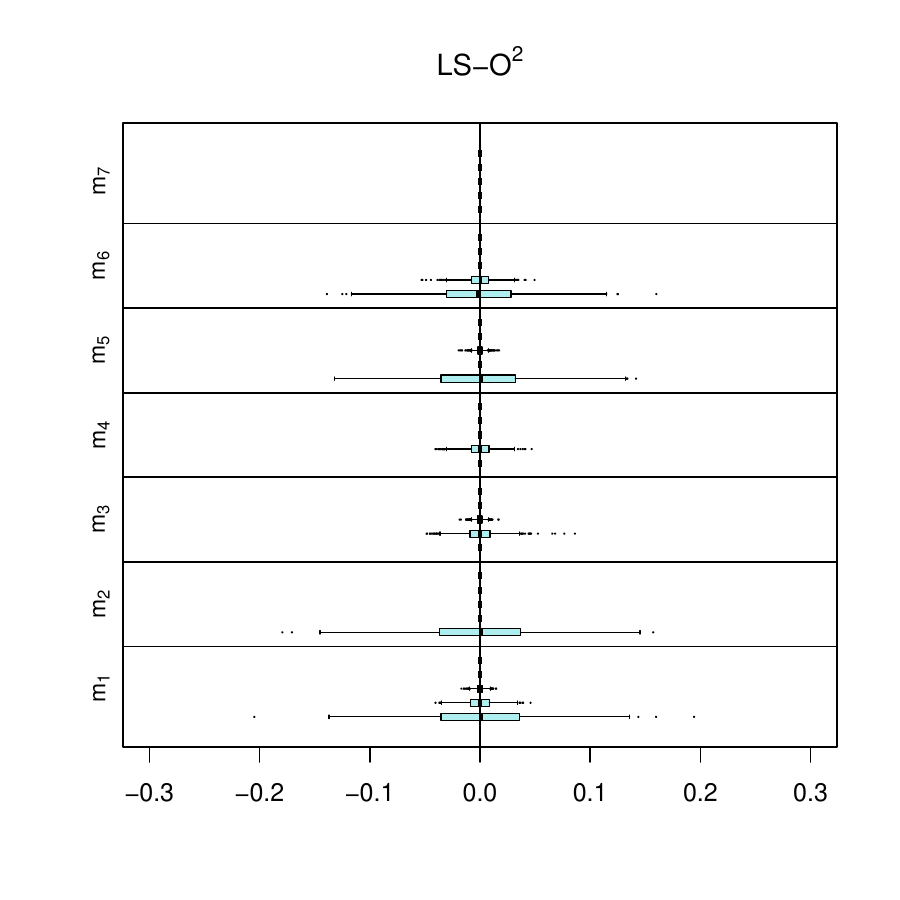} 
\subcaption{$T=15$}
\end{subfigure}
\hspace{0.2cm}
\begin{subfigure}[p]{0.475\textwidth}   
\centering 
\includegraphics[width=\textwidth]{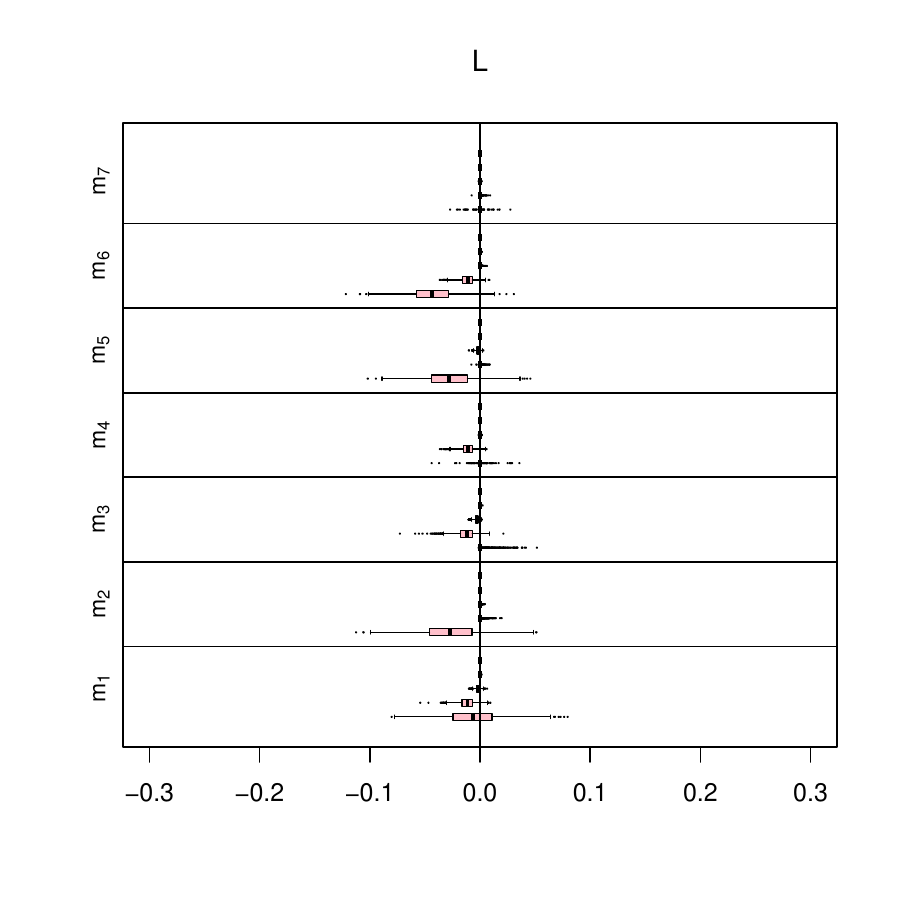} \\[-0.3cm]
\includegraphics[width=\textwidth]{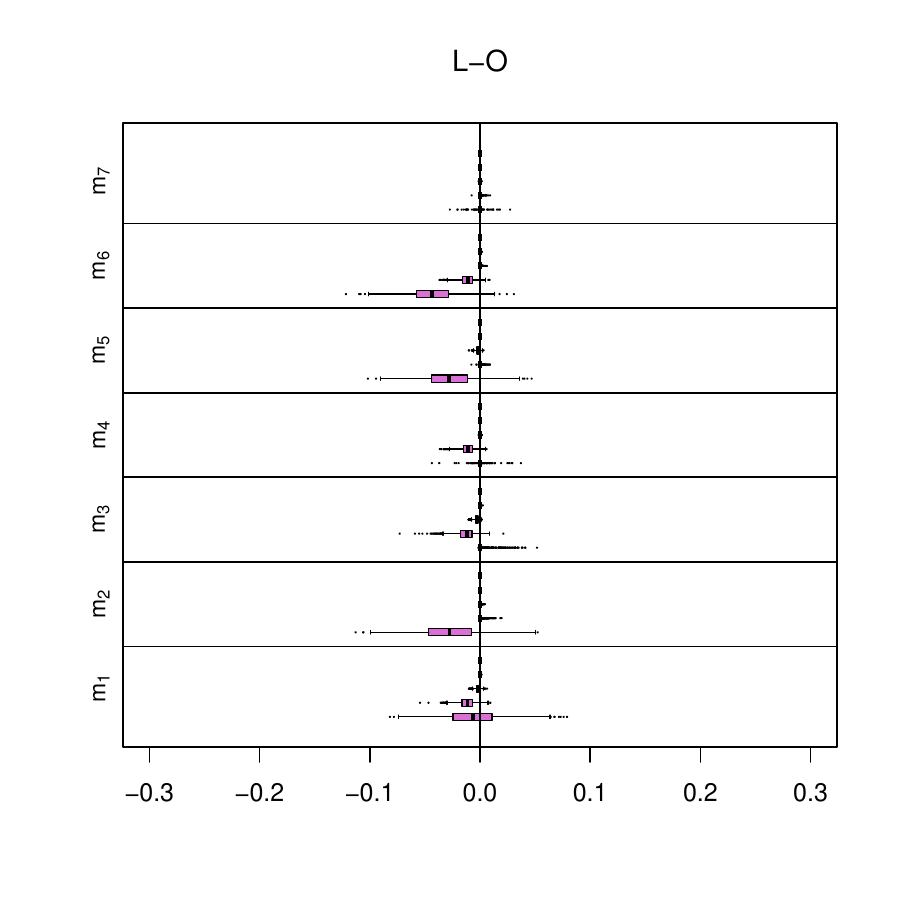} \\[-0.3cm]
\includegraphics[width=\textwidth]{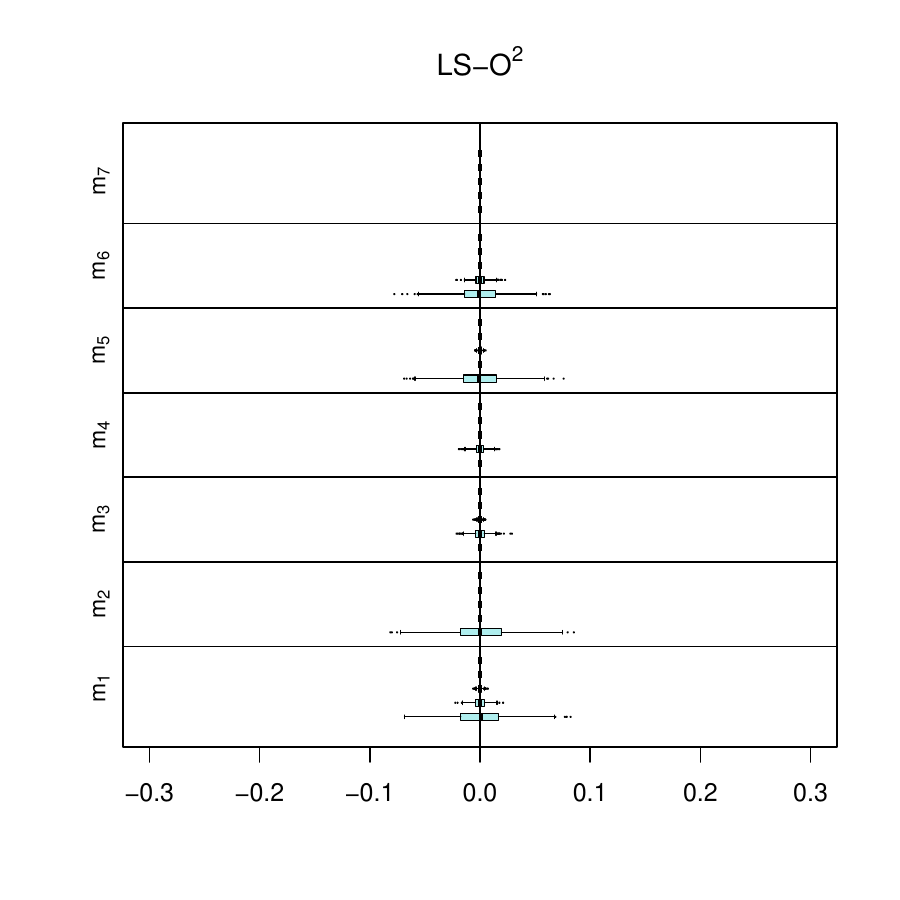}
\subcaption{$T=50$}
\end{subfigure}
\caption{Box plots of the estimation error $\widehat{\beta}_{j\ell}^{(a)}-\beta_{j\ell}$ for $j \in \{1,\ldots, 7\}$, $\ell \in \{1,\ldots, 5\}$ and $a \in  \{$L, L-O, LS-O$^2\}$. (Note: For the double-oracle estimator LS-O$^2$, there are box plots only for index pairs $(j,\ell)$ with $\beta_{j\ell} \ne 0$. To facilitate comparison with the other methods, empty rows corresponding to index pairs $(j,\ell)$ with $\beta_{j\ell} = 0$ have been retained in the two double-oracle panels.)}\label{fig:EstAccuracy1}
\end{figure}

\begin{figure}[t]
\centering
\begin{subfigure}[t]{0.475\textwidth}
\centering
\includegraphics[width = \textwidth]{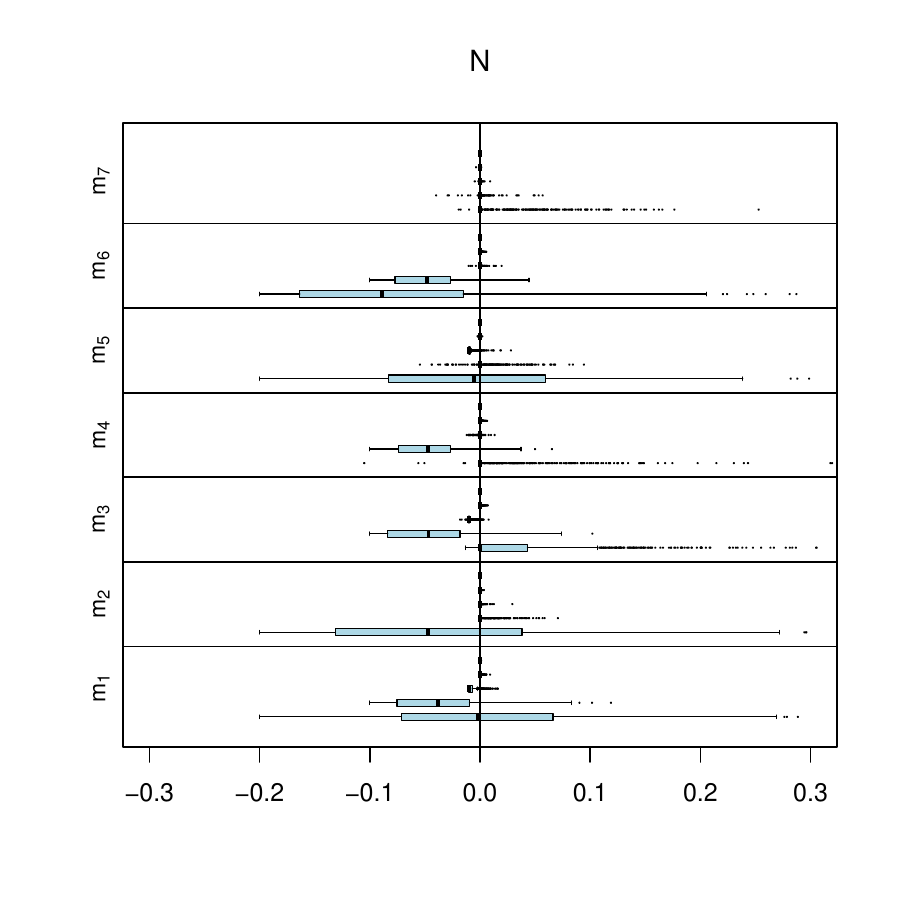}
\vspace{-0.75cm}

\subcaption{$T=15$}
\end{subfigure}
\hspace{0.2cm}
\begin{subfigure}[t]{0.475\textwidth}
\centering
\includegraphics[width = \textwidth]{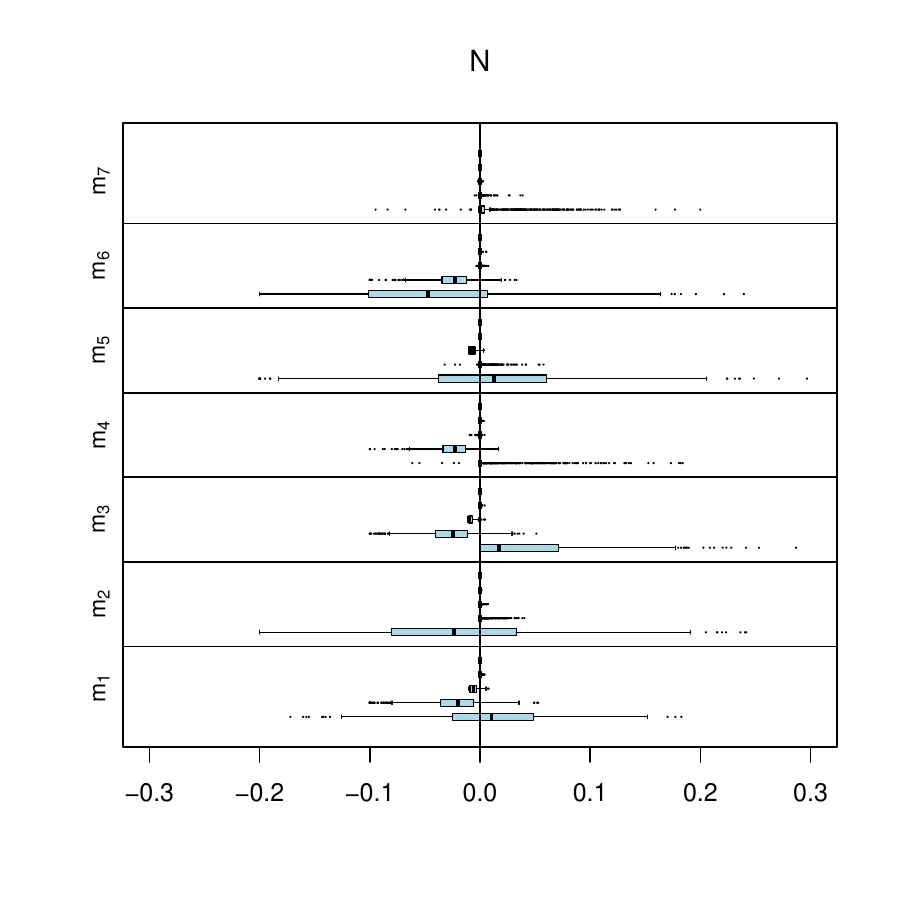}
\vspace{-0.75cm}

\subcaption{$T=50$}
\end{subfigure}

\caption{Box plots of the estimation error $\widehat{\beta}_{j\ell}^{(N)}-\beta_{j\ell}$ for $j \in \{1,\ldots, 7\}$ and $\ell \in \{1,\ldots, 5\}$.}
\label{fig:EstAccuracy1-fourth}
\end{figure}

\begin{figure}[!ht]
\hspace{-0.5cm}
\begin{subfigure}[p]{0.35\textwidth}   
\centering 
\includegraphics[width=\textwidth]{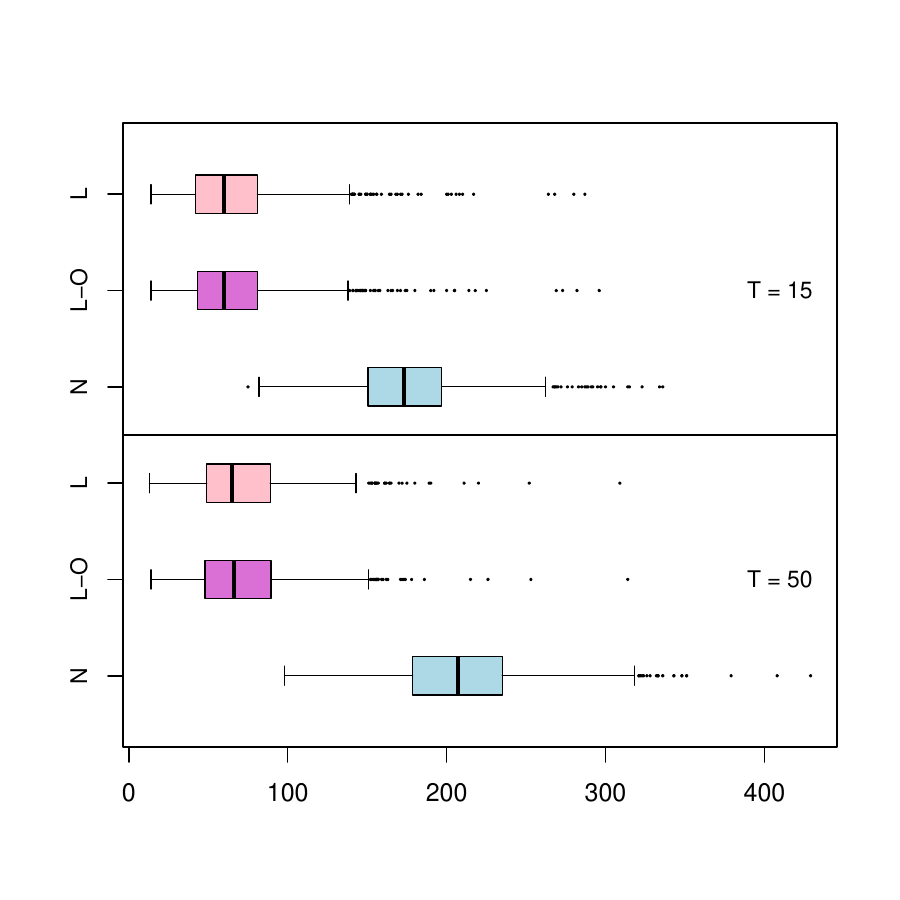} 
\vspace{-0.75cm}

\subcaption{}
\end{subfigure}%
\begin{subfigure}[p]{0.35\textwidth}   
\centering 
\includegraphics[width=\textwidth]{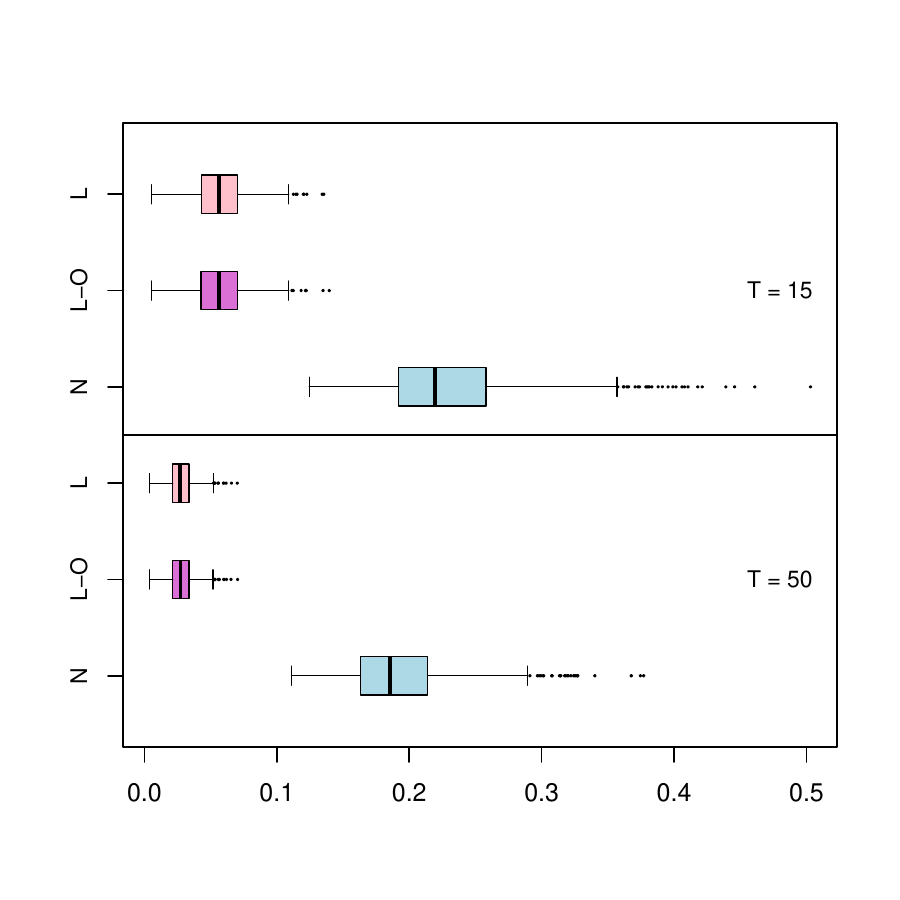} 
\vspace{-0.75cm}

\subcaption{}
\end{subfigure}%
\begin{subfigure}[p]{0.35\textwidth}   
\centering 
\includegraphics[width=\textwidth]{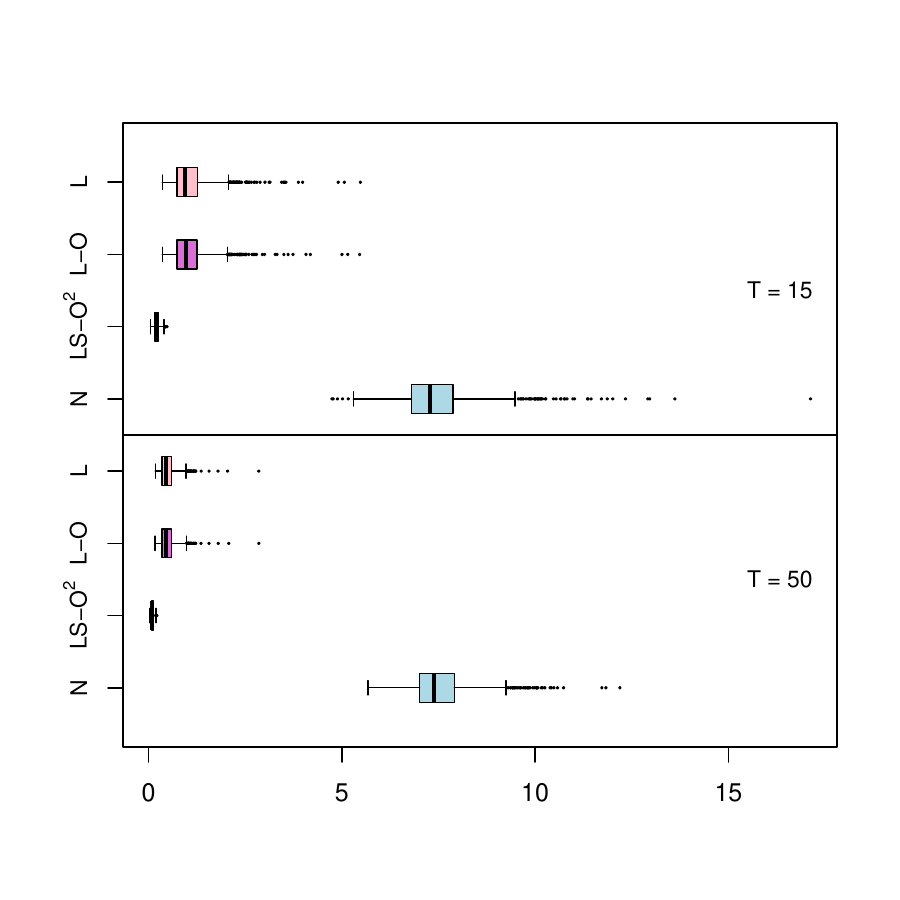}
\vspace{-0.75cm}

\subcaption{}
\end{subfigure}%
\caption{Box plots of (a) the number of false positives $|\textnormal{FP}^{(a)}|$, where $\textnormal{FP}^{(a)}  = \{(j,\ell): \widehat{\beta}_{j\ell}^{(a)}  \neq 0 \} \cap \{(j,\ell): \beta_{j\ell} = 0\}$ is the set of false positives, (b) the maximal size $\max_{(j,\ell) \in \textnormal{FP}^{(a)}}|\widehat{\beta}_{j\ell}^{(a)}|$ of a false positive, and (c) the $\ell_1$-error $\|\widehat{\beta}^{(a)}-\beta\|_1$ for $a \in  \{$L, L-O, LS-O$^2$, N$\}$. The top half of each panel shows the results for $T=15$, the  bottom half those for $T=50$.}\label{fig:EstAccuracy2}
\end{figure}

Figure \ref{fig:EstAccuracy3} presents the estimated regression functions produced by our HD-CCE method for $j \in \{1,\ldots,6\}$. More precisely, each panel displays the HD-CCE estimates $\widehat{m}_j^{(L)}(x)$ $= \widehat{\beta}_{j1}^{(L)}x + \ldots + \widehat{\beta}_{j5}^{(L)}x^5$ from the first $100$ (out of $1000$) simulation runs (pink lines), together with the true regression function $m_j$ (black line), for a different $j \in \{1,\ldots,6\}$. As can be seen, the estimates recover the underlying curves $m_j$ quite well, the estimation accuracy improving noticeably as $T$ increases. 

Figure \ref{fig:EstAccuracy1} depicts box plots of the estimated coefficients $\widehat{\beta}_{j\ell}^{(a)}$ for the methods $a \in  \{\textnormal{L, L-O}, \textnormal{LS-O}^2\}$ over $1000$ simulation runs. Specifically, each panel shows the estimation error $\widehat{\beta}_{j\ell}^{(a)} - \beta_{j\ell}$ for the coefficients of the first seven component functions $m_j$ (from bottom to top). We thus take into account the six component functions $m_1,\ldots,m_6$ with non-zero coefficients as well as $m_7 \equiv 0$ as a representative inactive function. (Note that for the double-oracle estimator, there are box plots only for the non-zero coefficients $\beta_{j\ell}$.)
Most notably, the box plots of our HD-CCE estimator (L) are very similar to those of its oracle counterpart (L-O). This demonstrates that our method performs nearly as well as the oracle, and in particular suggests that the estimated projection matrix $\widehat{\bs{\Pi}}$ is an accurate proxy of $\bs{\Pi}$. Unsurprisingly, the double-oracle estimator (LS-O$^2$) performs slightly better, as evidenced by the tighter box plots. However, our HD-CCE estimator and its oracle version are not far off, exhibiting quite similar performance. The main difference between the least-squares-based double-oracle (LS-O$^2$) and the lasso-based estimators (L and L-O) lies in the bias. Whereas the double-oracle is approximately unbiased, our HD-CCE estimator (L) and its oracle version (L-O) exhibit a substantial shrinkage bias, as is typical for lasso-type estimators. In particular, for zero coefficients $\beta_{j\ell}=0$, the box plots produced by our HD-CCE estimator (L) and its oracle version (L-O) often collapse to zero because the lasso correctly identifies these coefficients as zero in many simulation runs; only a small fraction of replications yield non-zero estimates, which appear as outliers in the plots.

Figure \ref{fig:EstAccuracy1-fourth} presents additional box plots for the naive estimator (N) which show that it performs much worse than our HD-CCE method. In particular, the box plots are much wider, evidencing a substantial loss in accuracy.  
Figure \ref{fig:EstAccuracy2} further reveals that compared to our HD-CCE method, the naive estimator performs poorly in terms of variable selection and $\ell_1$-error. Taken together, Figures \ref{fig:EstAccuracy1-fourth} and \ref{fig:EstAccuracy2} demonstrate that naively ignoring the factor structure and the resulting endogeneity produces flawed estimation results. 

As a final remark, we observe that the performance of our HD-CCE estimator improves in all of the above simulation exercises as we increase the time series length from $T=15$ to $T=50$. This improvement is not reflected in the convergence rate of Theorem \ref{theo:main}, which does essentially not depend on $T$ at all. This suggests that the rate of Theorem \ref{theo:main} is not sharp in the simulation setting at hand.

\subsection{Parameter estimation with interactions}
\label{sec:SIM-parameter-estimationwithinteractions}

\begin{figure}[p]
\centering
\begin{subfigure}[p]{0.475\textwidth}   
\centering 
\includegraphics[width=\textwidth]{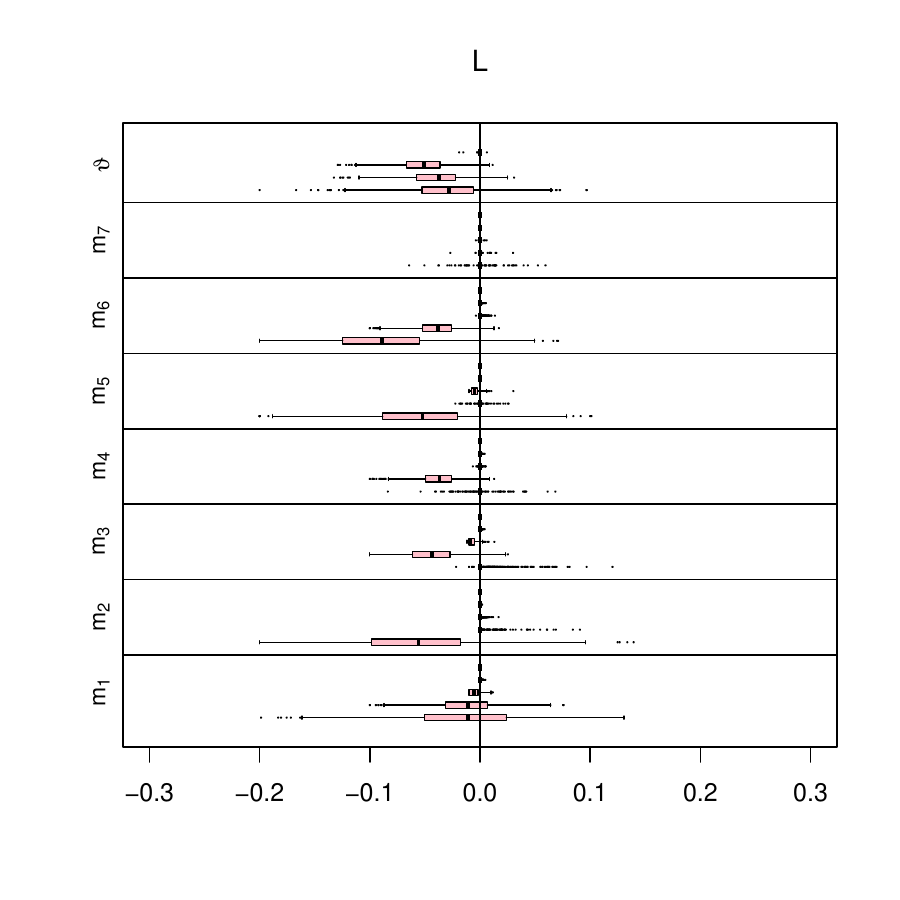} \\[-0.3cm]
\includegraphics[width=\textwidth]{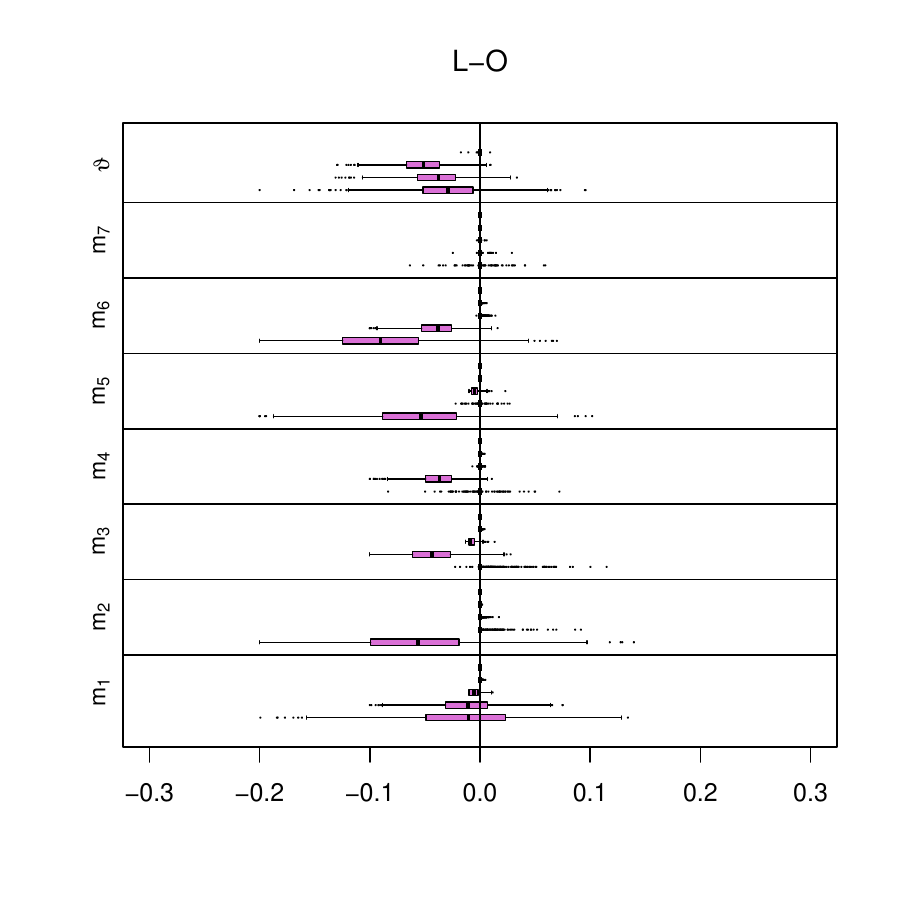} \\[-0.3cm]
\includegraphics[width=\textwidth]{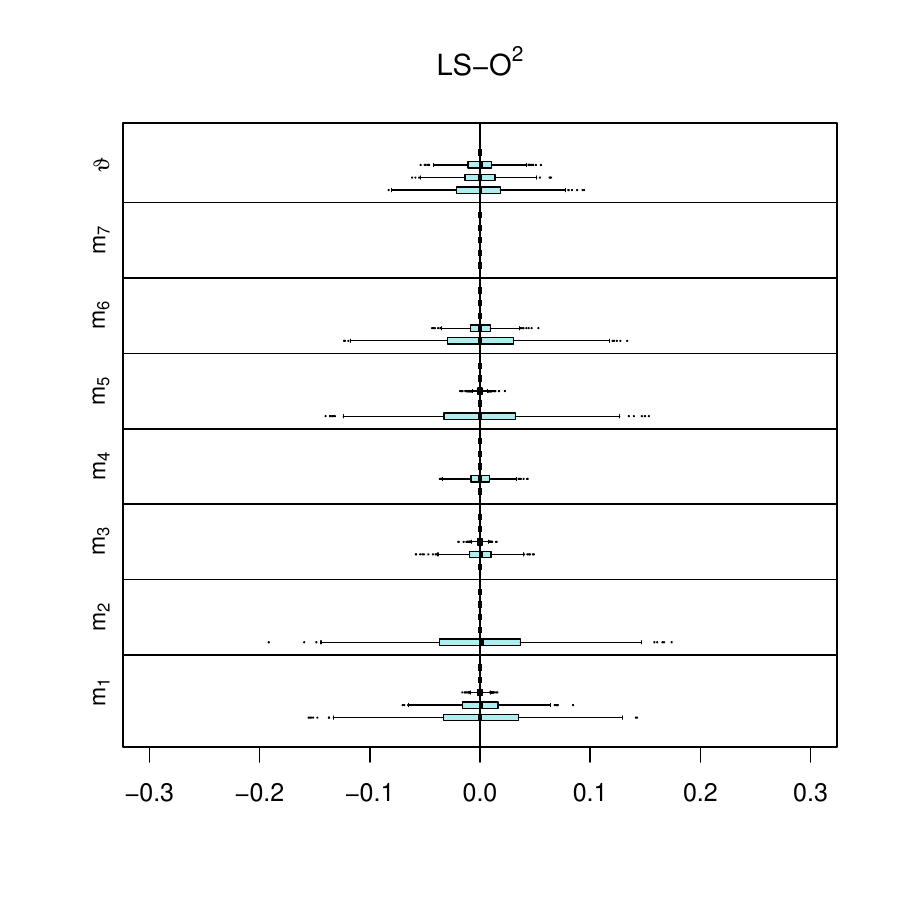}
\subcaption{$T=15$}
\end{subfigure}
\hspace{0.2cm}
\begin{subfigure}[p]{0.475\textwidth}   
\centering 
\includegraphics[width=\textwidth]{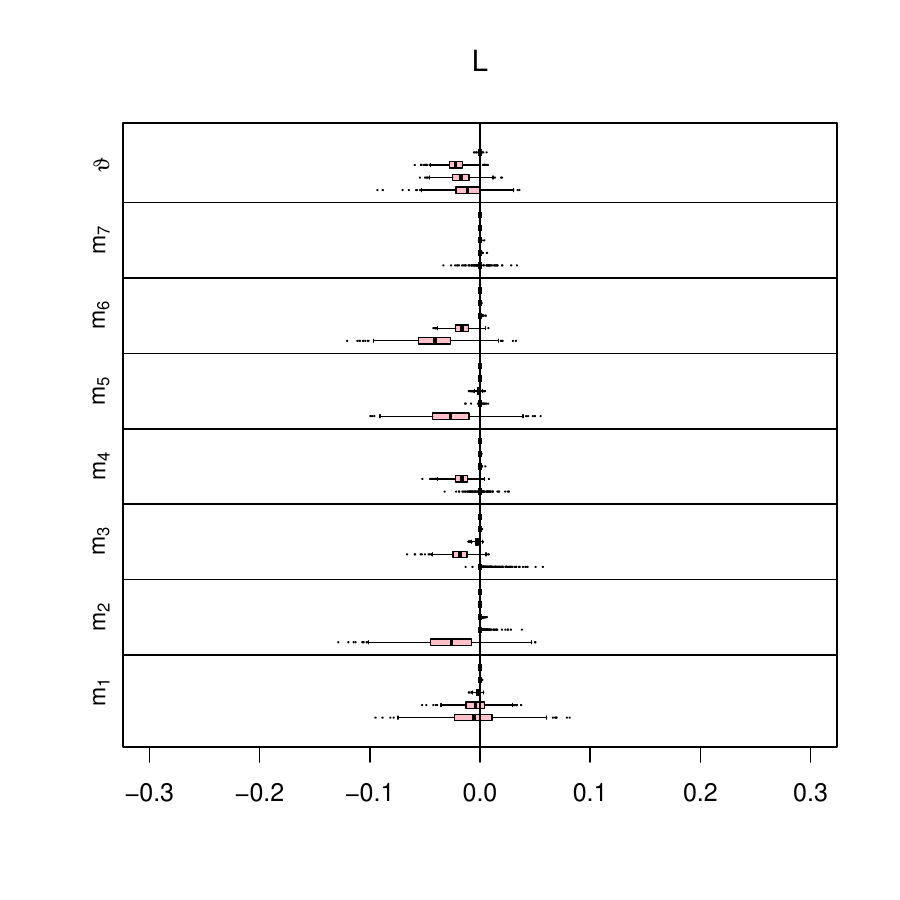} \\[-0.3cm]
\includegraphics[width=\textwidth]{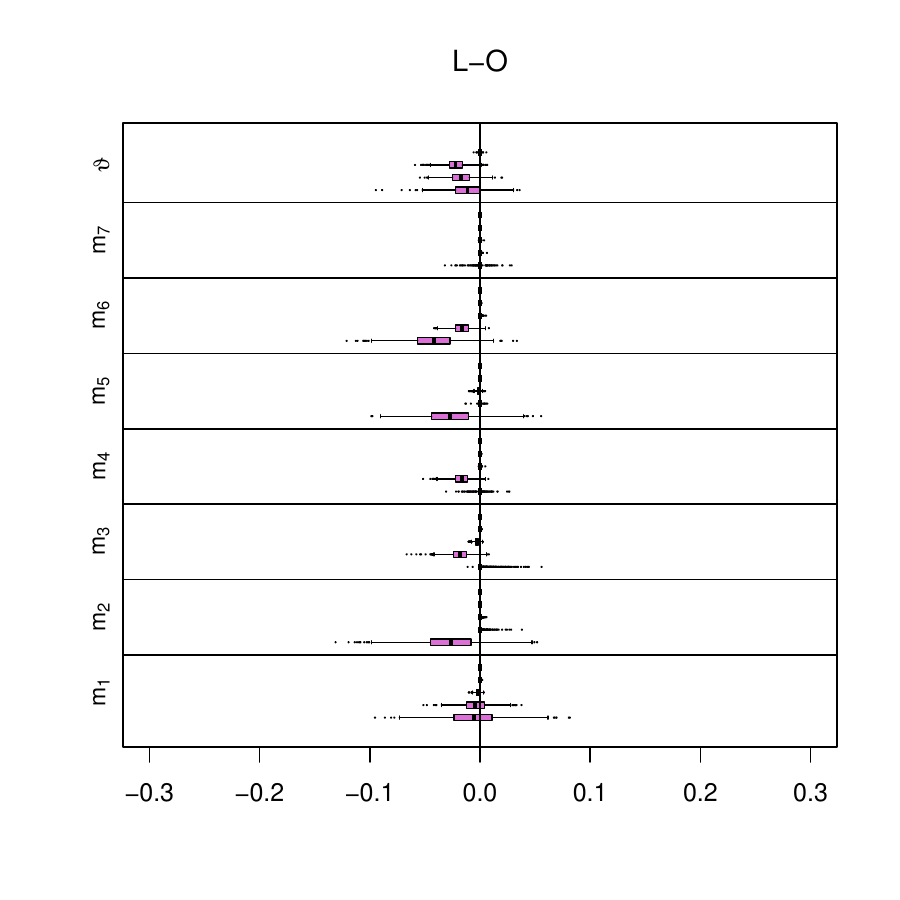} \\[-0.3cm]
\includegraphics[width=\textwidth]{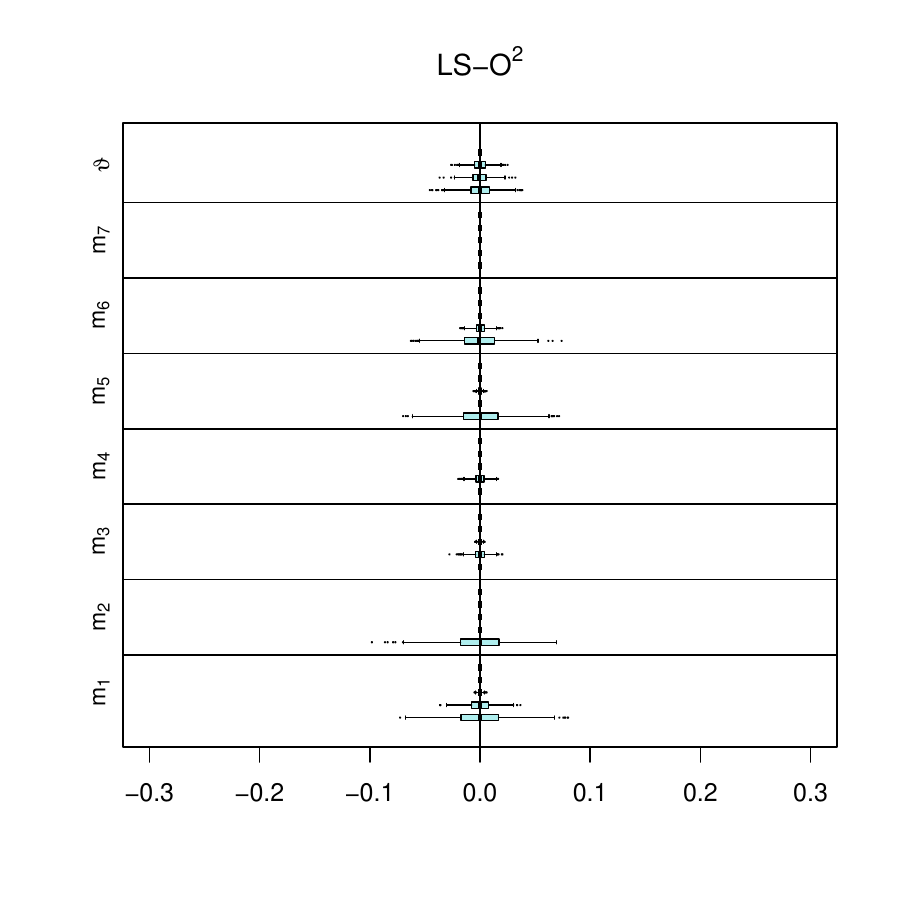}
\subcaption{$T=50$}
\end{subfigure}
\caption{Box plots of the estimation errors $\widehat{\beta}_{j\ell}^{(a)}-\beta_{j\ell}$ and $\widehat{\vartheta}_{k}^{(a)}-\vartheta_{k}$ for $j \in \{1,\ldots, 7\}$,  $\ell \in \{1,\ldots, 5\}$ and $k \in \{\supp{(\vartheta)}, 4095\}$, where $a$ runs over the methods L, L-O and LS-O$^2$. By $\supp(\vartheta)$, we denote the set of indices where $\vartheta$ has a non-zero entry, and $k=4095$ is chosen as a representative index where $\vartheta$ has a zero entry. }
\label{fig:INT1.1}
\end{figure}

\begin{figure}[t]
\centering
\begin{subfigure}[t]{0.475\textwidth}
\centering
\includegraphics[width = \textwidth]{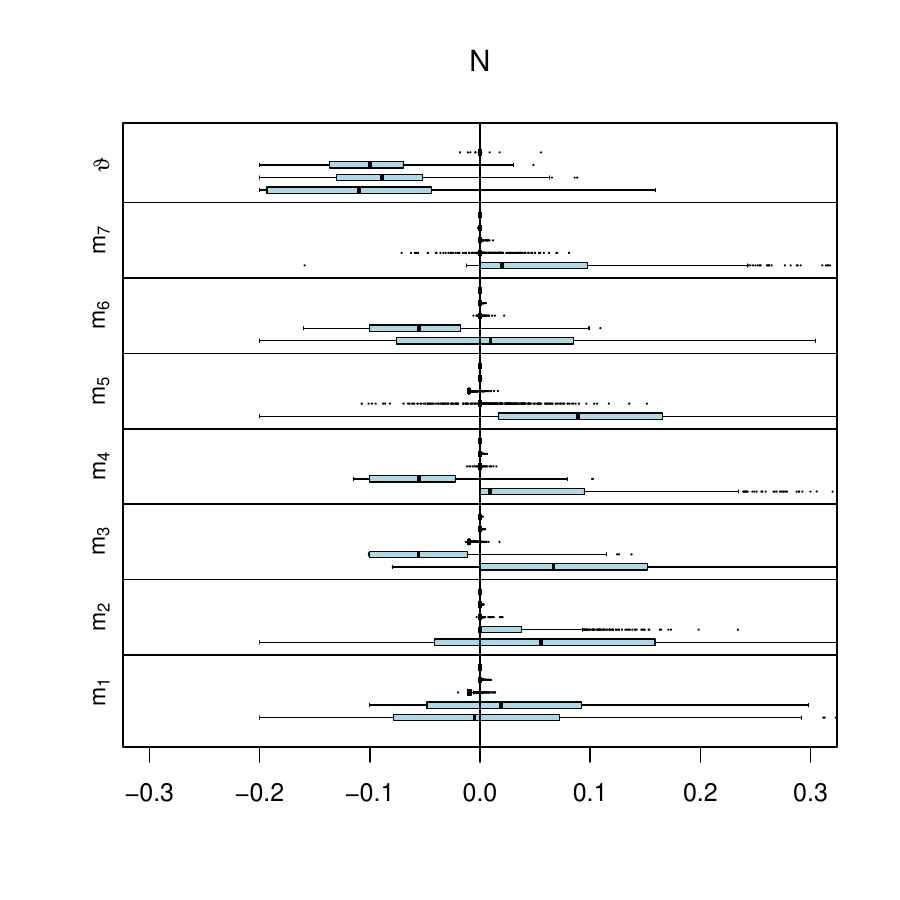}
\vspace{-0.75cm}

\subcaption{$T=15$}
\end{subfigure}
\hspace{0.2cm}
\begin{subfigure}[t]{0.475\textwidth}
\centering
\includegraphics[width = \textwidth]{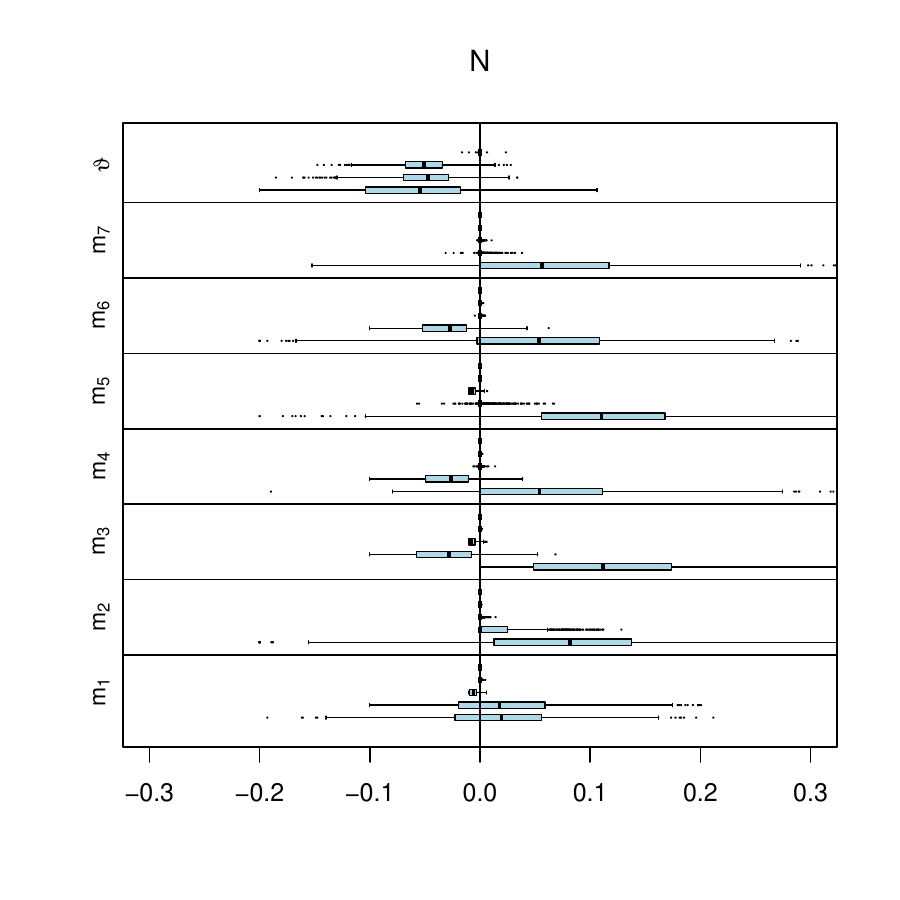}
\vspace{-0.75cm}

\subcaption{$T=50$}
\end{subfigure}

\caption{Box plots of the estimation errors $\widehat{\beta}_{j\ell}^{(N)}-\beta_{j\ell}$ and $\widehat{\vartheta}_{k}^{(N)}-\vartheta_{k}$ for $j \in \{1,\ldots, 7\}$, $\ell \in \{1,\ldots, 5\}$ and $k \in \{\supp{(\vartheta)}, 4095\}$.}
\label{fig:INT1.1-fourth}
\end{figure}

\begin{figure}[!ht]
\hspace{-0.5cm}
\begin{subfigure}[p]{0.35\textwidth}   
\centering 
\includegraphics[width=\textwidth]{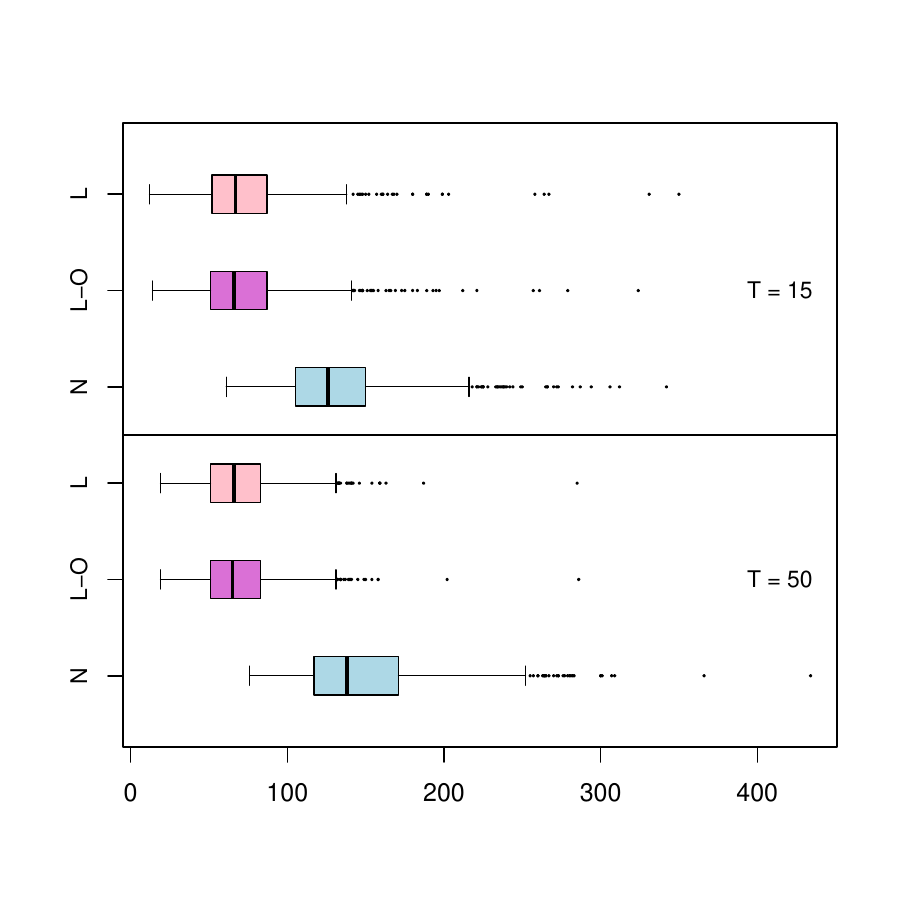} 
\vspace{-0.75cm}

\subcaption{}
\end{subfigure}%
\begin{subfigure}[p]{0.35\textwidth}   
\centering 
\includegraphics[width=\textwidth]{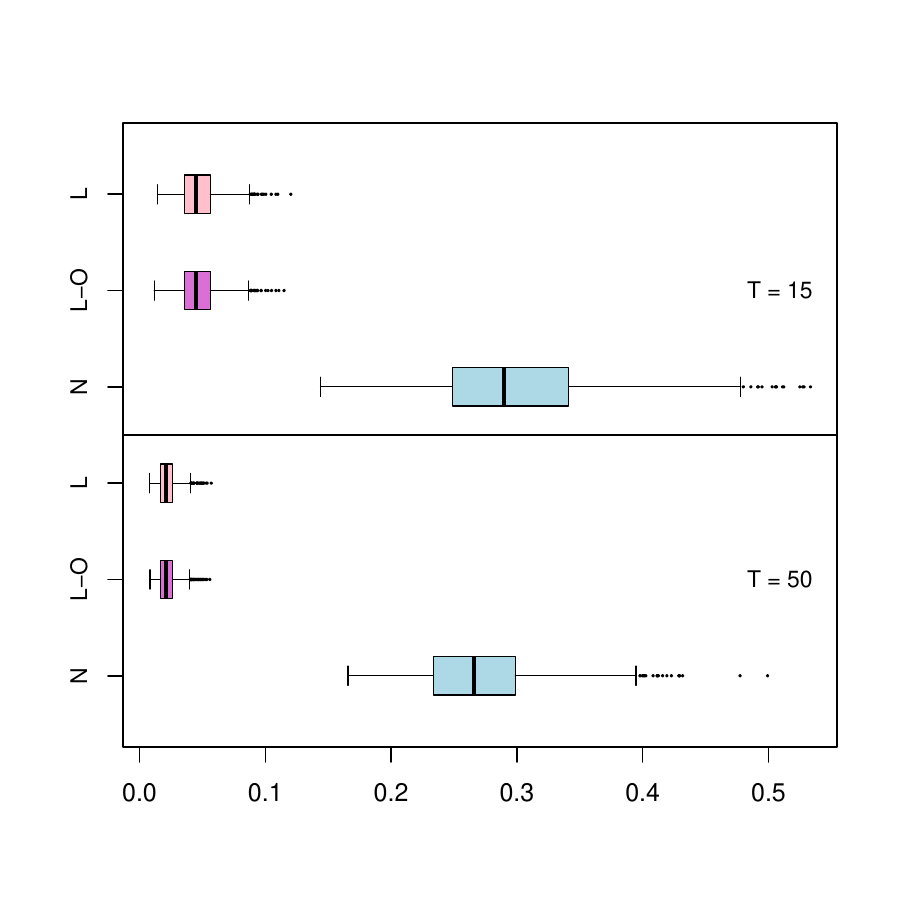} 
\vspace{-0.75cm}

\subcaption{}
\end{subfigure}%
\begin{subfigure}[p]{0.35\textwidth}   
\centering 
\includegraphics[width=\textwidth]{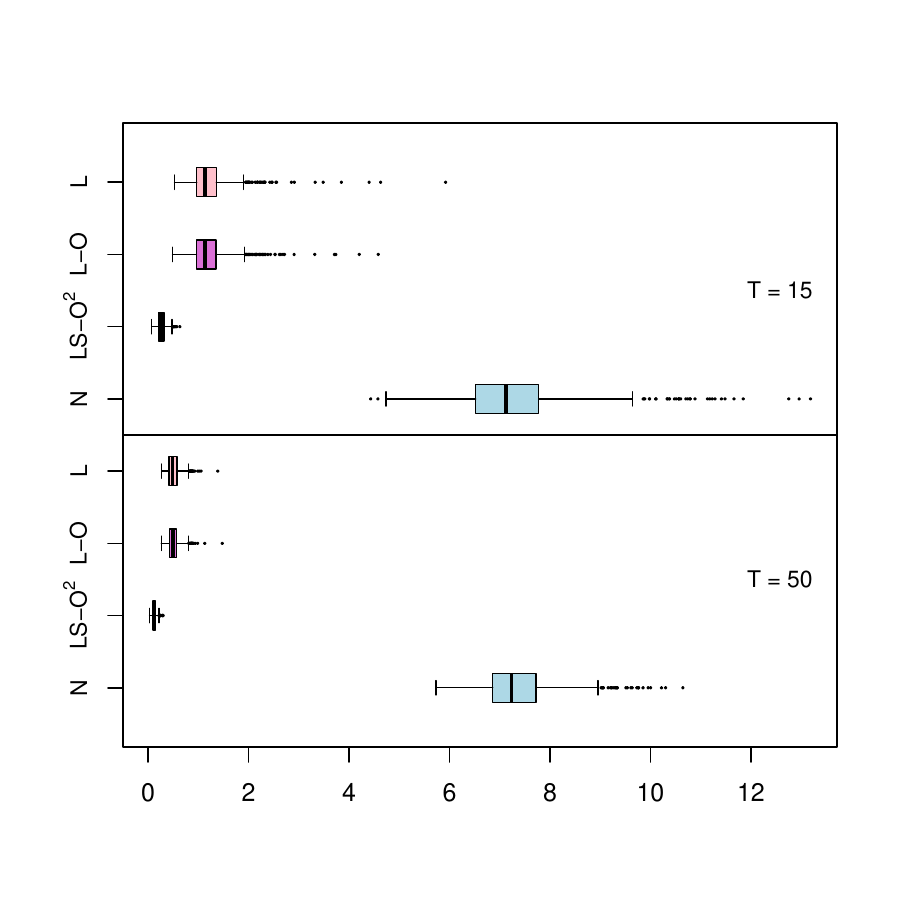}
\vspace{-0.75cm}

\subcaption{}
\end{subfigure}%
\caption{Box plots of (a) the number of false positives, (b) the maximal size of a false positive, and (c) the $\ell_1$-error for the different methods L, L-O, LS-O$^2$ and N. The top half of each panel shows the results for $T=15$, the  bottom half those for $T=50$.}
\label{fig:INT1.2}
\end{figure}

We now investigate how our HD-CCE method performs in the presence of interaction effects. To do so, we include pairwise interaction terms in our simulation setting. In particular, we consider the model  
\[ Y_{it} = m_1(X_{it,1}) + \ldots + m_6(X_{it,6}) + D_{it}^\top \vartheta + F_t^\top \gamma_i + \varepsilon_{it}, \]
where $D_{it}$ contains all pairwise interactions of the covariates $X_{it,j}$, i.e., $D_{it}$ is a vector of length $p(p-1)/2$ with the entries $X_{it,j}X_{it,j'}$ for $1\le j<j'\le p$. 
All other model components are chosen as before. As adding interaction terms increases the dimensionality of the model, we reduce the number of baseline covariates to $p=91$, corresponding to $p=3g+1$ with $g=30$.
The sparse parameter vector $\vartheta$ is chosen as follows: we set the three entries of $\vartheta$ corresponding to the interactions 
$X_{it,1}X_{it,2}$, $X_{it,1}X_{it,g+1}$ and $X_{it,1}X_{it,2g+1}$ to $0.2$, while all other entries of $\vartheta$ are set to zero. The signal-to-noise ratio in the resulting model is approximately equal to 2: 
\[ \textnormal{SNR} =
\frac{\var\big(\sum_{j=1}^6 m_j(X_{it,j}) + D_{it}^\top \vartheta\big)}
{\var(F_t^\top \gamma_i + \varepsilon_{it})}
\approx 2. \]
The full parameter vector to be estimated is $(\beta^\top,\vartheta^\top)^\top$, where $\beta$ is a vector of length $pL = 455$ and $\vartheta$ has length $p(p-1)/2=4095$. The overall dimensionality thus amounts to $4550$. The estimators analyzed in this section are the same as in Section \ref{sec:SIM-parameter-estimation}, where we let $\widehat{\beta}^{(a)}$ and $\widehat{\vartheta}^{(a)}$ denote the estimator of $\beta$ and $\vartheta$, respectively, for $a \in  \{$L, L-O, LS-O$^2$, N$\}$.

We carry out the same simulation exercises as in the previous section. The results are reported in Figures \ref{fig:INT1.1}--\ref{fig:INT1.2} and suggest that our HD-CCE method works well also in the presence of interactions. The main findings are analogous to those from before: our HD-CCE estimator (L) performs very similar to its oracle version (L-O), the performance is not far off from the double-oracle (LS-O$^2$), and the naive estimator (N) is clearly outperformed in all metrics considered.


\subsection{Inference} \label{sec:SIMInference}

\begin{table}[t]
\centering
\caption{Empirical size and power of the test for $T=15.$}
\label{tab:size_power_h_T15}
\resizebox{0.95\textwidth}{!}{%
\begin{tabular}{@{\extracolsep{5pt}} lccccccccc}
\\[-1.8ex]\hline
\hline \\[-1.8ex]
& \multicolumn{3}{c}{\(h=0.2\)} & \multicolumn{3}{c}{\(h=0.3\)} & \multicolumn{3}{c}{\(h=0.4\)} \\[0.3ex]
\cline{2-4}\cline{5-7}\cline{8-10} \\[-1.8ex]
& \(\alpha=0.01\) & \(\alpha=0.05\) & \(\alpha=0.1\) 
& \(\alpha=0.01\) & \(\alpha=0.05\) & \(\alpha=0.1\) 
& \(\alpha=0.01\) & \(\alpha=0.05\) & \(\alpha=0.1\) \\[0.3ex]
\hline \\[-1.8ex]
\(H_0\) &$0.009$ & $0.056$ & $0.108$ & $0.015$ & $0.045$ & $0.092$ & $0.011$ & $0.049$ & $0.097$ \\ 
\(H_1^{(1)}\) & $0.596$ & $0.811$ & $0.872$ & $0.633$ & $0.813$ & $0.876$ & $0.785$ & $0.898$ & $0.943$ \\ 
\(H_1^{(2)}\) & $0.830$ & $0.939$ & $0.969$ & $0.840$ & $0.928$ & $0.970$ & $0.934$ & $0.985$ & $0.992$ \\ 
\hline \\[-1.8ex]
\end{tabular}
}
\vspace{0.4cm}

\caption{Empirical size and power of the test for $T=50.$}
\label{tab:size_power_h_T50}
\resizebox{0.95\textwidth}{!}{%
\begin{tabular}{@{\extracolsep{5pt}} lccccccccc}
\\[-1.8ex]\hline
\hline \\[-1.8ex]
& \multicolumn{3}{c}{\(h=0.2\)} & \multicolumn{3}{c}{\(h=0.3\)} & \multicolumn{3}{c}{\(h=0.4\)} \\[0.3ex]
\cline{2-4}\cline{5-7}\cline{8-10} \\[-1.8ex]
& \(\alpha=0.01\) & \(\alpha=0.05\) & \(\alpha=0.1\) 
& \(\alpha=0.01\) & \(\alpha=0.05\) & \(\alpha=0.1\) 
& \(\alpha=0.01\) & \(\alpha=0.05\) & \(\alpha=0.1\) \\[0.3ex]
\hline \\[-1.8ex]
\(H_0\) & $0.010$ & $0.054$ & $0.106$ & $0.007$ & $0.044$ & $0.078$ & $0.010$ & $0.034$ & $0.082$ \\
\(H_1^{(1)}\) &$0.999$ & $1$ & $1$ & $0.999$ & $1$ & $1$ & $1$ & $1$ & $1$ \\  
\(H_1^{(2)}\) & $1$ & $1$ & $1$ & $1$ & $1$ & $1$ & $1$ & $1$ & $1$ \\
\hline \\[-1.8ex]
\end{tabular}
}
\end{table}

We finally examine the finite sample performance of the test procedure proposed in Section \ref{sec:inference}. We consider the test problem 
\[H_0: m_1 = 0 \quad  \textnormal{versus} \quad H_1: m_1\neq 0 \]
within the simulation design from Section \ref{sec:SIM-parameter-estimation}. To analyze the performance of the test under the null, we set $m_1 = 0$ (but do not change anything else in the design). To investigate the performance under the alternative, we consider two different specifications of $m_1$: $m_1(x) = 0.75 \cdot m(x)$ and $m_1(x)=1 \cdot m(x)$ with $m(x) =  0.2x+0.1x^2+0.01x^3$ (again leaving everything else in the design unchanged). We refer to the first specification as $H_1^{(1)}$ and to the second as $H_1^{(2)}$. 
The test statistic  
\[ \Psi := \max_{w\in \mathcal{W}} |\Psi_{w,h}| \quad \text{with} \quad 
\Psi_{w,h} = \frac{\sum_{i=1}^n R_i^\top \tau_{w,h}(\widecheck{u}_i)}{\{ \sum_{i=1}^n \|\widecheck{\bs{\Pi}} \tau_{w,h}(\widecheck{u}_i)\|^2 \}^{1/2}} \]
is computed with the following choices:
\begin{itemize}[leftmargin=0.45cm]
\item $\mathcal{W} = \{ w \in [-2,2]: w = -2 + (2\ell-1) h$ for some positive integer $ \ell \}$ with three different choices of $h$, in particular, $h \in \{0.2, 0.3, 0.4\}$.
\item $\tau$ is the Epanechnikov kernel: $\tau(x) = \frac{3}{4}(1-x^2)\mathbf{1}\{|x|\leq 1\}$.
\item The regularization parameters $\lambda$ and $\kappa$ of the estimators $\widecheck{\beta}_\lambda$ and estimator $\widecheck{\vartheta}_\kappa$ are computed by 10-fold cross-validation as detailed in Section \ref{sec:est:tuning:tau} combined with the one-standard-error rule (which is implemented as option \texttt{lambda.1se} in the \texttt{R} package \texttt{glmnet}). The threshold parameter $\tau$ is chosen as before, i.e., as described in Section \ref{sec:sim-design}.
\end{itemize}
The test rejects $H_0$ at level $\alpha \in (0,1)$ if $\Psi>\widehat{c}_{1-\alpha}$, where $\widehat{c}_{1-\alpha}$ is computed by Monte Carlo simulations as described in Section \ref{sec:inference}. We set the significance level to $\alpha \in \{0.01,0.05,0.1\}$ throughout.

Tables \ref{tab:size_power_h_T15} and \ref{tab:size_power_h_T50} report the empirical size and power of the test for $T=15$ and $T=50$, respectively. Both empirical size and power are defined as the proportion of simulation runs in which the test rejects the null, with empirical size being computed under the null and empirical power under the alternative.
Tables \ref{tab:size_power_h_T15} and \ref{tab:size_power_h_T50} indicate that the test provides satisfactory size control. Notably, the size numbers are fairly stable across the different choices of the bandwidth $h$. Moreover, the test exhibits substantial power against both alternatives under consideration, the power numbers increasing markedly as we move further away from the null, i.e., from $H_1^{(1)}$ to $H_1^{(2)}$.

%% file: empirical_study.tex
\section{Empirical study}
\label{sec:empiricalstudy}

\begin{table}[b!]
\caption{Ticker symbols and company names of stocks used in the application.}
\label{tab:Stock_list_Application} 
\hspace{0.5cm}
\scriptsize{
\begin{tabular}{p{1.2cm}p{5.25cm}p{1.2cm}p{5.25cm}}
\textbf{Ticker}  &  \textbf{Company} & \textbf{Ticker}  &  \textbf{Company} \\[0.1cm]
\textit{AAPL}  & Apple Inc.  & \textit{INTC}  & Intel Corporation  \\
\textit{AMGN}  & Amgen Inc.  & \textit{JNJ}  & Johnson \& Johnson  \\
\textit{AMZN}  & Amazon.com, Inc. & \textit{JPM}   & JPMorgan Chase \& Co.  \\
\textit{AXP}   & American Express Company & \textit{KO}  & The Coca-Cola Company    \\
\textit{BA}  & The Boeing Company  & \textit{MCD}  & McDonald's Corporation  \\
\textit{CAT}  & Caterpillar Inc.  & \textit{MMM}  & 3M Company  \\
\textit{CRM}   & Salesforce, Inc. & \textit{MSFT}  & Microsoft Corporation   \\
\textit{CSCO}  & Cisco Systems, Inc. & \textit{NKE}   & Nike, Inc.  \\
\textit{CVX}  & Chevron Corporation  & \textit{PFE}  & Pfizer Inc.  \\
\textit{DIS}  & The Walt Disney Company  & \textit{PG}   & The Procter \& Gamble Company  \\
\textit{GE}  & General Electric Company & \textit{TRV}   & The Travelers Companies, Inc.  \\
\textit{GS}   & The Goldman Sachs Group, Inc. &  \textit{UNH} & UnitedHealth Group Incorporated 
\\ \textit{HD}  & The Home Depot, Inc. & \textit{VZ}  & Verizon Communications Inc.  \\
\textit{HON}  & Honeywell International Inc. & \textit{WMT}  & Walmart Inc.  \\	  
\textit{IBM}   & Internat.\ Business Machines Corp. & & \\[0.1cm]
\end{tabular}}
\end{table}

We revisit the empirical application to financial data from \cite{Mruecker2025CCE}. The main objective is to determine which firm characteristics help explain stock returns. We suppose that the monthly excess stock return $R_{it}$ of firm $i$ at time $t$ satisfies the model equation 
\begin{equation}\label{eq:model-app-prelim}
R_{it} = \mu_i + \sum_{j=1}^p \beta_j^\top \phi_j(C_{i,t-1,j}) + \vartheta^\top D_{i,t-1} + \sum_{k=1}^K \gamma_{i,k} F_{t,k} + \varepsilon_{it}
\end{equation}
for $1 \le i \le n$ and $1 \le t \le T$, where the latent factors $F_{t,k}$ capture the comovement of returns, $\mu_i$ is a firm-specific mean, $C_{i,t-1} \in \reals^p$ is a vector of firm characteristics observed at time $t-1$ and  $D_{i,t-1} \in \mathbb{R}^{p(p-1)/2}$ collects all pairwise interactions $C_{i,t-1,j} C_{i,t-1,j'}$ between firm characteristics $j$ and $j'$. We assume each characteristic $j$ to enter the model via a cubic polynomial, i.e., $\phi_j(x) = (\phi_{j1}(x),\phi_{j2}(x),\phi_{j3}(x))^\top$ with $\phi_{j\ell}(x) = x^\ell$ for each $j$. Using the notation $\gamma_{i,0} := \mu_i$ and $F_{t,0} = 1$ for all $t$, model \eqref{eq:model-app-prelim} can be reformulated as
\begin{equation}\label{eq:model-app}
R_{it} = \sum_{j=1}^p \beta_j^\top \phi_j(C_{i,t-1,j}) + \vartheta^\top D_{i,t-1} + \sum_{k=0}^K \gamma_{i,k} F_{t,k} + \varepsilon_{it},
\end{equation}
which is an additive parametric model of the form \eqref{eq:model-add} that has been extended to include interactions. Model \eqref{eq:model-app} generalizes the linear framework from Section 8 in \cite{Mruecker2025CCE} by allowing for (i) quadratic and cubic covariate effects as well as (ii) interaction effects.

\begin{table}[b!]
\scriptsize
\caption{Firm characteristics used in the application. The table contains the acronym and a brief description of the considered characteristics. More details on the variable definitions and constructions can be found in \cite{Green2017}.}
\label{tab:Var_list_Application} 
\begin{tabular}{@{} p{2.1cm} p{4.95cm}}
\textbf{Acronym} & \textbf{Description} \\[0.1cm]
\textit{absacc}  & absolute accruals \\ 
\textit{acc}   & working capital accruals \\ 
\textit{aeavol}   & abnormal earnings announcement vol.\ \\  
\textit{age}  & years since first Compustat coverage \\
\textit{agr}   & asset growth \\
\textit{baspread}   & bid-ask spread \\ 
\textit{beta}   & beta \\
\textit{betasq}   & beta squared \\
\textit{bm}   & book-to-market \\ 
\textit{bm\_ia}  & book-to-market (ind adj) \\ 
\textit{cash}  & cash to assets \\ 
\textit{cashdebt}   & cash flow to debt \\ 
\textit{cashpr}  & cash productivity \\ 
\textit{cfp}  & operating cash flow to price \\  
\textit{cfp\_ia}   & operating cash flow to price (ind adj) \\ 
\textit{chatoia}   & change in asset turnover (ind adj) \\  
\textit{chcsho}  & change in outstanding shares \\  
\textit{chempia}  & change in employees (ind adj) \\ 
\textit{chfeps}  & change in EPS forecast \\ 
\textit{chinv}   & change in inventory \\  
\textit{chnanalyst}  & change in analyst coverage \\ 
\textit{chpmia} & change in profit margin (ind adj) \\ 
\textit{chtx}   & change in taxes \\
\textit{cinvest}   & corporate investment \\ 
\textit{convind}   & convertible debt indicator \\  
\textit{currat}   & current ratio \\ 
\textit{depr}  & depreciation to PP\&{}E \\ 
\textit{disp}   & EPS forecast dispersion \\
\textit{divo}  & dividend omission \\
\textit{dolvol}  & dollar volume \\
\textit{dy}   & dividend yield \\
\textit{egr}   & growth in book value of equity \\
\textit{ep}  & earning to price \\ 
\textit{fgr5yr}  & 5yr EPS growth forecast \\
\textit{gma}  & gross profits to assets \\
\textit{grcapx}   & growth in capital expenditure \\ 
\textit{grltnoa}  & growth in long term net-op.\ assets \\
\textit{herf}  & industry sales concentration \\
\textit{hire}  & employment growth \\
\textit{idiovol}    & idiosyncratic return volatility \\
\textit{indmom}  & industry momentum \\ 
\textit{invest}  & capital expenditures and inventory \\ 
\textit{lev}  & leverage \\
\textit{lgr}  & growth in liabilities \\
\textit{ms} & financial performance score \\  
\textit{mve}  & size  \\
\end{tabular}
\begin{tabular}{p{2.1cm} p{4.95cm}}
\textbf{Acronym} & \textbf{Description} \\[0.1cm]
\textit{mve\_ia}  & size (ind adj)  \\
\textit{nanalyst}   & analyst coverage  \\
\textit{nincr}  & length of earnings run   \\
\textit{operprof}   & operating profits / book equity   \\
\textit{orgcap}  & organizational capital  \\
\textit{pchcapx\_ia}   & \% change in capital exp. (ind adj)   \\
\textit{pchcurrat}  & \% change in current ratio  \\
\textit{pchdepr}   & \% change in depreciation to PP\&{}E   \\
\textit{pchgm\_pchsale}  & \% change in sales \\
\textit{pchquick} & \% change in quick ratio \\
\textit{pchsale\_pchinvt}  & \% change in inventory   \\
\textit{pchsale\_pchrect} & \% change in receivables    \\
\textit{pchsale\_pchxsga}  & \% change in sales growth \\ 
& - \% change in overheads \\
\textit{pchsaleinv} & \% change in sales-to-inventory   \\
\textit{pctacc}  & percent accruals  \\
\textit{pricedelay}  & price delay \\
\textit{ps}   & financial health score  \\
\textit{quick}   & quick ratio  \\
\textit{rd}  & large R\&{}D increase  \\
\textit{rd\_mve}  & R\&{}D to size   \\
\textit{rd\_sale}  & R\&{}D to sales  \\
\textit{realestate}   & real estate holdings  \\
\textit{roaq}   & return on assets  \\
\textit{roavol}  & earnings volatility \\
\textit{roeq}  & return on equity  \\
\textit{roic}   & return on invested capital  \\
\textit{rsup}  & revenue surprise  \\
\textit{salecash}   & sales to cash  \\
\textit{saleinv}  & sales to inventory  \\
\textit{salerec}   & sales to receivables  \\
\textit{secured}   & secured debt  \\
\textit{securedind}  & secured debt indicator \\
\textit{sfe}  & scaled earnings forecast  \\
\textit{sgr}   & sales growth  \\
\textit{sp}  & sales to price  \\
\textit{std\_dolvol} & volume volatility  \\
\textit{std\_turn}  & turnover volatility  \\
\textit{stdacc}  & accrual volatility \\
\textit{stdcf}   & cash flow volatility  \\
\textit{sue}  & surprise earnings  \\
\textit{tang}   & asset tangibility  \\
\textit{tb} & tax to book income  \\
\textit{turn}  & turnover  \\
\textit{zerotrade}  & (weighted) days with zero trades  \\
 & \\
\end{tabular} 
\end{table}

We apply model \eqref{eq:model-app} to a sample of large-cap stocks {($n=29$)} from April 2017 to March 2022 {($T=60$)} that are or recently were constituent stocks of the Dow Jones Industrial Average. The names and ticker symbols of the stocks are given in Table \ref{tab:Stock_list_Application}. The firm characteristics ($p=90$) used in the application including a brief description are listed in Table \ref{tab:Var_list_Application}. They comprise a subset of the characteristics collected by \cite{Green2017} for the period 1980--2014, which have been extended to 2022 by Shaoran Li, who kindly shared the data with us. We refer to Section 8 in \cite{Mruecker2025CCE} for more details on the data.

As our model incorporates quadratic and cubic effects for all firm characteristics as well as all possible pairwise interactions, we have a setting with $3p + p(p-1)/2 = 4275$ unknown parameters to be estimated, while the total sample size amounts to $nT = 1740$. In what follows, we estimate the unknown parameters $\beta_{jl}$ and $\vartheta_h$ by our methods. The tuning parameter $\lambda$ is chosen by a leave-one-firm-out version of the cross-validation procedure explained in Section \ref{sec:est:tuning:lambda}. The thresholding parameter $\tau$ is set to $\tau = \alpha \widehat{\eig}_1$ with $\alpha = 0.01$, where $\widehat{\eig}_1$ is the largest eigenvalue of $\widehat{\bs{\Sigma}}$ (see Section \ref{sec:est:tuning:tau}). This results in the estimate $\widehat{K}=2$. The parameter estimates are denoted by $\widehat{\beta}_{jl}$ and $\widehat{\vartheta}_{h}$. Notably, the firm characteristics $C_{i,t-1}$ are lagged by one time period, which requires some slight modifications of the projection matrix $\widehat{\bs{\Pi}}$ as explained in Section 8 of \cite{Mruecker2025CCE}. With these modifications, we can run our methods exactly as laid out in the previous sections. 
For better interpretability, we report the scaled coefficient estimates $\widehat{\beta}_{jl}^{\text{sc}} := \widehat{\beta}_{jl} \cdot \widehat{s}_{jl}$ and $\widehat{\vartheta}_{h}^{\textnormal{sc}}:= \widehat{\vartheta}_{h} \cdot  \widehat{s}_{h} $, where $\widehat{s}_{jl}^2$ and $\widehat{s}_{h}^2$ are the empirical variances of the $(j,l)$-th regressor $\phi_{jl}(C_{i,t-1,j})$ and $h$-th interaction $D_{i,t-1}$. We thus normalize the design to have empirical variance $1$ and report the coefficient estimates on the resulting scale.

\begin{figure}[t!]
    \centering
    \includegraphics[width=0.95\linewidth]{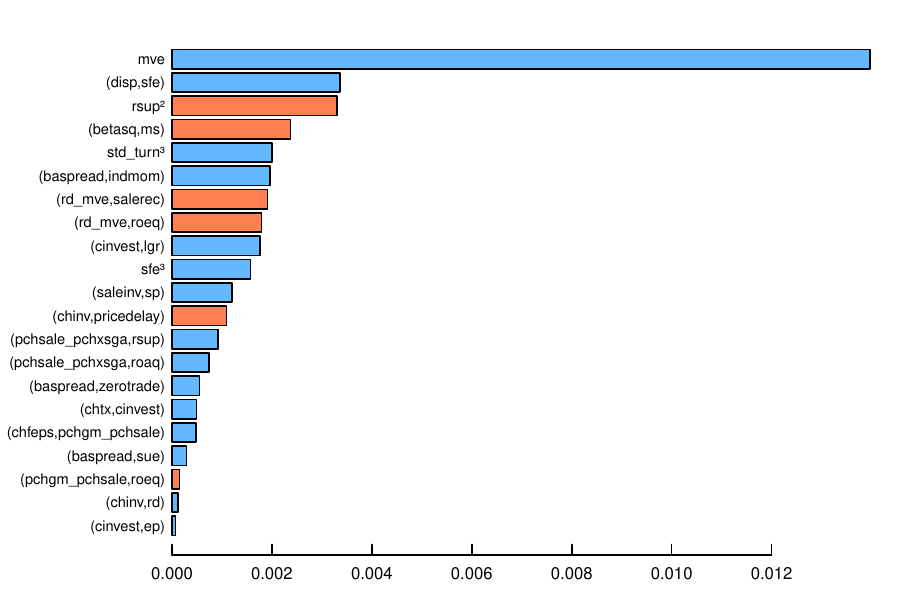}
    \caption{Scaled coefficient estimates $\widehat{\beta}_{jl}^{\textnormal{sc}}$ and $\widehat{\vartheta}_{h}^{\textnormal{sc}}$ ordered by absolute size. The colour of the bars indicates the sign (red = positive sign, blue = negative sign).}
    \label{fig:empApp1}
\end{figure}

\begin{table}[t!]
{\small
\centering
\begin{tabular}{cccc}
\hline \hline \\[-0.4cm]
effect type & total number & selected (non-zero) &\% selected \\[0.1cm]
\hline \\[-0.4cm]
main effect $\beta_{jl}$ & $ pL= 270$ & $4$ & $1.48\%$ \\
interaction $\vartheta_{h}$ &$p(p-1)/2= 4005$ & $17$ & $0.42\%$ \\[0.1cm]
\hline
\end{tabular}
\caption{Summary of estimated sparsity pattern.}\label{table:empAPP}}
\end{table}

Figure \ref{fig:empApp1} displays all nonzero coefficient estimates $\widehat{\beta}_{jl}^{\text{sc}}$ and $\widehat{\vartheta}_{h}^{\textnormal{sc}}$ ordered by absolute magnitude. The estimated coefficient vector is highly sparse: only $21$ out of $4275$ coefficients are nonzero. Table \ref{table:empAPP} summarizes the resulting sparsity pattern. The parameter estimates suggest that only a small subset of firm characteristics contributes to explaining stock returns.
This finding is consistent with recent evidence from the empirical asset-pricing literature. \cite{Green2017} show that, although numerous firm characteristics have been proposed as return predictors in the literature, only a relatively small subset provides independent information about average monthly stock returns once the characteristics are considered jointly. Furthermore, the selection of nonlinear transformations and interaction terms accords with the machine-learning evidence in \cite{Gu2020}, who document that nonlinearities and interactions play an important role in return prediction. Closely related to our empirical application, \cite{Freyberger2020} model expected returns as a nonparametric additive function of firm characteristics, estimated using the adaptive group lasso. Out of $62$ candidate characteristics, only about a dozen are selected, and the estimated effects are strongly nonlinear.
Our analysis differs from the aforementioned studies by explicitly accounting for unobserved heterogeneity through latent interactive fixed effects while simultaneously allowing for polynomial transformations and interaction terms in a high-dimensional panel-data framework.


We find that firm size \textit{mve} is by far the most influential predictor, its scaled estimated coefficient ($-0.0140$) being more than four times as large in absolute value as the second largest one. This is in line with the findings for the linear specification considered in \cite{Mruecker2025CCE}, where \textit{mve} likewise has the largest scaled coefficient in absolute value, and strengthens this result considerably: even in a very flexible specification that allows for quadratic and cubic effects of every characteristic as well as for all $4005$ pairwise interactions, \textit{mve} survives as a purely linear and additive effect. Moreover, its negative sign is consistent with the classical size premium.
The next largest scaled coefficients are the interaction
$\textit{disp} \times \textit{sfe}$ ($-0.00337$), the quadratic term $\textit{rsup}^2$ ($0.00330$), the interaction
$\textit{betasq} \times \textit{ms}$ ($0.00237$), the cubic term
$\textit{std\_turn}^3$ ($-0.00200$) and the interaction
$\textit{baspread} \times \textit{indmom}$ ($-0.00197$). These five terms illustrate the two channels through which our specification departs from the linear model: genuine curvature in a single characteristic ($\textit{rsup}^2$, $\textit{std\_turn}^3$) and effect heterogeneity via interactions ($\textit{disp} \times \textit{sfe}$,
$\textit{betasq} \times \textit{ms}$, $\textit{baspread} \times \textit{indmom}$).


%% file: appendixA.tex
\begin{center}
{\LARGE \textbf{Technical Appendix}}
\end{center}
\def\theequation{A.\arabic{equation}}
\setcounter{equation}{0}

\noindent Throughout the appendix, we let $c$, $C$, $c'$ and $C'$ denote generic positive constants that may take a different value on each occurrence. The symbols $c_j$ and $C_j$ with subscript $j$ (which may be either a natural number or a letter) are specific constants that are defined in the course of the appendix. Unless stated differently, the constants $c$, $C$, $c'$, $C'$, $c_j$ and $C_j$ depend neither on the dimensions $n$, $T$, $p$ nor on the sparsity index $s$. To emphasize that they do not depend on any of these parameters, we sometimes refer to them as absolute constants.

\section*{Auxiliary lemmas}\label{appendixA: auxilliary lemmas}

To begin with, we collect several auxiliary lemmas that will be used throughout the appendix. These lemmas are adaptations of results established in \cite{Mruecker2025CCE}. We provide proofs only where the technical arguments differ from those in \cite{Mruecker2025CCE}. For all other results whose proofs carry over essentially unchanged, we refer to the relevant parts of \cite{Mruecker2025CCE}.
We assume throughout that \ref{C:loadings}--\ref{C:id2GAMMATGAMMAEV} hold, together with \ref{C:nTp-large}--\ref{C:K-large} in the large-$T$ case and \ref{C:nTp-small}--\ref{C:K-small} in the small-$T$ case. Moreover, we pick the threshold parameter $\tau$ of the estimator $\widehat{K}$ as specified in Theorem \ref{theo:main}, i.e., such that $\tau = o(p)$ and $\{p \sqrt{\log(pT)}/\sqrt{n}\}  / \tau = o(1)$. 
All lemmas are formulated such that they hold for both the large-$T$ and the small-$T$ case.

\begin{lemmaA}\label{lemma:generictruncationarguments}
Let $\{V_{i,tjkl}\}$ be a collection of random variables with $1\leq i \leq n$, $1\leq t\leq T$, $1\leq j \leq p$, $1\leq k \leq K$ and $1\leq l\leq L$, where $p$ and $T$ are allowed to depend on $n$ but $K$ and $L$ are fixed. Assume there exists a constant $C>0$ with $\ex[|V_{i,tjkl}|^\nu]\leq C$ for some $\nu > 8$. 
\begin{enumerate}[label = (\roman*), leftmargin=0.975cm]

\item \label{lemma:generictruncationarguments1} It holds that 
\[\max_{i,t,j,k,l}| V_{i,tjkl} | = O_p\left( (npT)^{1/\nu}\right).\]
    
\item \label{lemma:generictruncationarguments2} Let the random variables $V_{i,tjkl}$ be independent across $i$ and have mean zero. Moreover, suppose that $(npT)^{\frac{1}{\nu - \xi}} = o(\sqrt{n} / \sqrt{\log(pT)} )$ for some small $\xi>0$. Then there exists a constant $C>0$ such that 
\[\pr\left(\max_{t,j,k,l}\Big|\frac{1}{n}\sum_{i=1}^n V_{i,tjkl} \Big| \geq C\sqrt{ \frac{\log(pT)}{n}}\right) = o(1).\]

\end{enumerate}
\end{lemmaA}

\begin{proof}
Statement \ref{lemma:generictruncationarguments1} is a simple consequence of the union bound and Markov's inequality: 
\begin{align*}
    \pr\left( \max_{i,t,j,k,l}| V_{i,tjkl} | \geq C' (npT)^{1/\nu} \right) &\leq \sum_{i,t,j,k,l} \pr\left( | V_{i,tjkl} | \geq C' (npT)^{1/\nu} \right) \\
    &\leq \sum_{i,t,j,k,l} \frac{\ex[|V_{i,tjkl}|^\nu] }{(C')^\nu(npT)} \leq \frac{K L C}{(C')^\nu},
\end{align*}
where we can choose $C'$ as large as desired to make the right-hand side arbitrarily small.
Statement \ref{lemma:generictruncationarguments2} can be proven by virtually the same arguments as the assertions of Lemma A.1 in \cite{Mruecker2025CCE}. 
\end{proof}

The next two lemmas restate Lemmas A.3/A'.3 and A.5(i)/A'.5(i) from \cite{Mruecker2025CCE} and can be proven exactly as laid out there.

\begin{lemmaA}\label{lemma:bargammagammanorm}
Let $\overline{\bs{\Gamma}} = n^{-1} \sum_{i=1}^n \bs{\Gamma}_i$ be the cross-sectional average of the loading matrices $\bs{\Gamma}_i$ and $\bs{\Gamma} = \ex[\bs{\Gamma}_i]$. It holds that 
\begin{enumerate}[label=(\roman*),leftmargin=0.975cm,topsep=0.5cm]
\item \label{lemma:bargammagammanorm:1} $\displaystyle{\big\| \overline{\bs{\Gamma}} - \bs{\Gamma} \big\| = O_p \Big( \sqrt{\frac{p \log p}{n}} \Big)}$.
\item \label{lemma:bargammagammanorm:2} $\displaystyle{\big\| \overline{\bs{\Gamma}} \overline{\bs{\Gamma}}^\top - \ex \, \overline{\bs{\Gamma}}  \overline{\bs{\Gamma}}^\top \big\| = O_p \Big( p \sqrt{\frac{\log p}{n}} \Big)}$.
\end{enumerate}
\end{lemmaA}

\begin{lemmaA}\label{lemmaA:barZ} 
The cross-sectional average $\overline{\bs{Z}} = n^{-1} \sum_{i=1}^n \bs{Z}_i$ of the idiosyncratic regressor components $\bs{Z}_i$ has the property that 
\[ \|\overline{\bs{Z}}\| = O_p\left(\sqrt{\frac{pT\log(pT)}{n}} \right). \]
\end{lemmaA}

We next use Lemmas \ref{lemma:bargammagammanorm} and \ref{lemmaA:barZ} to characterize the spectral distance between the matrices $\widehat{\bs{\Sigma}} = T^{-1} \bs{\overline{X}}^\top \bs{\overline{X}} = T^{-1} \sum_{t=1}^T \overline{X}_t \overline{X}_t^\top$ and $\bs{\Sigma} = \bs{\Gamma} \bs{\Gamma}^\top$.

\begin{lemmaA}\label{lemma1:eigenstructure}
It holds that 
\[  \| \widehat{\bs{\Sigma}} - \bs{\Sigma} \| = O_p\Big( p \sqrt{\frac{\log (pT)}{n}} \Big). \]
\end{lemmaA}

\begin{proof}
It suffices to prove that 
\begin{enumerate}[label=(\roman*),leftmargin=0.975cm,topsep=0.5cm]
\item $\displaystyle{ \| \widehat{\bs{\Sigma}} - \overline{\bs{\Sigma}} \| = O_p\Big( p \sqrt{\frac{\log (pT)}{n}} \Big) }$. 
\item $\displaystyle{ \| \overline{\bs{\Sigma}} - \bs{\Sigma} \| = O\Big(  \frac{p}{n} \Big) }$. 
\end{enumerate}
where $\overline{\bs{\Sigma}} = \ex [T^{-1} \bs{\overline{X}}^\top \bs{\overline{X}}] = \ex [T^{-1} \sum_{t=1}^T \overline{X}_t \overline{X}_t^\top]$.

We first verify the bound (i) on $\| \widehat{\bs{\Sigma}} - \overline{\bs{\Sigma}} \|$. Since 
\begin{align*} 
\| \widehat{\bs{\Sigma}} - \overline{\bs{\Sigma}} \|
 & \le \Big\| \frac{1}{T} \sum_{t=1}^T \big\{ \overline{\bs{\Gamma}} F_t F_t^\top \overline{\bs{\Gamma}}^\top - \ex \overline{\bs{\Gamma}} F_t F_t^\top \overline{\bs{\Gamma}}^\top \big\} \Big\| \\
 & \quad + 2 \Big\| \overline{\bs{\Gamma}} \Big\{ \frac{1}{T} \sum_{t=1}^T F_t \overline{Z}_t^\top \Big\} \Big\| + \Big\| \frac{1}{T} \sum_{t=1}^T \big( \overline{Z}_t \overline{Z}_t^\top - \ex \overline{Z}_t \overline{Z}_t^\top \big) \Big\|,
\end{align*} 
it suffices to bound the three terms on the right-hand side. By Lemma \ref{lemma:bargammagammanorm} \ref{lemma:bargammagammanorm:2}, 
\begin{align*} 
 & \Big\| \frac{1}{T} \sum_{t=1}^T \big\{ \overline{\bs{\Gamma}} F_t F_t^\top \overline{\bs{\Gamma}}^\top - \ex \overline{\bs{\Gamma}} F_t F_t^\top \overline{\bs{\Gamma}}^\top \big\} \Big\|   =   \Big\| \overline{\bs{\Gamma}} \overline{\bs{\Gamma}}^\top - \ex \overline{\bs{\Gamma}} \overline{\bs{\Gamma}}^\top \Big\| = O_p \Big(  p \sqrt{\frac{\log p}{n}} \Big).
\end{align*}
Moreover, as $\| \bs{\Gamma} \| = O(\sqrt{p})$ by \ref{C:id2GAMMATGAMMAEV} and $\| \overline{\bs{\Gamma}} - \bs{\Gamma} \|= O_p(\sqrt{p \log p / n})$ by Lemma \ref{lemma:bargammagammanorm}\ref{lemma:bargammagammanorm:1}, we obtain that $\| \overline{\bs{\Gamma}} \| \le \| \overline{\bs{\Gamma}} - \bs{\Gamma} \| + \| \bs{\Gamma} \| = O_p(\sqrt{p})$. From this and Lemma \ref{lemmaA:barZ}, it follows that
 \begin{align*} 
\Big\| \overline{\bs{\Gamma}} \Big\{ \frac{1}{T} \sum_{t=1}^T F_t \overline{Z}_t^\top \Big\} \Big\| 
 & \le \big\| \overline{\bs{\Gamma}} \big \| \Big\| \frac{1}{T} \sum_{t=1}^T F_t \overline{Z}_t^\top \Big\| \\
 & \le \frac{1}{T} \big\| \overline{\bs{\Gamma}} \big \| \big\| \bs{F} \big\| \big\|\bs{\overline{Z}} \big\|  
 =O_p\Big( p \sqrt{\frac{\log(pT)}{n}}\Big).
\end{align*}
Finally, since $|\ex ( n^{-1} \sum_{i=1}^n Z_{it,j} ) ( n^{-1} \sum_{i=1}^n Z_{it,j^\prime} )| \le C/n$ and $\max_{j,t} | n^{-1} \sum_{i=1}^n Z_{it,j} | = O_p(\sqrt{\log(pT)/n})$ by Lemma \ref{lemma:generictruncationarguments}\ref{lemma:generictruncationarguments2}, we obtain that 
\begin{align*}
\Big\| & \frac{1}{T} \sum_{t=1}^T \big( \overline{Z}_t \overline{Z}_t^\top - \ex \overline{Z}_t \overline{Z}_t^\top \big) \Big\|
 \le p \Big\| \frac{1}{T} \sum_{t=1}^T \big( \overline{Z}_t \overline{Z}_t^\top - \ex \overline{Z}_t \overline{Z}_t^\top \big) \Big\|_{\max} \\
 & = p \max_{j,j^\prime} \Big| \frac{1}{T} \sum_{t=1}^T \bigg\{ \Big( \frac{1}{n} \sum_{i=1}^n Z_{it,j} \Big) \Big( \frac{1}{n} \sum_{i=1}^n Z_{it,j^\prime} \Big) - \ex \Big( \frac{1}{n} \sum_{i=1}^n Z_{it,j} \Big) \Big( \frac{1}{n} \sum_{i=1}^n Z_{it,j^\prime} \Big) \bigg\} \Big| \\
 & \le p \bigg\{ \Big( \max_{j,t} \Big| \frac{1}{n} \sum_{i=1}^n Z_{it,j} \Big| \Big)^2 + \frac{C}{n} \bigg\} = O_p\Big( \frac{p \log(pT)}{n} \Big). 
\end{align*} 
Putting everything together, we arrive at the bound $\| \widehat{\bs{\Sigma}} - \overline{\bs{\Sigma}} \| = O_p( p \sqrt{\log (pT)/n})$, which is statement (i). Note that this bound differs slightly from the one derived in Proposition A.1 of \cite{Mruecker2025CCE}. The difference stems from the fact that, unlike there, we do not impose any weak dependence conditions over time on the model variables in the large-$T$ case. The bound (ii) on $\| \overline{\bs{\Sigma}} - \bs{\Sigma} \|$, on the other hand, is the same as in Proposition A.1 of \cite{Mruecker2025CCE} and can be proven exactly as laid out there.  
\end{proof}

\pagebreak
Now let $\psi_1 \ge \ldots \ge \psi_p$ and $\widehat{\psi}_1 \ge \ldots \ge \widehat{\psi}_p$ be the eigenvalues of $\bs{\Sigma}$ and $\widehat{\bs{\Sigma}}$, respectively. By Lemma \ref{lemma1:eigenstructure} and analogous arguments as in \cite{Mruecker2025CCE} (see in particular the proofs of Propositions A.2--A.5 therein), we can infer the following: 
\begin{lemmaA}\label{lemma2:eigenstructure} 
There exists an absolute constant $c > 0$ such that 
\begin{enumerate}[label=(\roman*),leftmargin=0.975cm]
\item $\eig_k \ge c p$ for all $k \le K$ and $\eig_k = 0$ for all $k > K$; 
\item with probability tending to $1$, $\widehat{\eig}_k \ge c p$ for all $k \le K$ and $\widehat{\eig}_k = O_p( p \sqrt{\log (pT)} / \linebreak \sqrt{n}) = o_p(p)$ for all $k > K$. 
\end{enumerate}
\end{lemmaA} 
\noindent A direct consequence of this lemma is that $\widehat{K} \convp K$, or put differently, that $\pr(\widehat{K} = K) \to 1$.

We next have a closer look at the projection matrix $\widehat{\bs{\Pi}} = \bs{I} - \widehat{\bs{W}} (\widehat{\bs{W}}^\top \widehat{\bs{W}})^{-} \widehat{\bs{W}}^\top$. From Lemma \ref{lemma2:eigenstructure} and the fact that $\widehat{K} \convp K$, it immediately follows that the matrix  
\[ \frac{\widehat{\bs{W}}^\top \widehat{\bs{W}}}{T} = \widehat{\bs{U}}^\top \Big( \frac{\overline{\bs{X}}^\top \overline{\bs{X}}}{T} \Big) \widehat{\bs{U}} = \widehat{\bs{U}}^\top \widehat{\bs{\Sigma}} \widehat{\bs{U}} = \text{diag}(\widehat{\eig}_1,\ldots,\widehat{\eig}_{\widehat{K}}) =: \widehat{\bs{\Eig}} \]
is invertible with probability tending to $1$. Hence, we can replace the generalized inverse $(\widehat{\bs{W}}^\top \widehat{\bs{W}})^{-}$ in the definition of $\widehat{\bs{\Pi}}$ by the proper inverse and write $\widehat{\bs{\Pi}} = \bs{I} - \widehat{\bs{W}} (\widehat{\bs{W}}^\top \widehat{\bs{W}})^{-1} \widehat{\bs{W}}^\top$ with probability tending to $1$. For simplicity of exposition, we often use this formulation without the specifier ``with probability tending to $1$'' in what follows. Our analysis of $\widehat{\bs{\Pi}}$ relies on the following decomposition (which holds with probability tending to $1$):
\begin{align*}
\widehat{\bs{\Pi}} 
 & = \bs{I} - \widehat{\bs{W}} (\widehat{\bs{W}}^\top \widehat{\bs{W}})^{-1} \widehat{\bs{W}}^\top \\
 & = \bigg\{ \bs{I} - \frac{1}{T} (\bs{F} \overline{\bs{\Gamma}}^\top \widehat{\bs{U}}) \Big[ \frac{1}{T} (\bs{F} \overline{\bs{\Gamma}}^\top \widehat{\bs{U}})^\top (\bs{F} \overline{\bs{\Gamma}}^\top \widehat{\bs{U}}) \Big]^{-1} (\bs{F} \overline{\bs{\Gamma}}^\top \widehat{\bs{U}})^\top \bigg\} - \widehat{\bs{R}},
\end{align*}
where 
\begin{align*} 
\widehat{\bs{R}} 
 & = \frac{1}{T} (\bs{F} \overline{\bs{\Gamma}}^\top \widehat{\bs{U}}) \bigg\{ \widehat{\bs{\Eig}}^{-1} - \Big[ \frac{1}{T} (\bs{F} \overline{\bs{\Gamma}}^\top \widehat{\bs{U}})^\top (\bs{F} \overline{\bs{\Gamma}}^\top \widehat{\bs{U}}) \Big]^{-1} \bigg\} (\bs{F} \overline{\bs{\Gamma}}^\top \widehat{\bs{U}})^\top \\
 & \quad + \frac{1}{T} (\bs{F} \overline{\bs{\Gamma}}^\top \widehat{\bs{U}}) \widehat{\bs{\Eig}}^{-1} (\overline{\bs{Z}} \widehat{\bs{U}})^\top 
 + \frac{1}{T} (\overline{\bs{Z}} \widehat{\bs{U}}) \widehat{\bs{\Eig}}^{-1} (\bs{F} \overline{\bs{\Gamma}}^\top \widehat{\bs{U}})^\top 
 + \frac{1}{T} (\overline{\bs{Z}} \widehat{\bs{U}}) \widehat{\bs{\Eig}}^{-1} (\overline{\bs{Z}} \widehat{\bs{U}})^\top 
\end{align*} 
and $\widehat{\bs{\Eig}}$ has already been defined above. Using this decomposition, we can prove the following result.

\begin{lemmaA}\label{lemma:projectionmatrixddecomposable}
With probability tending to $1$, $\widehat{\bs{\Pi}} = \bs{\Pi} - \widehat{\bs{R}}$.
\end{lemmaA}

\begin{proof}
The proof is virtually identical to that of Proposition A.6 in \cite{Mruecker2025CCE}.
\end{proof}

Importantly, the remainder $\widehat{\bs{R}}$ asymptotically vanishes in the following sense:

\begin{lemmaA}\label{lemma:NormOfRemainderR}
It holds that
\[ \|\widehat{\bs{R}}\| = O_p\left(\sqrt{\frac{\log(pT)}{n} }\right). \]
\end{lemmaA}

\begin{proof}
From the definition of $\widehat{\bs{R}}$, it immediately follows that 
\begin{align}
\big\|\widehat{\bs{R}}\big\| 
 & \leq \frac{1}{T} \Big\|(\bs{F} \overline{\bs{\Gamma}}^\top \widehat{\bs{U}}) \bigg\{ \widehat{\bs{\Eig}}^{-1} - \Big[ \frac{1}{T} (\bs{F} \overline{\bs{\Gamma}}^\top \widehat{\bs{U}})^\top (\bs{F} \overline{\bs{\Gamma}}^\top \widehat{\bs{U}}) \Big]^{-1} \bigg\} (\bs{F} \overline{\bs{\Gamma}}^\top \widehat{\bs{U}})^\top \Big\| \nonumber \\
 & \quad + \frac{1}{T} \big\|(\bs{F} \overline{\bs{\Gamma}}^\top \widehat{\bs{U}}) \widehat{\bs{\Eig}}^{-1} (\overline{\bs{Z}} \widehat{\bs{U}})^\top \big\|  + \frac{1}{T} \big\|(\overline{\bs{Z}} \widehat{\bs{U}}) \widehat{\bs{\Eig}}^{-1} (\bs{F} \overline{\bs{\Gamma}}^\top \widehat{\bs{U}})^\top \big\| \nonumber \\
 & \quad + \frac{1}{T} \big\|(\overline{\bs{Z}} \widehat{\bs{U}}) \widehat{\bs{\Eig}}^{-1} (\overline{\bs{Z}} \widehat{\bs{U}})^\top \big\| \nonumber \\
 &\leq \frac{1}{T} \big\|\bs{F} \overline{\bs{\Gamma}}^\top \big\| \Big\|\widehat{\bs{\Eig}}^{-1} - \Big[ \frac{1}{T} (\bs{F} \overline{\bs{\Gamma}}^\top \widehat{\bs{U}})^\top (\bs{F} \overline{\bs{\Gamma}}^\top \widehat{\bs{U}}) \Big]^{-1} \Big\| \big\|\bs{F} \overline{\bs{\Gamma}}^\top \big\| \nonumber \\
 & \quad + \frac{1}{T} \big\|\bs{F} \overline{\bs{\Gamma}}^\top \big\| \big\| \widehat{\bs{\Eig}}^{-1} \big\| \big\|\overline{\bs{Z}} \big\| + \frac{1}{T} \big\|\overline{\bs{Z}} \big\| \big\|\widehat{\bs{\Eig}}^{-1} \big\| \big\|\bs{F} \overline{\bs{\Gamma}}^\top \big\| + \frac{1}{T} \big\|\overline{\bs{Z}}\big\|^2 \big\|\widehat{\bs{\Eig}}^{-1} \big\| \nonumber \\
 &\leq \big\|\overline{\bs{\Gamma}} \big\|^2 \Big\|\widehat{\bs{\Eig}}^{-1} - \Big[ \frac{1}{T} (\bs{F} \overline{\bs{\Gamma}}^\top \widehat{\bs{U}})^\top (\bs{F} \overline{\bs{\Gamma}}^\top \widehat{\bs{U}}) \Big]^{-1} \Big\| \nonumber \\
 & \quad + \frac{1}{\sqrt{T}} \big\| \overline{\bs{\Gamma}} \big\| \big\| \widehat{\bs{\Eig}}^{-1} \big\| \big\|\overline{\bs{Z}} \big\| + \frac{1}{\sqrt{T}} \big\|\overline{\bs{Z}} \big\| \big\|\widehat{\bs{\Eig}}^{-1} \big\| \big\| \overline{\bs{\Gamma}} \big\| + \frac{1}{T} \big\|\overline{\bs{Z}}\|^2 \|\widehat{\bs{\Eig}}^{-1} \big\|, \label{lemma:NormOfRemainderR:bound}
\end{align}
where we have used that $\| \widehat{\bs{U}} \| = 1$ and $\| \bs{F} \| = \sqrt{T}$. Below, we show that 
\begin{align}
\| \overline{\bs{\Gamma}} \| & = O_p(\sqrt{p}) \label{lemma:NormOfRemainderR:statement0} \\[0.2cm]
\big\lVert \widehat{\bs{\Psi}}^{-1} \big\rVert & = O_p\bigg(\frac{1}{p}\bigg) \label{lemma:NormOfRemainderR:statement1} \\
\Big\| \widehat{\bs{\Eig}}^{-1} - \Big[ \frac{1}{T} (\bs{F} \overline{\bs{\Gamma}}^\top \widehat{\bs{U}})^\top (\bs{F} \overline{\bs{\Gamma}}^\top \widehat{\bs{U}}) \Big]^{-1} \Big\| & = O_p\bigg( \frac{1}{p} \sqrt{\frac{\log(pT)}{n}} \bigg). \label{lemma:NormOfRemainderR:statement2}
\end{align}
Inserting these bounds together with the bound on $\|\overline{\bs{Z}}\|$ from Lemma \ref{lemmaA:barZ} into \eqref{lemma:NormOfRemainderR:bound} completes the proof.  
\end{proof}

\begin{proof}[Proof of \eqref{lemma:NormOfRemainderR:statement0}]
Since $\| \overline{\bs{\Gamma}} \| \le \| \overline{\bs{\Gamma}} - \bs{\Gamma} \| + \| \bs{\Gamma} \|$, \eqref{lemma:NormOfRemainderR:statement0} follows from Lemma \ref{lemma:bargammagammanorm}\ref{lemma:bargammagammanorm:1} and the fact that $\| \bs{\Gamma} \| = O(\sqrt{p})$ by assumption \ref{C:id2GAMMATGAMMAEV}.
\end{proof}

\begin{proof}[Proof of \eqref{lemma:NormOfRemainderR:statement1}]
Since $\widehat{\bs{\Eig}} = \text{diag}(\widehat{\eig}_1,\ldots,\widehat{\eig}_{\widehat{K}})$, it holds that 
$\| \widehat{\bs{\Psi}}^{-1} \| = \lambda_{\textnormal{max}}(\widehat{\bs{\Psi}}^{-1}) = 1/\widehat{\psi}_{\widehat{K}} = O_p(p^{-1})$, where the final equality is a direct consequence of Lemma \ref{lemma2:eigenstructure} and the fact that $\pr(\widehat{K} = K) \to 1$. This shows \eqref{lemma:NormOfRemainderR:statement1}. 
\end{proof}

\begin{proof}[Proof of \eqref{lemma:NormOfRemainderR:statement2}]
By definition, 
\[ \widehat{\bs{\Eig}} = \frac{\widehat{\bs{W}}^\top \widehat{\bs{W}}}{T} = \frac{1}{T} (\bs{F} \overline{\bs{\Gamma}}^\top \widehat{\bs{U}} + \overline{\bs{Z}} \widehat{\bs{U}})^\top (\bs{F} \overline{\bs{\Gamma}}^\top \widehat{\bs{U}} + \overline{\bs{Z}} \widehat{\bs{U}}). \]
Hence, 
\begin{align*} 
\Big\| \widehat{\bs{\Eig}} - \frac{1}{T} (\bs{F} \overline{\bs{\Gamma}}^\top \widehat{\bs{U}})^\top (\bs{F} \overline{\bs{\Gamma}}^\top \widehat{\bs{U}}) \Big\| 
 & \le 2 \Big\| \frac{1}{T} (\bs{F} \overline{\bs{\Gamma}}^\top \widehat{\bs{U}})^\top (\overline{\bs{Z}} \widehat{\bs{U}}) \Big\| + \Big\| \frac{1}{T} (\overline{\bs{Z}} \widehat{\bs{U}})^\top (\overline{\bs{Z}} \widehat{\bs{U}}) \Big\| \\
 & \le 2 \big\| \overline{\bs{\Gamma}} \big\| \Big\| \frac{\bs{F}^\top \overline{\bs{Z}}}{T} \Big\| + \Big\| \frac{\overline{\bs{Z}}^\top \overline{\bs{Z}}}{T} \Big\|
  \\
 & \le \frac{2}{T} \big\| \overline{\bs{\Gamma}} \big\| \big\| \bs{F} \big\| \big\|\overline{\bs{Z}}\big\| +  \frac{1}{T}\big\|\overline{\bs{Z}} \big\|^2,  
\end{align*} 
where we have used that $\| \widehat{\bs{U}} \| = 1$. Since $\| \bs F \| = \sqrt{T}$ and $\| \overline{\bs{\Gamma}} \| = O_p(\sqrt{p})$ by \eqref{lemma:NormOfRemainderR:statement0}, we can apply Lemma \ref{lemmaA:barZ} to obtain that 
\begin{equation}\label{lemma:bargammagammanorm:proof1}
\Big\| \widehat{\bs{\Eig}} - \frac{1}{T} (\bs{F} \overline{\bs{\Gamma}}^\top \widehat{\bs{U}})^\top (\bs{F} \overline{\bs{\Gamma}}^\top \widehat{\bs{U}}) \Big\| = O_p\bigg( p \sqrt{\frac{\log(pT)}{n}}  \bigg).
\end{equation}
We next make use of the following bound for invertible matrices $\bs{A}$ and $\bs{B}$: Since $\bs{A}^{-1} - \bs{B}^{-1} = (\bs{A}^{-1} - \bs{B}^{-1} + \bs{B}^{-1}) (\bs{B} - \bs{A}) \bs{B}^{-1}$, it holds that  $\| \bs{A}^{-1} - \bs{B}^{-1} \| \le ( \|\bs{A}^{-1} - \bs{B}^{-1} \| + \| \bs{B}^{-1} \|) \| \bs{B} - \bs{A} \| \| \bs{B}^{-1} \|$ and thus
\[ \| \bs{A}^{-1} - \bs{B}^{-1} \| \le \frac{\| \bs{B}^{-1} \|^2 \|\bs{B} - \bs{A}\|}{1 - \| \bs{B}^{-1} \| \| \bs{B} - \bs{A} \|}, \] 
provided that $\| \bs{B}^{-1} \| \| \bs{B} - \bs{A} \| < 1$. With this bound, we obtain that 
\begin{align*}
\Big\| \widehat{\bs{\Eig}}^{-1} - \Big[ \frac{1}{T} (\bs{F} \overline{\bs{\Gamma}}^\top \widehat{\bs{U}})^\top (\bs{F} \overline{\bs{\Gamma}}^\top \widehat{\bs{U}}) \Big]^{-1} \Big\| 
 & \le \frac{\| \widehat{\bs{\Eig}}^{-1} \|^2 \|\widehat{\bs{\Eig}} - \frac{1}{T} (\bs{F} \overline{\bs{\Gamma}}^\top \widehat{\bs{U}})^\top (\bs{F} \overline{\bs{\Gamma}}^\top \widehat{\bs{U}}) \|}{1 - \| \widehat{\bs{\Eig}}^{-1} \| \|\widehat{\bs{\Eig}} - \frac{1}{T} (\bs{F} \overline{\bs{\Gamma}}^\top \widehat{\bs{U}})^\top (\bs{F} \overline{\bs{\Gamma}}^\top \widehat{\bs{U}}) \|}.
\end{align*} 
Using this together with \eqref{lemma:NormOfRemainderR:statement1} and \eqref{lemma:bargammagammanorm:proof1}, we can conclude that 
\[ \Big\| \widehat{\bs{\Eig}}^{-1} - \Big[ \frac{1}{T} (\bs{F} \overline{\bs{\Gamma}}^\top \widehat{\bs{U}})^\top (\bs{F} \overline{\bs{\Gamma}}^\top \widehat{\bs{U}}) \Big]^{-1} \Big\| = O_p\bigg( \frac{1}{p} \sqrt{\frac{\log(pT)}{n}} \bigg). \qedhere \]
\end{proof}

According to Lemma \ref{lemma:projectionmatrixddecomposable},
$\| \widehat{\bs{\Pi}} \bs{F} \| = \| (\bs{\Pi} - \widehat{\bs{R}}) \bs{F} \| = \| \widehat{\bs{R}} \bs{F} \|$
with probability tending to $1$. Moreover, since $\| \bs{F} \| = \sqrt{T}$, 
\[  \| \widehat{\bs{R}} \bs{F} \| \le  \| \widehat{\bs{R}} \| \| \bs{F} \| = O_p\left(\sqrt{\frac{\log(pT) T}{n} }\right) \]
by Lemma \ref{lemma:NormOfRemainderR}. This yields the final result of this section.
\begin{lemmaA}\label{lemma:PihatF}
It holds that 
\[ \| \widehat{\bs{\Pi}} \bs{F} \|  = O_p\left(\sqrt{\frac{\log(pT) T}{n} }\right). \]    
\end{lemmaA}

%% file: appendixB.tex
\section*{Proof of Theorem \ref{theo:main}}

Let $\bs{\Phi}_{i(j\ell)}$ be the $(j,\ell)$-th column of $\bs{\Phi}_i$ with $1 \le j \le p$ and $1 \le \ell \le L_j$. Moreover, let $\widetilde{\bs{\Phi}}_{i(j\ell)} = \bs{\Phi}_{i(j\ell)} - \overline{\bs{\Phi}}_{(j\ell)}$ be the empirically centred version of $\bs{\Phi}_{i(j\ell)}$ and analogously set $\widetilde{\gamma}_i = \gamma_i - \overline{\gamma}$ along with $\widetilde{\varepsilon}_i = \varepsilon_i - \overline{\varepsilon}$. Define the event 
\[ \mathcal{T}_\lambda = \left\{\frac{4}{nT}\max_{\substack{1 \le j \le p \\ 1 \le \ell \le L_j}} \Big|\sum_{i=1}^n( \widehat{\bs{\Pi}} \widetilde{\bs{\Phi}}_{i(j\ell)})^\top \widehat{\bs{\Pi}}(\bs{F} \widetilde{\gamma}_i + \widetilde{\varepsilon}_i ) \Big| \leq \lambda\right\} \]
and let $\mathcal{T}_{\text{RE}}$ denote the event that the design matrix 
\[\widehat{\bs{\Phi}} =(\widehat{\bs{\Phi}}_1^\top, \ldots, \widehat{\bs{\Phi}}_n^\top)^\top \quad \textnormal{with} \quad \widehat{\bs{\Phi}}_i =  \widehat{\bs{\Pi}} \widetilde{\bs{\Phi}}_i \] 
fulfills the restricted eigenvalue condition $\textnormal{RE}(S,{\varphi^2})$. We first show that the HD-CCE estimator is well-behaved on the event $\mathcal{T}_\lambda \cap \mathcal{T}_{\text{RE}}$.
\begin{propA}\label{proposition:ACCEBound} 
On the event $\mathcal{T}_\lambda \cap \mathcal{T}_{\textnormal{RE}}$, the HD-CCE estimator $\widehat{\beta}_\lambda$ satisfies the bound 
\[ \sum_{j=1}^p\sum_{\ell=1}^{L_j}  |\beta_{j\ell}- \widehat{\beta}_ {\lambda,j\ell}| \leq \frac{4 \lambda s }{\varphi^2}. \]
\end{propA}

\begin{proof}[Proof of Proposition \ref{proposition:ACCEBound}]
The proposition follows from standard finite-sample theory for the lasso. We provide a short proof for completeness. 
The basic inequality of the lasso reads
\begin{align*}
\frac{1}{nT} \sum_{i=1}^n \big\| \widehat{Y}_i - \widehat{\bs{\Phi}}_i \widehat{\beta}_\lambda \big\|^2 + \pen \sum_{j=1}^p \sum_{\ell=1}^{L_j}|\widehat{\beta}_{\lambda,j\ell}| 
 & \leq \frac{1}{nT} \sum_{i=1}^n \big\| \widehat{Y}_i - \widehat{\bs{\Phi}}_i \beta \big\|^2 + \pen \sum_{j=1}^p\sum_{\ell=1}^{L_j} |\beta_{j\ell}| \\
 & = \frac{1}{nT} \sum_{i=1}^n \big\| \widehat{\bs{\Pi}} (\bs{F}\widetilde{\gamma}_i + \widetilde{\varepsilon}_i)  \big\|^2 + \pen \sum_{j=1}^p \sum_{\ell=1}^{L_j}|\beta_{j\ell}|.
\end{align*}
Rearranging this inequality yields 
\begin{align*}
& \frac{1}{nT} \sum_{i=1}^n \big\|  \widehat{\bs{\Phi}}_i (\beta- \widehat{\beta}_\lambda)  \big\|^2   + \pen\sum_{j=1}^p \sum_{\ell=1}^{L_j}|\widehat{\beta}_{\lambda,j\ell}| \\*
& \leq - \frac{2}{nT} \sum_{i=1}^n (\widehat{\bs{\Phi}}_i (\beta- \widehat{\beta}_\lambda)) ^\top \widehat{\bs{\Pi}}(\bs{F}\widetilde{\gamma}_i  + \widetilde{\varepsilon}_i) + \pen  \sum_{j=1}^p \sum_{\ell=1}^{L_j}|\beta_{j\ell}| \\
& = -\frac{2}{nT} \sum_{i=1}^n \left(\sum_{j=1}^p\sum_{\ell=1}^{L_j} \widehat{\bs{\Pi}} \widetilde{\bs{\Phi}}_{i(j\ell)} (\beta_{j\ell}- \widehat{\beta}_ {\lambda,j\ell})\right)^\top\widehat{\bs{\Pi}}(\bs{F} \widetilde{\gamma}_i  +\widetilde{\varepsilon}_i)  + \pen  \sum_{j=1}^p \sum_{\ell=1}^{L_j}|\beta_{j\ell}| \\
& \leq \frac{2}{nT} \bigg\{ \max_{j,\ell} \bigg|\sum_{i=1}^n( \widehat{\bs{\Pi}} \widetilde{\bs{\Phi}}_{i(j\ell)})^\top \widehat{\bs{\Pi}}(\bs{F} \widetilde{\gamma}_i + \widetilde{\varepsilon}_i )\bigg| \bigg\} \sum_{j=1}^p \sum_{\ell=1}^{L_j}   |\beta_{j\ell}- \widehat{\beta}_ {\lambda,j\ell}|  + \pen  \sum_{j=1}^p \sum_{\ell=1}^{L_j}|\beta_{j\ell}|.
\end{align*}
Thus, on the event $\mathcal{T}_\lambda$, 
\begin{align*}
\frac{2}{nT} \sum_{i=1}^n \big\|  \widehat{\bs{\Phi}}_i (\beta- \widehat{\beta}_\lambda)  \big\|^2  
 &\leq 3\lambda \sum_{(j,\ell) \in S}  |\beta_{j\ell}- \widehat{\beta}_ {\lambda,j\ell}| - \lambda \sum_{(j,\ell) \in S^c}  | \widehat{\beta}_ {\lambda,j\ell}|,
\end{align*}
implying that the lasso's error vector $\beta - \widehat{\beta}_\lambda$ satisfies the constraint $\sum_{(j,\ell) \in S^c}  | \beta_{j\ell} - \widehat{\beta}_ {\lambda,j\ell}| \le 3 \sum_{(j,\ell) \in S}  |\beta_{j\ell}- \widehat{\beta}_ {\lambda,j\ell}|$. Consequently, on the event $\mathcal{T}_\lambda \cap \mathcal{T}_{\text{RE}}$, 
\begin{align*}
 & \frac{2}{nT} \sum_{i=1}^n \big\|  \widehat{\bs{\Phi}}_i (\beta- \widehat{\beta}_\lambda)  \big\|^2 + \lambda \sum_{j=1}^p \sum_{\ell=1}^{L_j}  |\beta_{j\ell}- \widehat{\beta}_{\lambda,j\ell}| \\*
 & \leq 3\lambda \sum_{(j,\ell) \in S}  |\beta_{j\ell}- \widehat{\beta}_{\lambda,j\ell}| - \lambda \sum_{(j,\ell) \in S^c}  | \widehat{\beta}_{\lambda,j\ell}| + \lambda \sum_{j=1}^p\sum_{\ell=1}^{L_j}  |\beta_{j\ell}- \widehat{\beta}_{\lambda,j\ell}| \\
 & \leq 4\lambda \sum_{(j,\ell) \in S}  |\beta_{j\ell}- \widehat{\beta}_ {\lambda,j\ell}| 
  \leq \frac{4\lambda \sqrt{s} }{\sqrt{{\varphi^2}}  } \sqrt{\frac{1}{nT} \sum_{i=1}^n \|\widehat{\bs{\Phi}}_i (\beta - \widehat{\beta}_{\lambda})\|^2} \\
 & \leq \frac{4 \lambda^2 s }{{\varphi^2} } + \frac{1}{nT} \sum_{i=1}^n \|\widehat{\bs{\Phi}}_i (\beta- \widehat{\beta}_ {\lambda})\|^2,
\end{align*}
where the last inequality uses $4ab \le 4a^2 + b^2$. From this, it immediately follows that 
\[ \sum_{j=1}^p\sum_{\ell=1}^{L_j}  |\beta_{j\ell}- \widehat{\beta}_ {\lambda,j\ell}| \leq \frac{4 \lambda s }{\varphi^2} \]
on the event $\mathcal{T}_\lambda \cap \mathcal{T}_{\text{RE}}$.
\end{proof}

\noindent The next result shows that the event $\mathcal{T}_\lambda$ occurs with probability tending to $1$, provided that the penalty constant $\lambda$ is chosen as $\lambda = h_{n} \nu_{n,T}$ with
\[ \nu_{n,T} = \sqrt{\frac{\log(pT)}{n}} (npT)^{2/\theta}\]
and $\{h_{n}\}$ some slowly diverging sequence, e.g., $h_{n} =\log \log(n)$.

\begin{propA}\label{prop:convergenceofsetSl}
Under \ref{C:loadings}--\ref{C:id2GAMMATGAMMAEV}, it holds that
\[ \frac{1}{nT}\max_{j,\ell} \left|\sum_{i=1}^n( \widehat{\bs{\Pi}} \widetilde{\bs{\Phi}}_{i(j\ell)})^\top \widehat{\bs{\Pi}}(\bs{F} \widetilde{\gamma}_i + \widetilde{\varepsilon}_i )\right| = O_p(\nu_{n,T}). \]
\end{propA}

\begin{proof}[Proof of Proposition \ref{prop:convergenceofsetSl}]
It suffices to prove the following two statements: 
\begin{align}
\frac{1}{nT}  \max_{j,\ell} \left|\sum_{i=1}^n( \widehat{\bs{\Pi}} \widetilde{\bs{\Phi}}_{i(j\ell)})^\top \widehat{\bs{\Pi}}\bs{F} \widetilde{\gamma}_i \right| & =  O_p\left(\sqrt{\frac{\log(pT)}{n}} (npT)^{2/\theta} \right)  \label{proof:prop:SLTERM1} \\  
\frac{1}{nT}\max_{j,\ell} \left|\sum_{i=1}^n( \widehat{\bs{\Pi}} \widetilde{\bs{\Phi}}_{i(j\ell)})^\top \widehat{\bs{\Pi}}\widetilde{\varepsilon}_i \right| & = O_p\left(\sqrt{\frac{\log(pT)}{n}} (npT)^{2/\theta} \right).  \label{proof:prop:SLTERM3}
\end{align}
For the proof, we make use of the notation $\widetilde{\bs{\Phi}}_{i(j\ell)} = (\widetilde{\phi}_{j\ell}(X_{i1,j}),\ldots,\widetilde{\phi}_{j\ell}(X_{iT,j}))^\top$ with $\widetilde{\phi}_{j\ell}(X_{it,j}) = \phi_{j\ell}(X_{it,j}) - n^{-1} \sum_{i=1}^n \phi_{j\ell}(X_{it,j})$ and $\widetilde{\varepsilon}_i = (\widetilde{\varepsilon}_{i1},\ldots, \widetilde{\varepsilon}_{iT})^\top$ with $\widetilde{\varepsilon}_{it} = \varepsilon_{it} - n^{-1} \sum_{i=1}^n \varepsilon_{it}$.

We start with the proof of \eqref{proof:prop:SLTERM1}. With the help of Lemma \ref{lemma:generictruncationarguments}, we can show that $\max_{i,t,j,\ell} |  \widetilde{\phi}_{j\ell}(X_{it, j})| = O_p((npT)^{1/\theta})$ and $\max_i \| \widetilde{\gamma}_i \| = O_p(n^{1/\theta})$. Using this together with Lemma \ref{lemma:PihatF} immediately yields that 
\begin{align*}
   \frac{1}{nT}\max_{j,\ell} \left|\sum_{i=1}^n( \widehat{\bs{\Pi}} \widetilde{\bs{\Phi}}_{i(j\ell)})^\top \widehat{\bs{\Pi}}\bs{F} \widetilde{\gamma}_i \right| &\leq \|\widehat{\bs{\Pi}}\bs{F}\| \frac{1}{nT}\max_{j,\ell} \sum_{i=1}^n \| \widetilde{\bs{\Phi}}_{i(j\ell)}\| \| \widetilde{\gamma}_i \|\\
   &\leq \|\widehat{\bs{\Pi}}\bs{F}\| \frac{1}{\sqrt{T}}\max_{i,t,j,\ell} |  \widetilde{\phi}_{j\ell}(X_{it, j})| \max_i \| \widetilde{\gamma}_i \| \\
   &= O_p\left(\sqrt{\frac{\log(pT)}{n}} (npT)^{2/\theta} \right),
\end{align*}
which is statement \eqref{proof:prop:SLTERM1}.

We next turn to the proof of \eqref{proof:prop:SLTERM3}. Since $\widehat{\bs{\Pi}} = \bs{\Pi} - \widehat{\bs{R}} = \bs{I} - T^{-1} \bs{F} \bs{F}^\top - \widehat{\bs{R}}$ with probability tending to $1$ by Lemma \ref{lemma:projectionmatrixddecomposable}, it holds that 
\begin{align}
\frac{1}{nT}\max_{j,\ell} \bigg|\sum_{i=1}^n( \widehat{\bs{\Pi}} \widetilde{\bs{\Phi}}_{i(j\ell)})^\top \widehat{\bs{\Pi}}\widetilde{\varepsilon}_i \bigg| 
 & \le \frac{1}{nT}\max_{j,\ell} \bigg|\sum_{i=1}^n \widetilde{\bs{\Phi}}_{i(j\ell)}^\top \widetilde{\varepsilon}_i \bigg| \nonumber \\*
 & \quad + \frac{1}{nT}\max_{j,\ell} \bigg|\sum_{i=1}^n \widetilde{\bs{\Phi}}_{i(j\ell)}^\top (T^{-1} \bs{F} \bs{F}^\top) \widetilde{\varepsilon}_i \bigg| \nonumber \\
 & \quad + \frac{1}{nT}\max_{j,\ell} \bigg|\sum_{i=1}^n \widetilde{\bs{\Phi}}_{i(j\ell)}^\top \widehat{\bs{R}} \widetilde{\varepsilon}_i \bigg| \label{proof:prop:SLTERM3:step1}
\end{align}
with probability tending to $1$. From Lemma \ref{lemma:generictruncationarguments}, it follows that 
\begin{align}
\frac{1}{nT}\max_{j,\ell} \bigg|\sum_{i=1}^n \widetilde{\bs{\Phi}}_{i(j\ell)}^\top \widetilde{\varepsilon}_i \bigg| 
 & \le \max_{j,\ell,t} \bigg| \frac{1}{n} \sum_{i=1}^n \phi_{j\ell}(X_{it,j})\varepsilon_{it} \bigg| \nonumber \\
 & \quad + \max_{j,\ell,t} \bigg| \Big\{ \frac{1}{n} \sum_{i'=1}^n\phi_{j\ell}(X_{i't,j}) \Big\} \Big\{\frac{1}{n}\sum_{i'=1}^n\varepsilon_{i't}\Big\} \bigg| \nonumber \\
 & = O_p\left(\sqrt{\frac{\log(pT)}{n}}\right). \label{proof:prop:SLTERM3:step2} 
 \end{align}
Moreover, Lemma \ref{lemma:generictruncationarguments} and \ref{C:factorssummable} yield that
\begin{align}
 & \frac{1}{nT}\max_{j,\ell} \bigg|\sum_{i=1}^n \widetilde{\bs{\Phi}}_{i(j\ell)}^\top (T^{-1} \bs{F} \bs{F}^\top) \widetilde{\varepsilon}_i \bigg| \nonumber \\
 & = \max_{j,\ell}\frac{1}{nT^2}\bigg|\sum_{k=1}^K\sum_{t,t'=1}^T \sum_{i=1}^n \widetilde{\phi}_{j\ell}(X_{it,j}) F_{tk} F_{t'k}\widetilde{\varepsilon}_{it'} \bigg| \nonumber \\
 & \le K \max_k \bigg\{ \frac{1}{T^2}\sum_{t,t'=1}^T| F_{tk}||F_{t'k}| \bigg\} \max_{j,\ell,t,t'} \bigg| \frac{1}{n}  \sum_{i=1}^n \widetilde{\phi}_{j\ell}(X_{it,j}) \widetilde{\varepsilon}_{it'} \bigg| \nonumber \\
 & = O_p\left(\sqrt{\frac{\log(pT)}{n}}\right). \label{proof:prop:SLTERM3:step3}
\end{align}
Finally, using again that $\max_{i,j,\ell,t}|\widetilde{\phi}_{j\ell}(X_{it,j})| = O_p((npT)^{1/\theta})$ and $\max_{i,t}|\widetilde{\varepsilon}_{it}| = O_p((nT)^{1/\theta})$ by Lemma \ref{lemma:generictruncationarguments} and applying Lemma \ref{lemma:NormOfRemainderR}, we obtain that
\begin{align}
\frac{1}{nT}\max_{j,\ell} \bigg|\sum_{i=1}^n \widetilde{\bs{\Phi}}_{i(j\ell)}^\top \widehat{\bs{R}} \widetilde{\varepsilon}_i \bigg|    & \le \frac{1}{T} \big\|\widehat{\bs{R}}\big\| \max_{i,j,\ell} \big\|\widetilde{ \bs{\Phi}}_{i(j\ell)} \big\| \max_{i} \big\|\widetilde{\varepsilon}_i\big\| \nonumber \\
 & \le \big\|\widehat{\bs{R}}\big\| \max_{i,j,\ell,t} \big|\widetilde{\phi}_{j\ell}(X_{it,j}) \big| \max_{i,t} |\widetilde{\varepsilon}_{it}| \nonumber \\
 & = O_p\left(\sqrt{\frac{\log(pT)}{n}} (npT)^{2/\theta} \right). \label{proof:prop:SLTERM3:step4}
\end{align}
Plugging \eqref{proof:prop:SLTERM3:step2}--\eqref{proof:prop:SLTERM3:step4} into the bound \eqref{proof:prop:SLTERM3:step1} yields statement \eqref{proof:prop:SLTERM3}.
\end{proof}

\noindent Theorem \ref{theo:main} follows immediately by combining Propositions \ref{proposition:ACCEBound} and \ref{prop:convergenceofsetSl}.

%% file: appendixC.tex
\section*{Proof of Theorem \ref{prop:compatibility}}

The RE condition RE$(S,\varphi^2)$ can be formulated as follows: 
\[ \bigg(\sum_{(j,\ell)\in S}| v_{j\ell} |\bigg)^2 \le \frac{\| \widehat{\bs{\Phi}} v \|^2}{nT} \frac{|S|}{{\varphi^2}} \qquad \text{for all } v \in \mathcal{C}(S,3), \] 
where 
$\mathcal{C}(S,3) = \{ v: \sum_{(j,\ell)\in S^c} | v_{j\ell} | \leq  3 \sum_{(j,\ell)\in S}| v_{j\ell} | \}$. 
The aim is to verify that this condition holds with probability tending to $1$, provided that the population-level RE condition \eqref{eq:CC-theoretical} is satisfied.

To prove this, we bound the term $\| \widehat{\bs{\Phi}} v \|^2 / (nT)$ from below. As before, we use the notation $\widehat{\bs{\Phi}}_i = \widehat{\bs{\Pi}} \widetilde{\bs{\Phi}}_i$ with $\widetilde{\bs{\Phi}}_i = \bs{\Phi}_i - n^{-1} \sum_{i'=1}^n \bs{\Phi}_{i'}$. Moreover, we let $\bs{\Phi}_{i(j\ell)} = (\phi_{j\ell}(X_{i1,j}),\ldots,\phi_{j\ell}(X_{iT,j}))^\top$ be the $(j,\ell)$-th column of $\bs{\Phi}_i$ and define $\widetilde{\bs{\Phi}}_{i(j\ell)}$ analogously. 
For any fixed $v=(v_1,\ldots,v_p)^\top \in \mathcal{C}(S,3)$, the following holds with probability tending to $1$: 
\begin{align*}
\frac{\| \widehat{\bs{\Phi}} v \|^2}{nT} 
 & = \frac{1}{nT} \sum_{i=1}^n \|\widehat{\bs{\Pi}} \widetilde{\bs{\Phi}}_i v\|^2 \\
 & \geq \frac{1}{nT} \sum_{i=1}^n \|\bs{\Pi} \widetilde{\bs{\Phi}}_i v\|^2 - \frac{2}{nT  } \sum_{i=1}^n  (\widehat{\bs{R}} \widetilde{\bs{\Phi}}_i  v)^\top (\bs{\Pi} \widetilde{\bs{\Phi}}_i v) \\
 & = \frac{1}{nT} \sum_{i=1}^n \bigg\| \bs{\Pi} \bigg(\bs{\Phi}_i - \frac{1}{n}\sum_{i'=1}^n \bs{\Phi}_{i'}\bigg)v \bigg\|^2 - \frac{2}{nT  } \sum_{i=1}^n (\widehat{\bs{R}} \widetilde{\bs{\Phi}}_i  v)^\top (\bs{\Pi} \widetilde{\bs{\Phi}}_i v) \\
 & \geq \frac{1}{nT} \sum_{i=1}^n \big\|\bs{\Pi} (\bs{\Phi}_i - \ex[\bs{\Phi}_i])v\big\|^2 \\
 & \quad + \frac{2}{nT} \sum_{i=1}^n  \big\{ \bs{\Pi} (\bs{\Phi}_i -  \ex[\bs{\Phi}_i])v \big\}^\top \bigg\{ \bigg( \ex[\bs{\Phi}_i] - \frac{1}{n}\sum_{i'=1}^n \bs{\Phi}_{i'} \bigg) v \bigg\} \\
 & \quad - \frac{2}{nT} \sum_{i=1}^n  (\widehat{\bs{R}} \widetilde{\bs{\Phi}}_i  v)^\top (\bs{\Pi} \widetilde{\bs{\Phi}}_i v) \\
 & \geq - \frac{1}{nT} \left |\sum_{i=1}^n \bigg\{ \big\|\bs{\Pi} (\bs{\Phi}_i -  \ex[\bs{\Phi}_i])v \big\|^2 - \ex\big[\big\|\bs{\Pi} (\bs{\Phi}_i -  \ex[\bs{\Phi}_i])v\big\|^2 \big] \bigg\} \right| \\
 & \quad + \frac{1}{nT} \sum_{i=1}^n \ex\big[ \big\|\bs{\Pi} (\bs{\Phi}_i -  \ex[\bs{\Phi}_i])v\big\|^2 \big] \\
 & \quad + \frac{2}{nT} \sum_{i=1}^n  \big\{ \bs{\Pi} (\bs{\Phi}_i -  \ex[\bs{\Phi}_i])v \big\}^\top \bigg\{ \bigg( \ex[\bs{\Phi}_i] - \frac{1}{n}\sum_{i'=1}^n \bs{\Phi}_{i'} \bigg) v \bigg\} \\
 & \quad- \frac{2}{nT  } \sum_{i=1}^n  (\widehat{\bs{R}} \widetilde{\bs{\Phi}}_i  v)^\top (\bs{\Pi} \widetilde{\bs{\Phi}}_i v),
\end{align*} 
where the first inequality is by Lemma \ref{lemma:projectionmatrixddecomposable}. In the sequel, we show the following three statements:
\begin{align}
\frac{1}{nT} & \Bigg|\sum_{i=1}^n \bigg\{ \big\|\bs{\Pi} (\bs{\Phi}_i - \ex[\bs{\Phi}_i])v \big\|^2 - \ex\big[\big\|\bs{\Pi} (\bs{\Phi}_i -  \ex[\bs{\Phi}_i])v\big\|^2 \big] \bigg\} \Bigg| \nonumber \\ & = O_p\left(\sqrt{\frac{\log(pT)}{n}} \bigg(\sum_{(j,\ell)\in S}|v_{j\ell}| \bigg)^2\right) \label{eq:comp-cond-R1} \\
\frac{1}{nT} & \sum_{i=1}^n  \big\{ \bs{\Pi} (\bs{\Phi}_i - \ex[\bs{\Phi}_i])v \big\}^\top \bigg\{ \bigg( \ex[\bs{\Phi}_i] - \frac{1}{n}\sum_{i'=1}^n \bs{\Phi}_{i'} \bigg) v \bigg\} \nonumber \\ & = O_p\left(\sqrt{\frac{\log(pT)}{n}} (npT)^{1/\theta} \bigg(\sum_{(j,\ell)\in S}|v_{j\ell}| \bigg)^2\right) \label{eq:comp-cond-R2} \\
\frac{1}{nT} & \sum_{i=1}^n (\widehat{\bs{R}} \widetilde{\bs{\Phi}}_i  v)^\top (\bs{\Pi} \widetilde{\bs{\Phi}}_i v) \nonumber \\ & =  O_p\left( \sqrt{\frac{\log(pT)}{n}}(npT)^{2/\theta} \bigg(\sum_{(j,\ell)\in S}|v_{j\ell}| \bigg)^2\right). \label{eq:comp-cond-R3}
\end{align}    
Using these three statements together with the population-level RE condition \eqref{eq:CC-theoretical}, we obtain that for any $v \in \mathcal{C}(S,3)$, 
\begin{align*}
\frac{\| \widehat{\bs{\Phi}} v \|^2}{nT} 
 & \ge \frac{1}{nT} \sum_{i=1}^n \ex\big[ \big\|\bs{\Pi} (\bs{\Phi}_i -  \ex[\bs{\Phi}_i])v\big\|^2 \big] \\*
 & \quad + O_p\left( \sqrt{\frac{\log(pT)}{n}}(npT)^{2/\theta} \right) \bigg(\sum_{(j,\ell)\in S}|v_{j\ell}| \bigg)^2 \\
 & \ge c^{(1)} \sum_{j=1}^p \sum_{\ell=1}^{L_j} |v_{j\ell}|^2 \\*
 & \quad + O_p\left( \sqrt{\frac{\log(pT)}{n}}(npT)^{2/\theta} \right) \bigg(\sum_{(j,\ell)\in S}|v_{j\ell}| \bigg)^2 \\
 & \ge \Bigg\{ \frac{c^{(1)}}{s} +  O_p\left( \sqrt{\frac{\log(pT)}{n}}(npT)^{2/\theta} \right) \Bigg\} \bigg(\sum_{(j,\ell)\in S}|v_{j\ell}| \bigg)^2, 
\end{align*}
where the last inequality is due to the fact that $\sum_{j=1}^p \sum_{\ell=1}^{L_j} |v_{j\ell}|^2 \ge \sum_{(j,\ell) \in S} |v_{j\ell}|^2 \ge  (\sum_{(j,\ell) \in S} |v_{j\ell}|)^2/s$. Since $\sqrt{\log(pT)/n}(npT)^{2/\theta} = o(1/s)$ by \ref{C:s-large} in the large-$T$ and \ref{C:s-small} in the small-$T$ case, respectively, we can conclude that 
\[ \frac{\| \widehat{\bs{\Phi}} v \|^2}{nT} \ge \frac{c^{(1)}}{2s} \bigg(\sum_{(j,\ell)\in S}|v_{j\ell}| \bigg)^2 \]
with probability tending to $1$. Put differently, the condition RE$(S,\varphi^2)$ with $\varphi^2 = c^{(1)}/2$ is satisfied with probability tending to $1$.

\begin{proof}[Proof of \eqref{eq:comp-cond-R1}]
It holds that
\begin{align*}
 & \frac{1}{nT} \left |\sum_{i=1}^n \Big\{ \|\bs{\Pi} (\bs{\Phi}_i -  \ex[\bs{\Phi}_i])v\|^2 - \ex\left[\|\bs{\Pi} (\bs{\Phi}_i -  \ex[\bs{\Phi}_i])v\|^2 \right] \Big\} \right| \\
 & = \Bigg|v^\top \frac{1}{nT}\sum_{i=1}^n \Big\{ (\bs{\Phi}_i -  \ex[\bs{\Phi}_i])^\top \bs{\Pi} (\bs{\Phi}_i -  \ex[\bs{\Phi}_i]) - \ex\left[(\bs{\Phi}_i -  \ex[\bs{\Phi}_i])^\top \bs{\Pi} (\bs{\Phi}_i -  \ex[\bs{\Phi}_i]) \right] \Big\} v \Bigg| \\
 & \leq \sum_{j,j'=1}^p\sum_{\ell=1}^{L_j}\sum_{\ell'=1}^{L_{j'}} |v_{j\ell}||v_{j'\ell'}| \\*
 & \qquad \Bigg |\frac{1}{nT}\sum_{i=1}^n  \Big\{ \big(\phi_{j\ell}(X_{i(j)}) -  \ex[\phi_{j\ell}(X_{i(j)})]\big)^\top \bs{\Pi} \big(\phi_{j'\ell'}(X_{i(j')}) -  \ex[\phi_{j'\ell'}(X_{i(j')})]\big) \\* 
 & \qquad - \ex\left[ \big(\phi_{j\ell}(X_{i(j)}) -  \ex[\phi_{j\ell}(X_{i(j)})]\big)^\top \bs{\Pi} \big(\phi_{j'\ell'}(X_{i(j')}) -  \ex[\phi_{j'\ell'}(X_{i(j')})]\big) \right] \Big\} \Bigg| \\
 & \le Q_A + Q_B
\end{align*} 
with 
\begin{align*} 
Q_A  & = \sum_{j,j'=1}^p\sum_{\ell=1}^{L_j}\sum_{\ell'=1}^{L_{j'}} |v_{j\ell}||v_{j'\ell'}| \\*
     & \qquad \Bigg |\frac{1}{nT}\sum_{i=1}^n  \Big\{ \big(\phi_{j\ell}(X_{i(j)}) -  \ex[\phi_{j\ell}(X_{i(j)})]\big)^\top \big(\phi_{j'\ell'}(X_{i(j')}) - \ex[\phi_{j'\ell'}(X_{i(j')})]\big) \\* 
     & \qquad - \ex\left[ \big(\phi_{j\ell}(X_{i(j)}) - \ex[\phi_{j\ell}(X_{i(j)})]\big)^\top \big(\phi_{j'\ell'}(X_{i(j')}) -  \ex[\phi_{j'\ell'}(X_{i(j')})]\big) \right] \Big\} \Bigg| \\
Q_B  & = \sum_{j,j'=1}^p\sum_{\ell=1}^{L_j}\sum_{\ell'=1}^{L_{j'}} |v_{j\ell}||v_{j'\ell'}| \\*
     & \qquad \Bigg |\frac{1}{nT}\sum_{i=1}^n  \bigg\{ \big(\phi_{j\ell}(X_{i(j)}) -  \ex[\phi_{j\ell}(X_{i(j)})]\big)^\top \frac{\bs{F}\bs{F}^\top}{T} \big(\phi_{j'\ell'}(X_{i(j')}) - \ex[\phi_{j'\ell'}(X_{i(j')})]\big) \\* 
     & \qquad - \ex\Big[ \big(\phi_{j\ell}(X_{i(j)}) - \ex[\phi_{j\ell}(X_{i(j)})]\big)^\top \frac{\bs{F}\bs{F}^\top}{T} \big(\phi_{j'\ell'}(X_{i(j')}) -  \ex[\phi_{j'\ell'}(X_{i(j')})]\big) \Big] \bigg\} \Bigg|.
\end{align*}
Moreover, for any $v \in \mathcal{C}(S,3)$, 
\begin{align*}
Q_A 
 &\leq \bigg(4 \sum_{(j,\ell)\in S}|v_{j\ell}| \bigg)^2 \\*
 & \quad \ \Bigg\{\max_{t,j,j',\ell,\ell'} \bigg|\frac{1}{n}\sum_{i=1}^n \Big\{ \big(\phi_{j\ell}(X_{it,j}) -  \ex[\phi_{j\ell}(X_{it,j})]\big) \big(\phi_{j'\ell'}(X_{it,j'}) -  \ex[\phi_{j'\ell'}(X_{it,j'})]\big) \\ 
 & \qquad - \ex\Big[ \big(\phi_{j\ell}(X_{it,j}) -  \ex[\phi_{j\ell}(X_{it,j})]\big) \big(\phi_{j'\ell'}(X_{it,j'}) -  \ex[\phi_{j'\ell'}(X_{it,j'})]\big) \Big] \Big\} \bigg| \Bigg\} \\ 
 & = O_p\left(\sqrt{\frac{\log(pT)}{n}} \bigg(\sum_{(j,\ell)\in S}|v_{j\ell}| \bigg)^2\right),    
\end{align*}
where the maximum over $t,j,j',\ell,\ell'$ can be shown to be of order $\sqrt{\log(pT)/n}$ by using the same techniques as in the proof of Lemma \ref{lemma:generictruncationarguments}\ref{lemma:generictruncationarguments2}. Analogously, 
\begin{align*}
Q_B 
 & = \sum_{j,j'=1}^p\sum_{\ell=1}^{L_j}\sum_{\ell' =1}^{L_{j'}} |v_{j\ell}||v_{j'\ell'}| \ \Bigg|\frac{1}{nT^2} \sum_{k=1}^K \sum_{t,t'=1}^T F_{tk}F_{t'k} \\ 
 & \qquad \sum_{i=1}^n \Big\{ \big(\phi_{j\ell}(X_{it,j}) -  \ex[\phi_{j\ell}(X_{it,j})]\big)\big(\phi_{j'\ell'}(X_{it',j'}) - \ex[\phi_{j'\ell'}(X_{it',j'})]\big) \\ 
 & \phantom{\qquad \sum_{i=1}^n \Big\{} - \ex\Big[ \big(\phi_{j\ell}(X_{it,j}) -  \ex[\phi_{j\ell}(X_{it,j})]\big) \big(\phi_{j'\ell'}(X_{it',j'}) -  \ex[\phi_{j'\ell'}(X_{it',j'})] \big) \Big] \Big\} \Bigg| \\
 & \le K\bigg(\max_k\frac{1}{T}\sum_{t=1}^T |F_{tk}|\bigg)^2 \bigg(4 \sum_{(j,\ell)\in S}|v_{j\ell}| \bigg)^2 \\
 & \qquad \max_{t,t',j,j',\ell,\ell'}\Bigg |\frac{1}{n}\sum_{i=1}^n \Big\{ \big(\phi_{j\ell}(X_{it,j}) -  \ex[\phi_{j\ell}(X_{it,j})]\big)\big(\phi_{j'\ell'}(X_{it',j'}) - \ex[\phi_{j'\ell'}(X_{it',j'})]\big) \\* 
 & \phantom{\qquad \max_{t,t',j,j',\ell,\ell'}\Bigg |} - \ex\Big[ \big(\phi_{j\ell}(X_{it,j}) -  \ex[\phi_{j\ell}(X_{it,j})]\big) \big(\phi_{j'\ell'}(X_{it',j'}) -  \ex[\phi_{j'\ell'}(X_{it',j'})] \big) \Big] \Big\} \Bigg| \\
 & = O_p\left(\sqrt{\frac{\log(pT)}{n}} \bigg(\sum_{(j,\ell)\in S}|v_{j\ell}| \bigg)^2\right),
\end{align*}
where $\max_k\frac{1}{T}\sum_{t=1}^T |F_{tk}| \le C < \infty$ by \ref{C:factorssummable}. 
\end{proof}

\pagebreak
\begin{proof}[Proof of \eqref{eq:comp-cond-R2}]
By Lemma \ref{lemma:generictruncationarguments},
\[ \max_{i,j,\ell,t} \big| \phi_{j\ell}(X_{it,j}) -  \ex[\phi_{j\ell}(X_{it,j})]\big| =  O_p\left((npT)^{1/\theta} \right) \]
and
\begin{align*} 
 & \max_{i,j,\ell,t} \bigg| \frac{1}{n}\sum_{i'=1}^n\phi_{j\ell}(X_{i't,j}) - \ex[\phi_{j\ell}(X_{it,j})] \bigg| \\
 & = \max_{j,\ell,t} \bigg| \frac{1}{n}\sum_{i'=1}^n \big( \phi_{j\ell}(X_{i't,j}) - \ex[\phi_{j\ell}(X_{i't,j})] \big) \bigg| \\
 & = O_p\left(\sqrt{\frac{\log(pT)}{n}} \right). 
\end{align*}
Using this together with the fact that $\|\bs{\Pi}\|= 1$, we get that  
\begin{align*}
 & \frac{1}{nT} \sum_{i=1}^n \big\{ \bs{\Pi} (\bs{\Phi}_i -  \ex[\bs{\Phi}_i])v \big\}^\top \bigg\{ \bigg(\ex[\bs{\Phi}_i] - \frac{1}{n}\sum_{i'=1}^n \bs{\Phi}_{i'} \bigg) v\bigg\} \\
 & =\sum_{j,j'=1}^p\sum_{\ell=1}^{L_j}\sum_{\ell'=1}^{L_{j'}} \frac{1}{nT} \sum_{i=1}^n  v_{j\ell}v_{j'\ell'} \big(\bs{\Pi} (\bs{\Phi}_i -  \ex[\bs{\Phi}_i])\big)_{(j\ell)}^\top \bigg(\ex[\bs{\Phi}_i] - \frac{1}{n}\sum_{i'=1}^n \bs{\Phi}_{i'}\bigg)_{(j'\ell')} \\
 & \leq \sum_{j,j'=1}^p\sum_{\ell=1}^{L_j}\sum_{\ell'=1}^{L_{j'}} \frac{1}{nT} \sum_{i=1}^n  |v_{j\ell}| |v_{j'\ell'}| \big\| \big( \bs{\Pi} (\bs{\Phi}_i -  \ex[\bs{\Phi}_i]) \big)_{(j\ell)} \big\|  \Bigg\|\bigg(\ex[\bs{\Phi}_i] - \frac{1}{n}\sum_{i'=1}^n \bs{\Phi}_{i'}\bigg)_{(j'\ell')} \bigg\| \\
 & \leq \frac{1}{T}\max_{i,j,\ell} \big\| (\bs{\Phi}_i -  \ex[\bs{\Phi}_i])_{(j\ell)} \big\| \max_{i,j,\ell} \bigg\| \bigg(\ex[\bs{\Phi}_i] - \frac{1}{n}\sum_{i'=1}^n \bs{\Phi}_{i'}\bigg)_{(j\ell)} \bigg\|    \sum_{j,j'=1}^p\sum_{\ell=1}^{L_j}\sum_{\ell' =1}^{L_{j'}}   |v_{j\ell}| |v_{j'\ell'}| \\
 & \leq \max_{i,j,\ell,t} \big| \phi_{j\ell}(X_{it,j}) -  \ex[\phi_{j\ell}(X_{it,j})]\big|  \max_{i,j,\ell,t} \bigg| \ex[\phi_{j\ell}(X_{it,j})] - \frac{1}{n}\sum_{i'=1}^n\phi_{j\ell}(X_{i't,j}) \bigg|  \bigg(4\sum_{(j,\ell)\in S}|v_{j\ell}| \bigg)^2 \\
 & = O_p\left(\sqrt{\frac{\log(pT)}{n}} (npT)^{1/\theta}  \bigg(\sum_{(j,\ell)\in S}|v_{j\ell}| \bigg)^2\right)
\end{align*}
for any $v \in \mathcal{C}(S,3)$.
\end{proof}

\begin{proof}[Proof of \eqref{eq:comp-cond-R3}] For any $v \in \mathcal{C}(S,3)$, it holds that 
\begin{align*}
 & \frac{1}{nT} \sum_{i=1}^n  (\widehat{\bs{R}} \widetilde{\bs{\Phi}}_i  v)^\top (\bs{\Pi} \widetilde{\bs{\Phi}}_i v) \\
 & = \frac{1}{nT} \sum_{j,j'=1}^p\sum_{\ell=1}^{L_j}\sum_{\ell' =1}^{L_{j'}}\sum_{i=1}^n v_{j\ell} \widetilde{\phi}_{j\ell}(X_{i(j)})^\top \widehat{\bs{R}}^\top  \bs{\Pi} \widetilde{\phi}_{j'\ell'}(X_{i(j')}) v_{j'\ell'} \\
 & \leq \frac{1}{nT} \sum_{j,j'=1}^p\sum_{\ell=1}^{L_j}\sum_{\ell' =1}^{L_{j'}} \sum_{i=1}^n |v_{j\ell}| \|\widetilde{\phi}_{j\ell}(X_{i(j)})\| \| \widehat{\bs{R}} \| \| \widetilde{\phi}_{j'\ell'}(X_{i(j')})\| | v_{j'\ell'}| \\
 & \leq \frac{1}{T} \| \widehat{\bs{R}}\| \left(\max_{i,j,\ell}\|\widetilde{\phi}_{j\ell}(X_{i(j)})\|\right)^2 \bigg(4\sum_{(j,\ell)\in S}|v_{j\ell}| \bigg)^2 \\
 & = O_p\left( \sqrt{\frac{\log(pT)}{n}}(npT)^{2/\theta}  \bigg(\sum_{(j,\ell)\in S}|v_{j\ell}| \bigg)^2\right),
\end{align*}
where we have used that $\| \widehat{\bs{R}}\| = O_p(\sqrt{\log(pT) / n})$ by Lemma \ref{lemma:NormOfRemainderR} and 
\[\max_{i,j,\ell}\|\widetilde{\phi}_{j\ell}(X_{i(j)})\| = O_p\left(\sqrt{T}(npT)^{1/\theta} \right),\] 
by Lemma \ref{lemma:generictruncationarguments}\ref{lemma:generictruncationarguments1}.
\end{proof}

\section*{Proof of Theorem \ref{prop:compatibility-parametric}}

We begin by proving the following lemma. 
\begin{lemmaA}\label{lemma:PolynomialBlocks} 
Under the assumptions of Theorem \ref{prop:compatibility-parametric}, there exists a fixed constant $c^{(2)}$ (which is in particular independent of $n$, $T$ and $p$) such that
\[ \min_{1 \le t \le T} \bigg\{ \sum_{j=1}^p v_j^\top \bs{\Omega}_{jt} v_j \bigg\} \ge  c^{(2)} \sum_{j=1}^p \|v_j\|^2 \quad \text{for any vectors } v_j \in \reals^{L_j}, \]
where $\bs{\Omega}_{jt} = \ex[\{ \phi_j(X_{it,j}) - \ex[\phi_j(X_{it,j})] \} \{ \phi_j(X_{it,j}) - \ex[\phi_j(X_{it,j})] \}^\top]$ is the covariance matrix of the random vector $\phi_j(X_{it,j}) = (\phi_{j1}(X_{it,j}),\ldots,\phi_{jL_j}(X_{it,j}))^\top$.
\end{lemmaA}

\begin{proof}
Since $v_j^\top \bs{\Omega}_{jt} v_j = \widetilde{v}_{jt}^\top \widetilde{\bs{\Omega}}_{jt} \widetilde{v}_{jt}$ with 
\[ \widetilde{\bs{\Omega}}_{jt} = \ex\left[ \begin{pmatrix} 1 \\ \phi_j(X_{it,j}) \end{pmatrix} \begin{pmatrix} 1 & \phi_j(X_{it,j})^\top \end{pmatrix} \right] \quad \text{and} \quad \widetilde{v}_{jt} = \begin{pmatrix} -\ex[\phi_j(X_{it,j})^\top] v_j \\ v_j \end{pmatrix}, \]
it suffices to show that 
\begin{equation}\label{eq:statement:PolynomialBlocksLemma} 
\min_{1 \le t \le T} \bigg\{ \sum_{j=1}^p \widetilde{v}_{jt}^\top \widetilde{\bs{\Omega}}_{jt} \widetilde{v}_{jt} \bigg\} \ge c^{(2)} \sum_{j=1}^p \|v_j\|^2. 
\end{equation}

We first derive some properties of the matrices $\widetilde{\bs{\Omega}}_{jt}$. To do so, we use the notation $\mu_{tj}=\ex[X_{it,j}]$ and $\sigma^2_{tj} = \var(X_{it,j})$. Notably, both $\mu_{tj}$ and $\sigma^2_{tj}$ do not depend on $i$, since the variables $X_{1t,j}, \ldots,X_{nt,j}$ are identically distributed under our assumptions. We verify the following two facts:
\begin{enumerate}[label=(\roman*)]
\item \label{proof:gaussianpolynomialblocks:determinantclaim} There exist positive real numbers $c^{(1)}_{L_j}$ and $c^{(2)}_{L_j}$ which only depend on $L_j$ such that 
\begin{align*}
\det\big(\widetilde{\bs{\Omega}}_{jt}\big) & = c^{(1)}_{L_j} (\sigma_{tj}^2)^{c^{(2)}_{L_j}}.
\end{align*}
\item \label{proof:gaussianpolynomialblocks:traceclaim} It holds that
\[ \tr\big(\widetilde{\bs{\Omega}}_{jt}\big) = \sum_{\ell=0}^{L_j} \ex[X_{it,j}^{2\ell}]. \]
\end{enumerate}
The second claim \ref{proof:gaussianpolynomialblocks:traceclaim} is clear. To prove the first claim \ref{proof:gaussianpolynomialblocks:determinantclaim}, we proceed as follows. According to the binomial theorem,
\begin{align*}
\left(\frac{X_{it,j} - \mu_{tj}}{\sigma_{tj}}\right)^\ell 
 & = \sum_{k=0}^\ell a_{\ell,k}^{[tj]} X_{it,j}^k \quad \text{with} \quad a_{\ell,k}^{[tj]} = \binom{\ell}{k} \left(\frac{1}{\sigma_{tj}}\right)^k \left(\frac{-\mu_{tj}}{\sigma_{tj}}\right)^{\ell-k}
\end{align*} 
for $\ell=0,\ldots,{L_j}$. Introducing the transformation matrix
\[ \bs{A}^{[tj]} = 
\begin{pmatrix}
a_{0,0}^{[tj]} & 0 & 0 & \cdots & 0 \\
a_{1,0}^{[tj]} & a_{1,1}^{[tj]} & 0 & \cdots & 0 \\
a_{2,0}^{[tj]} & a_{2,1}^{[tj]} & a_{2,2}^{[tj]} & \cdots & 0 \\
\vdots & \vdots & \vdots & \ddots & \vdots \\
a_{L_j,0}^{[tj]} & a_{L_j,1}^{[tj]} & a_{L_j,2}^{[tj]} & \cdots & a_{L_j,L_j}^{[tj]},
\end{pmatrix}, \]
we can thus write
\[ \begin{pmatrix} 1 \\ \phi_j(X_{it,j}) \end{pmatrix}  
= \begin{pmatrix} 1 \\ X_{it,j} \\ \vdots \\ X_{it,j}^{L_j} \end{pmatrix}
= \big( \bs{A}^{[tj]} \big)^{-1} 
\begin{pmatrix} 1  \\ \big(\frac{X_{it,j} - \mu_{tj}}{\sigma_{tj}}\big) \\ \vdots \\ \big(\frac{X_{it,j} - \mu_{tj}}{\sigma_{tj}}\big)^{L_j} \end{pmatrix}.  \]
Note that the inverse of $\bs{A}^{[tj]}$ exists since $\det(\bs{A}^{[tj]}) = \prod_{\ell=0}^{L_j} a_{\ell,\ell}^{[tj]}$ with $a_{\ell,\ell}^{[tj]} = (1/\sigma_{tj})^\ell$ for all $\ell$ and $0 < c_\sigma \le \sigma_{tj}^2 \le C_\sigma < \infty$ under our conditions. It follows that  
\begin{align*}
\det\big(\widetilde{\bs{\Omega}}_{jt}\big) 
 & = \det\left(\ex\left[ \begin{pmatrix} 1 \\ \phi_j(X_{it,j}) \end{pmatrix} \begin{pmatrix} 1 & \phi_j(X_{it,j})^\top \end{pmatrix} \right]\right) \\
 & = \frac{1}{\det(\bs{A}^{[tj]})^2} \det\left(\ex\left[ \begin{pmatrix*}[c] 1 \\ \big(\frac{X_{it,j} - \mu_{tj}}{\sigma_{tj}}\big) \\\vdots \\ \big( \frac{X_{it,j} - \mu_{tj}}{\sigma_{tj}} \big)^{L_j} \end{pmatrix*} \begin{pmatrix*}[c] 1 & \big(\frac{X_{it,j} - \mu_{tj}}{\sigma_{tj}}\big) & \ldots & \big( \frac{X_{it,j} - \mu_{tj}}{\sigma_{tj}} \big)^{L_j} \end{pmatrix*} \right]\right) \\
 & = \frac{1}{\det(\bs{A}^{[tj]})^2} \cdot \det\left(\ex\left[ \begin{pmatrix*}[c] 1 \\ W \\ \vdots \\ W^{L_j} \end{pmatrix*} \begin{pmatrix*} 1 & W & \ldots & W^{L_j} \end{pmatrix*} \right] \right) 
\end{align*}
with $W \sim N(0,1)$. Since 
\[ c_{L_j}^{(1)} := \det\left(\ex\left[ \begin{pmatrix*}[c] 1 \\ W \\ \vdots \\ W^{L_j} \end{pmatrix*} \begin{pmatrix*} 1 & W & \ldots & W^{L_j} \end{pmatrix*} \right] \right) > 0, \]
which can be shown by simple calculations, and $\det(\bs{A}^{[tj]}) = \prod_{\ell=0}^{L_j} a_{\ell,\ell}^{[tj]} = (1/\sigma_{tj})^{c_{L_j}^{(2)}}$ with $c_{L_j}^{(2)} ={{L_j}({L_j}+1)/2}$, we finally arrive at 
\[ \det\big(\widetilde{\bs{\Omega}}_{jt}\big) = c_{L_j}^{(1)} (\sigma_{tj}^2)^{c_{L_j}^{(2)}}, \]
which completes the proof of \ref{proof:gaussianpolynomialblocks:determinantclaim}.

We next show that the minimal eigenvalue $\psi_{\min}(\widetilde{\bs{\Omega}}_{jt})$ of $\widetilde{\bs{\Omega}}_{jt}$ is bounded from below in the following sense: 
\begin{equation}\label{eq:minEV:PolynomialBlocksLemma}
\min_{\substack{1 \le j \le p, \\ 1 \le t \le T}} \psi_{\min}(\widetilde{\bs{\Omega}}_{jt}) \ge c_{\Omega} > 0 
\end{equation}
for some fixed constant $c_{\Omega}$ that is independent of $T$ and $p$. Since \eqref{eq:minEV:PolynomialBlocksLemma} immediately implies \eqref{eq:statement:PolynomialBlocksLemma}, this completes the proof of Lemma \ref{lemma:PolynomialBlocks}.

We prove \eqref{eq:minEV:PolynomialBlocksLemma} by contradiction: Suppose that $\min_{j,t} \psi_{\min}(\widetilde{\bs{\Omega}}_{jt})$ is not bounded from below in the sense of \eqref{eq:minEV:PolynomialBlocksLemma}. In particular, suppose w.l.o.g.\ that $\min_{j,t} \psi_{\min}(\widetilde{\bs{\Omega}}_{jt})$ converges to zero as $T$ and/or $p$ increases. 
In this situation, the maximal eigenvalue $\max_{j,t} \psi_{\max}(\widetilde{\bs{\Omega}}_{jt})$ has to diverge. 
Otherwise, $\min_{j,t} \det(\widetilde{\bs{\Omega}}_{jt})$ would converge to zero, which would contradict \ref{proof:gaussianpolynomialblocks:determinantclaim} according to which 
\[ \min_{j,t} \det\big(\widetilde{\bs{\Omega}}_{jt}\big) = \min_{j,t} c^{(1)}_{L_j} (\sigma_{tj}^2)^{c^{(2)}_{L_j}} \geq c > 0 \]
for some sufficiently small but fixed constant $c$, where we have used that $1 \le L_j \le L_{\max} < \infty$ for all $j$ and $0 < c_\sigma \le \sigma_{tj}^2 \le C_\sigma < \infty$ under our conditions. 
Since $\tr(\widetilde{\bs{\Omega}}_{jt})$ is the sum of the eigenvalues of $\widetilde{\bs{\Omega}}_{jt}$ and thus at least as large as the maximal eigenvalue $\psi_{\max}(\widetilde{\bs{\Omega}}_{jt})$, $\max_{j,t} \tr(\widetilde{\bs{\Omega}}_{jt})$ must diverge as well. 
On the other hand, by \ref{proof:gaussianpolynomialblocks:traceclaim}, $\tr(\widetilde{\bs{\Omega}}_{jt})$ is equal to 
\[ h(\mu_{tj}, \sigma_{tj}^2) := \sum_{\ell=0}^{L_j} \ex[X_{it,j}^{2\ell}], \]
where the function $h$ is continuous by the properties of Gaussian moments. Under our conditions, there exist constants $C_\mu$, $c_\sigma$ and $C_\sigma$ such that $|\mu_{tj}|\leq C_\mu < \infty$ and $0 < c_\sigma \leq \sigma_{tj}^2 \leq C_\sigma < \infty$. Since $\mathcal{M} := \{(x,y): |x|\leq C_\mu,\ c_\sigma \leq y \leq C_\sigma\}$ is a compact set, the function $h$ is bounded on $\mathcal{M}$. Consequently, $\max_{j,t} \tr(\widetilde{\bs{\Omega}}_{jt}) = \max_{j,t} h(\mu_{tj}, \sigma_{tj}^2)$ is bounded, which obviously contradicts the above conclusion that $\max_{j,t} \tr(\widetilde{\bs{\Omega}}_{jt})$ diverges. We have thus verified \eqref{eq:minEV:PolynomialBlocksLemma}. 
\end{proof}

\noindent We now turn to the proof of Theorem \ref{prop:compatibility-parametric}. For fixed (but arbitrary) vectors $v_j \in \reals^{L_j}$, define the functions $g_{t'j}^{(t)}$ by setting 
\[ g_{t'j}^{(t)}(W_{t'j})  = \Pi_{tt'} \big( \phi_j(X_{it',j}) - \ex[\phi_j(X_{it',j})]\big)^\top v_j, \] 
where $\Pi_{tt'}$ denotes the $(t,t')$-th entry of the projection matrix $\bs{\Pi}$ and 
\begin{equation*}
W_{t'j} = \frac{X_{it',j} - \mu_{t'j}}{\sigma_{t'j}} \sim \normal(0,1)
\end{equation*}
with $\mu_{tj}=\ex[X_{it,j}]$ and $\sigma^2_{tj} = \var(X_{it,j})$. For ease of notation, we replace the set of double indices $\{(t',j): 1 \le t' \le T, \, 1 \le j \le p\}$ by the simpler index set $\{h: 1 \le h \le Tp \}$ in what follows, which can be achieved, e.g., by stipulating the bijective mapping $(t',j) \mapsto h:= (t'-1)p + j$. We thus write $g_{h}^{(t)}(W_{h})$ instead of $g_{t'j}^{(t)}(W_{t'j})$. Next, let $\vartheta_{\nu}$ for ${\nu}\in \naturals$ be the standardized Hermite-Tschebycheff polynomials given by
\[\vartheta_{\nu}(x) = \frac{1}{\sqrt{{\nu}!}} (-1)^{\nu} e^{x^2/2} \frac{d^{\nu}}{dx^{\nu}} \big( e^{-x^2/2} \big), \] 
which form an orthonormal basis of ${L}^2(\reals, (2\pi)^{-1/2} e^{-x^2/2}dx)$ and satisfy the equations
\begin{align*}
\ex[\vartheta_\nu(W_{h})\vartheta_{\mu}(W_{h'})] = \left(\ex [ W_{h} W_{h'}]\right)^\nu\delta_{\nu,\mu}
\end{align*} for $h,h'=1,\ldots,Tp$ and $\nu,\mu \in \naturals$ with $\delta_{\nu,\mu}$ denoting the Kronecker delta \citep[see eq.\ (9) in][]{lancasterOrthogonalPolynomials}. Since $g_{h}^{(t)}(\cdot)\in {L}^2\big(\reals, (2\pi)^{-1/2} e^{-x^2/2}dx\big)$, it has the  Hermite expansion 
\[ g_{h}^{(t)}(W_{h}) = \sum_{\nu =1}^\infty d_{h,\nu}^{(t)} \vartheta_\nu(W_{h}), \]  
implying that $\ex[\{ g_{h}^{(t)}(W_{h}) \}^2] = \sum_{\nu=1}^\infty \{d_{h,\nu}^{(t)}\}^2$. With these preparations in place, we get
\begin{align}
 & \frac{1}{T} \ex\left[ \big\|\bs{\Pi} (\bs{\Phi}_i -  \ex[\bs{\Phi}_i])v \big\|^2\right] \nonumber \\
 & = \frac{1}{T} \sum_{t=1}^T \ex\left[ \bigg\{\sum_{j=1}^p\sum_{t'=1}^T\Pi_{tt'} \big(\phi_j(X_{it',j}) - \ex[\phi_j(X_{it',j})]\big)^\top v_j \bigg\}^2 \right] \nonumber \\
 & = \frac{1}{T} \sum_{t=1}^T \ex\left[ \bigg\{\sum_{h=1}^{Tp} g_h^{(t)}(W_{h})\bigg\}^2 \right] \nonumber \\
 & = \frac{1}{T}\sum_{t=1}^T \sum_{h=1}^{Tp}  \sum_{h'=1}^{Tp} \sum_{\nu =1}^\infty d_{h,\nu}^{(t)} d_{h',\nu}^{(t)} \ex\left[   \vartheta_\nu(W_{h}) \vartheta_\nu(W_{h'}) \right ] \nonumber \\
 & = \frac{1}{T}\sum_{t=1}^T\sum_{h=1}^{Tp} \sum_{h'=1}^{Tp} \sum_{\nu = 1}^\infty d_{h,\nu}^{(t)} d_{h',\nu}^{(t)} \left(\ex [  W_{h} W_{h'}]\right)^\nu. \label{eq:applyGUO}
\end{align}
To obtain a lower bound on the term in the final line above, we apply the following result, which is a special case of Lemma 2 in \cite{GUO20221037}:
\begin{lemmaA}
Let $W \sim \normal(0, \ex[W W^\top])$ be a Gaussian random vector in $\reals^q$ with $\ex[W_h^2] = 1$ for all $1 \le h \le q$. For all $m \ge 1$ and all $a \in \reals^q$ with $\|a\|=1$,
\[ \sum_{h,h'=1}^q \big\{ \ex[W_h W_{h'}] \big\}^m a_h a_{h'} \ge \psi_{\min} \big( \ex[W W^\top] \big), \]
where $\psi_{\min} \big( \ex[W W^\top] \big)$ is the minimal eigenvalue of $\ex[W W^\top]$.     
\end{lemmaA}
\noindent Applying this lemma to \eqref{eq:applyGUO} yields that 
\begin{align*}
 & \frac{1}{T} \ex\left[ \big\|\bs{\Pi} (\bs{\Phi}_i -  \ex[\bs{\Phi}_i])v \big\|^2\right] \\
 & \geq \frac{1}{T}\sum_{t=1}^T  \sum_{\nu = 1}^\infty \psi_{\textnormal{min}}\big(\ex[W W^\top]\big) \sum_{h=1}^{Tp} (d_{h,\nu}^{(t)})^2 \\
 & = \frac{1}{T}\sum_{t=1}^T \psi_{\textnormal{min}}\big(\ex[W W^\top]\big) \sum_{h=1}^{Tp} \ex\big[\{g_h^{(t)}(W_{h})\}^2\big] \\
 & = \psi_{\textnormal{min}}\big(\ex[W W^\top]\big) \frac{1}{T} \sum_{t=1}^T \sum_{t'=1}^T  \Pi_{tt'}^2 \sum_{j=1}^p \ex\left[ \big\{ \big(\phi_j(X_{it',j}) - \ex[\phi_j(X_{it',j})]\big)^\top v_j \big\}^2 \right] \\
 & = \psi_{\textnormal{min}}\big(\ex[W W^\top]\big) \frac{1}{T} \sum_{t=1}^T \sum_{t'=1}^T  \Pi_{tt'}^2 \sum_{j=1}^p v_j^\top \bs{\Omega}_{jt'} v_j. 
\end{align*}
Moreover, applying Lemma \ref{lemma:PolynomialBlocks} to the final line above, we obtain that 
\begin{align}
\frac{1}{T} \ex\left[ \big\|\bs{\Pi} (\bs{\Phi}_i -  \ex[\bs{\Phi}_i])v \big\|^2\right] 
 & \ge c^{(2)} \sum_{j=1}^p \|v_j\|^2  \psi_{\textnormal{min}}(\ex[W W^\top]) \frac{1}{T}\sum_{t=1}^T  \sum_{t'=1}^T \Pi_{tt'}^2 \nonumber \\
 & =  c^{(2)} \sum_{j=1}^p \|v_j\|^2  \psi_{\textnormal{min}}(\ex[W W^\top]) \frac{T-K}{T}, \label{eq:applyEVlowerbound} 
 \end{align}
where the last equality follows from the fact that $\sum_{t=1}^T \sum_{t'=1}^T \Pi_{tt'}^2 = \|\bs{\Pi}\|_F^2 \geq \sum_{t=1}^T \psi_t(\bs{\Pi}) = T-K$ with $\psi_t(\bs{\Pi})$ being the $t$-th eigenvalue of $\bs{\Pi}$. We next show that 
\begin{equation}\label{eq:EVmin-EWW}
\psi_{\textnormal{min}}(\ex[W W^\top]) \ge c_W > 0
\end{equation}
with a fixed constant $c_W$ (that does not depend on the dimension $T p$ of the random vector $W$). Plugging \eqref{eq:EVmin-EWW} into \eqref{eq:applyEVlowerbound}, we arrive at  
\[ \frac{1}{T} \ex\left[ \big\|\bs{\Pi} (\bs{\Phi}_i -  \ex[\bs{\Phi}_i])v \big\|^2\right] \ge c^{(2)} c_W \frac{T-K}{T} \bigg( \sum_{j=1}^p \sum_{\ell=1}^{L_j} |v_{j\ell}|^2 \bigg), \]
which immediately implies the statement of Theorem \ref{prop:compatibility-parametric}. To complete the proof, it thus remains to verify \eqref{eq:EVmin-EWW}: for any $w\in\reals^{Tp}$ and $\bs{D}_t=\textnormal{diag}(1/\sqrt{\var(X_{1t,1})}, \ldots$ $\ldots, 1/\sqrt{\var(X_{1t,p})})$,
\begin{align*}
w^\top \ex[W W^\top] w 
 & = \sum_{t,t'=1}^T \sum_{j,j'=1}^p w_{tj} \ex[W_{tj} W_{t'j'}] w_{t'j'} \\
 & = \sum_{t,t'=1}^T \sum_{j,j'=1}^p w_{tj} \frac{\cov(X_{it,j}, X_{it',j'})} {\sqrt{\var(X_{it,j}) \var(X_{it',j'})}} w_{t'j'} \\
 & = \sum_{t,t'=1}^T \sum_{j,j'=1}^p w_{tj} \frac{\cov(\Gamma_{i,j}F_t + Z_{it,j}, \Gamma_{i,j'}F_{t'} + Z_{it',j'})}{\sqrt{\var(X_{it,j}) \var(X_{it',j'})}} w_{t'j'} \\
 & = \sum_{t,t'=1}^T \sum_{j,j'=1}^p w_{tj} \frac{\cov(\Gamma_{i,j}F_t, \Gamma_{i,j'}F_{t'})}{\sqrt{\var(X_{it,j}) \var(X_{it',j'})}} w_{t'j'} \\
 & \quad + \sum_{t,t'=1}^T \sum_{j,j'=1}^p w_{tj} \frac{\cov( Z_{it,j}, Z_{it',j'})}{\sqrt{\var(X_{it,j}) \var(X_{it',j'})}} w_{t'j'} \\
 & \geq \sum_{t,t'=1}^T \sum_{j,j'=1}^p w_{tj} \frac{\cov( Z_{it,j}, Z_{it',j'})}{\sqrt{\var(X_{it,j}) \var(X_{it',j'})}} w_{t'j'} \\
 & = \sum_{t=1}^T \sum_{j,j'=1}^p w_{tj} \frac{\cov( Z_{it,j}, Z_{it,j'})}{\sqrt{\var(X_{it,j}) \var(X_{it,j'})}} w_{tj'} \\
 & = \sum_{t=1}^T  w_t^\top \bs{D}_t^\top \bs{\Sigma}_Z \bs{D}_t w_t 
 \geq \sum_{t=1}^T \psi_{\textnormal{min}}(\bs{\Sigma}_Z) \|\bs{D}_t w_t\|^2 \\
 & = \sum_{t=1}^T  \psi_{\textnormal{min}}(\bs{\Sigma}_Z) \sum_{j=1}^p \frac{w_{tj}^2}{\var(X_{it,j})} \geq c_W \|w\|^2
\end{align*}
with a sufficiently small constant $c_W > 0$, since $\psi_{\min}(\bs \Sigma_Z) \ge c_Z > 0$ and $0 < c_\sigma \le \var(X_{it,j}) \le C_\sigma < \infty$ under our conditions.